\documentclass{amsart}
\usepackage[utf8]{inputenc}

\usepackage{amsmath}
\usepackage{amsthm}
\usepackage[foot]{amsaddr}

\usepackage{graphicx}

\makeatletter
\providecommand{\leftsquigarrow}{%
	\mathrel{\mathpalette\reflect@squig\relax}%
}
\newcommand{\reflect@squig}[2]{%
	\reflectbox{$\m@th#1\rightsquigarrow$}%
}
\makeatother

\makeatletter
\providecommand{\leftsquigarrowstar}{%
	\mathrel{\mathpalette\reflect@squigs\relax}%
}
\newcommand{\reflect@squigs}[2]{%
	\reflectbox{$\m@th#1\rightsquigarrow^*$}%
}
\makeatother

\usepackage{virginialake}

\usepackage{hyperref}

\makeatletter
\newcommand{\oset}[3][0ex]{%
  \mathrel{\mathop{#3}\limits^{
    \vbox to#1{\kern-2\ex@
    \hbox{$\scriptscriptstyle#2$}\vss}}}}
\makeatother

\renewcommand{\phi}{\varphi}

\newtheorem{theorem}{Theorem}
\newtheorem{proposition}[theorem]{Proposition}
\newtheorem{lemma}[theorem]{Lemma}
\newtheorem{corollary}[theorem]{Corollary}
\newtheorem{example}[theorem]{Example}
\newtheorem{fact}[theorem]{Fact}
\newtheorem{observation}[theorem]{Observation}

\theoremstyle{definition}
\newtheorem{definition}[theorem]{Definition}
\newtheorem{remark}[theorem]{Remark}
\newtheorem{convention}[theorem]{Convention}

\newcommand{\red}[1]{{\color{red}#1}}
\newcommand{\blue}[1]{{\color{blue}#1}}
\newcommand{\purple}[1]{{\color{purple}#1}}
\newcommand{\orange}[1]{{\color{orange}#1}}

\renewcommand{\emptyset}{\varnothing}
\renewcommand{\epsilon}{\varepsilon}

\renewcommand{\setminus}{-}

\newcommand{\IH}{\mathit{IH}}

\newcommand{\extensionality}{\mathrm{ER}}
\newcommand{\induction}{\mathrm{Ind}}

\newcommand{\nmod}{\mathfrak{N}}
\newcommand{\interp}[1]{#1^{\nmod}}
\newcommand{\cmod}[1]{\mathcal C_{#1}}
\newcommand{\SN}[1]{\mathrm{SN}_{#1}}

\newcommand{\cimp}{\supset}
\newcommand{\cor}{\vee}
\newcommand{\cand}{\wedge}
\newcommand{\cnot}{\neg}
\newcommand{\ciff}{\equiv}

\newcommand{\hro}[1]{\hr{#1}}

\newcommand{\hr}[1]{\mathsf{HR}_{#1}}
\newcommand{\he}[1]{\mathsf{HE}_{#1}}

\newcommand{\seqar}{\Rightarrow}
\newcommand{\Nat}{\mathbb{N}}

\newcommand{\id}{\mathsf{id}}

\newcommand{\wk}{\mathsf{wk}}
\newcommand{\cntr}{\mathsf{cntr}}
\newcommand{\exch}{\mathsf{ex}}
\newcommand{\cut}{\mathsf{cut}}
\newcommand{\lefrul}[1]{#1_l}
\newcommand{\rigrul}[1]{#1_r}

\newcommand{\leftimp}{\mathsf{L}}
\newcommand{\rightimp}{\mathsf{R}}

\newcommand{\rec}{\mathsf{rec}}
\newcommand{\cond}{\mathsf{cond}}

\newcommand{\comp}{\circ}

\newcommand{\nat}{N}
\newcommand{\level}{\mathsf{lev}}

\newcommand{\farc}{farc}
\newcommand{\Farc}{Farc}
\newcommand{\FARC}{\mathsf{FARC}}

\newcommand{\rhotrans}[2]{{#1^{#2}}}

\renewcommand{\succ}{\mathsf{s}}
\newcommand{\pred}{\mathsf{p}}
\newcommand{\Kcomb}[2]{\mathrm{K}_{#1#2}}
\newcommand{\Scomb}[3]{\mathrm{S}_{#1#2#3}}

\newcommand{\apply}{\mathsf{app}}

\newcommand{\PRF}[1]{\mathrm{PRF}_{#1}}
\newcommand{\CRF}[1]{\mathrm{CRF}_{#1}}

\newcommand{\T}{\mathit{T}}
\newcommand{\C}{\mathit{CT}}
\newcommand{\nT}[1]{\T_{#1}}
\newcommand{\nC}[1]{\C_{#1}}
\newcommand{\aca}{\mathsf{ACA_0}}
\newcommand{\PA}{\mathsf{PA}}

\newcommand{\RCA}{\mathsf{RCA}_0}
\newcommand{\WKL}{\mathsf{WKL}_0}
\newcommand{\ACA}{\aca}

\newcommand{\ISn}[1]{I\Sigma_{#1}}
\newcommand{\proves}{\vdash}

\newcommand{\CIND}[1]{I#1}

\newcommand{\reduces}{\rightsquigarrow}
\newcommand{\unreduces}{\leftsquigarrow}
\newcommand{\unreducess}{\leftsquigarrowstar}
\newcommand{\conv}{\approx}
\newcommand{\numeral}[1]{\underline{#1}}

\newcommand{\exteq}{\simeq}

\newcommand{\branch}{\mathrm{Gen}}

\newcommand{\parred}{\vartriangleright}
\newcommand{\parder}{\vartriangleleft}

\newcommand{\ifthelse}[3]{\mathtt{if}\ #1=0\ \mathtt{then}\ #2\ \mathtt{else}\ #3}

\newcommand{\tree}{\mathrm{RedTree}}

\newcommand{\Con}{\mathrm{UNF}_\nat}

\newcommand{\REC}{\mathrm{Rec}}
\newcommand{\SREC}{\mathrm{SimRec}}
\newcommand{\runreduces}{\rightarrowtriangle}
\newcommand{\rununreduces}{\leftarrowtriangle}

\newcommand{\runle}{\oset{\: {+}}{\rununreduces}}

\newcommand{\runge}{\oset{{+}\: }{\runreduces}}

\newcommand{\tuple}[1]{\left\langle #1 \right\rangle}
\newcommand{\proj}[2]{\pi_{#1}(#2)}
\newcommand{\projraised}[2]{\pi_{#1}'(#2)}

\title{A circular version of G\"odel's $\T$ and its abstraction complexity}

\author{Anupam Das$^*$}
\address{$^*$University of Birmingham}
\email{a.das@bham.ac.uk}

\begin{document}

\maketitle

\begin{abstract}
\emph{Circular} and \emph{non-wellfounded} proofs have become an increasingly popular tool for metalogical treatments of systems with forms of induction and/or recursion. In this work we investigate the expressivity of a variant $\C$ of G\"odel's system $\T$ where programs are circularly typed, rather than including an explicit recursion combinator. 
In particular, we examine the abstraction complexity (i.e.\ type level) of $\C$, and show that the G\"odel primitive recursive functionals may be typed more \emph{succinctly} with circular derivations, using types precisely one level lower than in $\T$.
In fact we give a logical correspondence between the two settings, interpreting the quantifier-free type 1 \emph{theory} of level $n+1$ $\T$ into that of level $n$ $\C$ and vice-versa.

We also obtain some further results and perspectives on circular `derivations', namely strong normalisation and confluence, models based on hereditary computable functionals, continuity at type 2, and a translation to terms of $\T$ computing the same functional, at all types.
\end{abstract}

\section{Introduction}

In recent years \emph{non-wellfounded} proofs have attracted increasing attention.
The modern inception of the area arguably lies with the celebrated work \cite{nw96} of Niwinski and Walukiewicz, where a circular analytic tableau system was proved sound and complete for Kozen's modal $\mu$-calculus \cite{kozen83:mu-calc}.
Since then several distinct lines of research have emerged, in particular pivoting towards \emph{proof theoretical} aspects of non-wellfounded reasoning:
\begin{itemize}
\item \emph{Modal logic}. Niwinski and Walukiewicz's system has been recast as more traditional sequent based systems in, e.g., the works \cite{dhl06} for the linear time fragment, and \cite{studer08} for the general fragment, both offering alternative cut-free completeness proofs. Recently more `constructive' proofs have emerged of both these results, namely \cite{doumane17:multl-constr-comp} and \cite{afsharileigh17}. Similar results may be readily recovered for many other fixed point modal logics.
\item \emph{Predicate logic}. Brotherston and Simpson initiated a program for non-wellfounded proof theory based on predicate logic \cite{BroSim07:comp-seq-calc-ind-inf-desc,BroSim11:seq-calc-ind-inf-desc}. Here, syntactic correctness criteria from modal logic, themselves inspired by automaton theory, were adapted to provide a sound circular proof theory for forms of \emph{inductive definitions}. This approach has found applications in automated theorem proving  \cite{BroGorPet12:gen-cyc-prover,BroDisPet11:aut-cyc-ent,BroRow17:realiz-cyc-proof}, and more recently variations of these systems in the setting of first-order arithmetic have been investigated \cite{Sim17:cyclic-arith,BerTat17:lics,Das19:log-comp-cyc-arith,CohRow20:tc-logic}. 
\item \emph{Type systems}. Recently, formulations of cyclic proofs through the lens of the \emph{Curry-Howard} correspondence have garnered attention in the French school of proof theory. Fortier and Santocanale, also inspired by \cite{nw96}, seem to have been the first of the modern era to give a \emph{cut-elimination} result for circular proofs with the aforementioned correctness conditions,\footnote{Of course, in this regard, one must also mention Mints' famous `continuous cut-elimination' \cite{mints78:cont-cut-elim}; while his approach indeed seems to apply, the main difficulty herein seems to be the preservation of syntactic correctness at the limit.} namely for \emph{additive linear logic} with least and greatest fixed points \cite{ForSan13:cuts-circ-proofs-sem-cut-elim}. This was generalised in later work to the logic including multiplicatives too (for closed formulas) in \cite{BaeDouSau16:cut-elim}, and with more expressive correctness conditions recently in \cite{BaeDouKupSau:bouncing-threads}. Deepening this Curry-Howard viewpoint, presentations of `proof nets' have recently appeared \cite{DeSaurin:infinets:2019}, yielding a form of Natural Deduction for circular proofs.
\item \emph{Algebras}. Inspired by the aforementioned work on type systems, there have been many recent applications to classes of algebras based on fixed points, in particular the \emph{Kleene star}. \cite{DasPou17:cut-free-cyc-prf-sys-kl-alg} has presented a cut-free complete system for Kleene algebra, and established a cut-elimination result for it and the extension by \emph{residuals} and \emph{meets} (i.e.\ Lambek calculus + Kleene star) in \cite{DasPous18:nwf-lal}. 
These have been used directly to obtain alternative completeness results, e.g.\ \cite{DasDouPou18:left-handed-comp}, and have inspired recent undecidability results, e.g.\ for the logic of \emph{action lattices} \cite{Kuz19:AL-undec}, which solved a longstanding open problem.
\end{itemize}

A key motivation in all these areas is the so-called `Brotherston-Simpson conjecture': are cyclic proofs and inductive ones equally powerful? Naturally, the answer depends on how one interprets `equally powerful', e.g.\ as provability, proof complexity, logical complexity etc., as well as on the logic at hand. One example of this nuance is readily found in the setting of first-order logic:
\begin{itemize}
\item Berardi and Tatsuta have shown that first-order proofs with induction, over particular inductive definitions, do not prove the same theorems as cyclic ones \cite{BerTat17:cfolid-neq-folid}.  
\item Simpson in \cite{Sim17:cyclic-arith} and Berardi and Tatsuta in \cite{BerTat17:lics} have independently shown that inductive and cyclic proofs are equally powerful in the presence of Peano Arithmetic.
\item Both these results were refined in \cite{Das19:log-comp-cyc-arith} where it was shown that provability by cyclic proofs containing only $\Sigma_n$ formulas ($C\Sigma_n$) coincides with provability by $\Sigma_{n+1}$-induction ($I\Sigma_{n+1}$), over $\Pi_{n+1}$ theorems.
\end{itemize}

The current work is somewhat inspired by the line of work just mentioned, in the sense that it attempts to understand the Brotherston-Simpson question for type systems and their corresponding equational theories at the level of \emph{abstraction complexity} (or \emph{type level}).
In particular, motivated by the results for arithmetic mentioned above, our natural starting point is G\"odel's `system $\T$' \cite{dialectica}.

$\T$ is a multi-sorted classical quantifier-free theory over a simply typed programming language based on primitive recursion at all finite types.
G\"odel's celebrated \emph{Dialectica} functional interpretation \cite{dialectica} allows all of first-order arithmetic to be interpreted by this simple theory,\footnote{In fact, G\"odel's interpretation was only for the intuitionistic theory Heyting Arithmetic. Peano Arithmetic may be duly interpreted upon composition with a suitable double negation translation.} essentially trading off logical complexity of quantifiers for abstraction complexity in functional programs.
This tradeoff was made precise by Parsons in \cite{Par72:n-quant-ind}: the fragment $I\Sigma_{n+1}$ of Peano Arithmetic is $\mathit{ND}$-interpreted into the fragment of $\T$ admitting only type $n$ recursion ($\nT n$).
Naturally, a converse result holds too, in the sense that 
$\nT n$ may be interpreted into appropriate fragments of arithmetic where the interpretation of types is relativised to classes of hereditarily computable functionals. 

In this work we present a \emph{circular} (or \emph{cyclic}) version $\C$ of G\"odel's $\T$, where typing derivations may be non-wellfounded, but remain finitely presentable.
At the level of the type system, we simply adapt the correctness conditions of many of the previously mentioned works to the language of simple types.
At the level of the theory we admit the same axioms as $\T$, in particular quantifier-free induction, and include a form of extensional equality.
Our main result is that, similar to the arithmetic setting, $\C$ and $\T$ are mutually interpretable.
In fact, we obtain a similar refinement: the type $n$ restriction of $\C$ ($\nC n$) is interpretable in $\nT{n+1}$ and vice versa, over at least the type 1 quantifier-free theory.
Intuitively this means that cyclic typing derivations, and their induced theory, are more succinct than ones in $\T$, by precisely 1 type level.

Our arguments, however, are more subtle and technical than the analogous ones from \cite{Das19:log-comp-cyc-arith}, since we are working simultaneously with systems for typing and systems for reasoning within a theory.
We take a \emph{proof mining} approach to interpreting $\C$ in $\T$, formalising a totality argument in suitable fragments of `second-order'\footnote{As for simple type theory, the allusion to `second' or `higher' order is only suggestive: formally speaking these are multi-sorted first-order settings and not bona fide second or higher order.} arithmetic. 
Extraction of witnessing functionals and corresponding specifications in fragments of $\T$ follow by the aforementioned results of Parsons, under well-known conservation results over fragments of first-order arithmetic.
Notably, since we are unable to formalise the standard set theoretic model of higher order functionals, we take a detour through the model theory of $\T$, in particular presenting `coterm' models that play the roles of the \emph{hereditarily recursive} (and \emph{hereditarily effective}) operations.
We give a rewriting theoretic implementation of program execution, and establish a \emph{confluence} result within $\RCA$, yielding determinism of normalising programs.
The interpretation of $\C$ in $\T$ is then a consequence of the fact that these type structures indeed constitute models of $\C$.
For the refinement at each type level, we take advantage of recent results on the reverse mathematics of cyclic proof checking from \cite{Das19:log-comp-cyc-arith}, inspired by \cite{KMPS16:buchi-reverse,KMPS19:buchi-reverse}.

Our ultimate motivation is to establish a correspondence between two of the proof theoretic worlds mentioned at the start: predicate logic and type systems.
In future work we would like to establish a `circular Dialectica' functional interpretation, thereby completing the picture and formally associating the two settings.

\subsection{Related work}
Kuperberg, Pinault and Pous have notably also studied non-wellfounded typing derivations inspired by $\T$ in \cite{KupPinPou20:sysT}, in particular investigating \emph{affinity}.
Their types are closed under a `Kleene star' operation for list formation, inspired by previous works such as \cite{DasPous18:nwf-lal}, and are equivalent to the usual notion of simple (or \emph{finite}) types. They show that the affine fragment of this type system, where \emph{contraction} is omitted, computes precisely the primitive recursive functions in the standard set-theoretic model, generalising a similar result by Dal Lago for affine $\T$ \cite{DalLago09:linear-ho-rec}.
They also show that, in the presence of contraction, their type 1 fragment computes just the type 1 primitive recursive function\emph{al}s (\`a la G\"odel), i.e.\ those computed by type 1 terms of $\T$ in the standard model.

This latter result is subsumed by and, in particular, refined in the current work in terms of {type level}  (we discuss this further in Section~\ref{sect:fixed points}).
Moreover, the point of this work is to establish a \emph{logical} correspondence between fragments of $\C$ and $\T$, i.e.\ at the level of their equational/quantifier-free theories, not only at the level of interpretation of their terms in the standard model. 

\subsection{Outline and prerequisites}
The remainder of this paper is structured as follows. 
In Section~\ref{sect:stt} we give some preliminaries on Church's simple type theory and we recall system $\T$ in Section~\ref{sect:t}, in particular giving a sequent-style presentation of typing derivations.

In Section~\ref{sect:ct} we present non-wellfounded typing derivations (`coderivations'), in particular giving semantic results with respect to the standard model such as extensional completeness at type 1, Turing completeness for the regular fragment and, of course, well-definedness of the induced functionals. This section concludes with the definition of $\C$ and its type-level-restricted fragments $\nC n$.
In Section~\ref{sect:c-sim-t} we give a simulation of $\nT{n+1} $ within $\nC n$, over the type $n+1$ theory, in the presence of extensionality.
The techniques of this section are entirely proof-theoretic.

In Section~\ref{sect:coterm-models} we turn to the model theory of $\C$. We recast traditional type structures of hereditarily recursive and hereditarily effective operations into `coterm' models, in light of the aforementioned Turing-completeness result. At the base level, program execution is implemented as a rewrite system induced from the equational axioms in the usual way; in particular we prove a confluence result, yielding determinism of computation.
In Section~\ref{sect:t-sim-c} we formalise the aforementioned type structures within fragments of second-order arithmetic, in particular proving (within these fragments) that they constitute models of $\C$ (with extensionality). 
By applying standard proof mining results, we obtain an interpretation of $\nC n$ into $\nT{n+1}$, over the type 1 equational theory, a converse result to that of Section~\ref{sect:c-sim-t}.

Finally, in Section~\ref{sect:perspectives} we give some further results and perspectives on $\C$ coterms, in particular obtaining type 2 continuity, weak and strong normalisation, and a translation to $\T$ terms computing the same functional.
We also discuss proof theoretic strength, cut-elimination and the incorporation of fixed point operators.
\smallskip

It would not be pertinent to give a purely self-contained presentation of the content herein, since we rely on a number of established disciplines.
That said, we aim for a level of exposition that highlights the significance and subtleties of our results and techniques for the general proof theorist.

Naturally, it is helpful to have some background with G\"odel's system $\T$ and the Dialectica functional interpretation (though we shall not explicitly work with it), for which \cite{avigad-feferman:dialectica} is an excellent survey and \cite{Troelstra73:metamathematical-investigations,kohlenbach08:applied-proof-theory} are more comprehensive.
There are also excellent references for the technical disciplines underlining this work, namely rewriting theory (e.g.\ \cite{terese}), reverse mathematics (e.g.\, \cite{Sim09:reverse-math,Hir14:reverse-math}) and higher-order computability theory (e.g.\ \cite{ho-computability}).
Finally, we give metamathematical accounts of many of our results,\footnote{This is for two reasons: (a) for self-contained interest; while some such results are probably folklore, they have not appeared elsewhere, as far as we know; and, more importantly, (b) since we exploit these metamathematical resuls in Section~\ref{sect:t-sim-c} in order to interpret $\C$ within $\T$.} and so assume
some familiarity with
metamathematics of first- and second-order arithmetic (e.g.\ \cite{hajek-pudlak:metamathematics}, also \cite{kleene:intro-to-metamath} for a recursion-theoretic viewpoint and \cite{Troelstra73:metamathematical-investigations} for a constructive viewpoint).

\section{Preliminaries on Church's simple type theory}
\label{sect:stt}

At the heart of G\"odel's $\T$ is a rudimentary version of Church's Simple Type Theory \cite{church40:stt}. 
Since we later consider a rather non-standard `circular' calculus $\C$, we will here give a presentation of simple type theory that underlies both $\T$ and $\C$.

\subsection{Simple types}
Throughout this work we will deal with \emph{terms} that are \emph{simply typed}. (Simple) \textbf{types}, written $\sigma, \tau $ etc., are defined as follows:
\begin{itemize}
\item $\nat$ is a type.
\item If $\sigma $ and $\tau$ are types, then so is $(\sigma \to \tau)$.
\end{itemize}

\noindent
We typically omit parentheses on types when they are associated to the right. E.g., we may write $\rho \to \sigma \to \tau$ instead of $\rho \to (\sigma \to \tau)$ and so on.

\medskip

We define the \textbf{level} of a type $\sigma$, written $\level (\sigma)$, inductively as follows:
\begin{itemize}
	\item $\level (\nat) = 0$
	\item $\level (\sigma \to \tau) = \max (1 + \level (\sigma), \level (\tau))$.
\end{itemize}

\medskip

Every type $\sigma$ can be uniquely written as $\sigma_1 \to \cdots \to \sigma_n \to \nat$, for some $n \in \Nat$. 
In this case we sometimes write $\vec \sigma$ for $(\sigma_1, \dots, \sigma_n)$ and, as an abuse of notation, we also sometimes write $\vec \sigma \to \nat$ for $\sigma$.
We call $n$ here the \textbf{arity} of $\sigma$ (in reference to the type isomorphism $\rho \to \sigma \to \tau \equiv (\rho \times \sigma) \to \tau$).
Note that, in this case, $\level (\sigma) = 1 + \max\limits_{i=1}^n (\level (\sigma_i))$. 

\subsection{Simply typed theories (STTs)}

A simply typed \textbf{language} is a multi-sorted first-order language,
whose sorts are just the simple types.
The simply typed languages considered in this work will consist of some basic set of {constants} of simple type, called \emph{combinators}, as well as infinitely many \emph{variables}, written $x,y,z$ etc., of each simple type.
We may sometimes indicate the type of a variable (or term) as superscript to aid parsing, e.g.\ writing $x^\sigma$ for a variable $x$ of type $\sigma$.

\textbf{Terms} are formed from constants (and variables) by typed \emph{application}:
\begin{itemize}
    \item Any constant or variable of type $\sigma$ is a term of type $\sigma$.
    \item if $s$ and $t$ are terms of types $\sigma$ and $\sigma \to \tau$, respectively, then $(t \comp s)$ is a term of type $\tau$.
\end{itemize}
We usually just omit the application symbol $\comp$, e.g.\ writing $t\, s$ instead of $t \comp s$, and omit parentheses for long applications when they are associated to the left, e.g.\ writing $r\, s\, t$ for $(r\, s)\, t$ and so on.
	We do not include a $\lambda$-abstraction operation as primitive, instead requiring that it is coded by constants and composition for combinatory completeness (see, e.g., Fact~\ref{fact:comb-comp-ks}).
	
	\begin{remark}
	[Application]
	Formally speaking, being in a multi-sorted first-order framework, `application' itself comprises a family of operations, one for each pair of types $(\sigma \to \tau, \sigma)$.
	We shall gloss over this formality in what follows, unless we need to distinguish application operations of differing types.
	\end{remark}

Simply typed languages, for us, always include a binary relation symbol $=_\sigma$ `\textbf{equality} at type $\sigma$', on each type $\sigma$.
\textbf{Atomic formulas} have the form $s=_\sigma t$, where $s$ and $t$ are terms of type $\sigma$, although we shall suppress the subscript $\sigma$ when it does not cause confusion.
\textbf{Formulas} are built from atomic ones in the usual way, using $\cnot$ (negation), $\cor$ (disjunction), $\cand$ (conjunction), $\cimp$ (implication), $\ciff$ (if and only if), $\exists$ (existential quantifier), and $\forall$ (universal quantifier).

Simply typed \textbf{theories} (or \emph{STTs}) are (first-order, classical) theories over a (simply typed) language, typically specified by a set of quantifier-free axioms and rules.
 We always assume that STTs include the axioms for equality from Figure~\ref{fig:equality-axioms}. 

\begin{figure}[h]
\begin{enumerate}
	\item\label{item:refl-eq-axiom} $ t=_\sigma t$, for any term $t$ of type $\sigma$. (Reflexivity)
	\item\label{item:leibniz-eq-axiom}
	$(s =_\sigma t \cand  \phi(s))  \cimp \phi(t)$, for terms $s,t$ of type $\sigma$.
(Leibniz)
\end{enumerate}
\caption{Axiom schemata for (intensional) equality.}
\label{fig:equality-axioms}
\end{figure}

We will also include a form of extensionality for equality given by the extensionality rule ($\extensionality$) in Figure~\ref{fig:extensionality-axiom}.

\begin{figure}[h]
\begin{itemize}
\item[($\extensionality$)]\label{item:extensionality-axiom} 
If $\proves s \, \vec x \, = \, t\, \vec x\, $ then $\proves s=t$.
\end{itemize}
\caption{Extensionality rule.}
\label{fig:extensionality-axiom}
\end{figure}

\begin{remark}
[On equality]
Our inclusion of an equality symbol in all finite types coupled with a rule for extensionality is non-standard, but it eases some of the technical development.
Let us note, however, that our version has been previously considered in the literature, e.g.\ in \cite{Par72:n-quant-ind} where $\extensionality$ is called `SI' and facilitates the translation from $I\Sigma_{n+1}$ into $\nT n$ (cf.~Section~\ref{sect:fragments-of-t} later).

Our axiomatisation of equality thus sits somewhere between \emph{intensional} or \emph{weakly extensional} variants and the \emph{fully extensional} variant.
Since our principal concern in this work is in comparing STTs by their type 1 theories, the precise variant of equality is not so important, vis a vis known extensionality-elimination techniques at lower types, cf.~\cite{luckhardt73:ext-elim}.
Weaker formulations, in particular {weak extensionality}, seem to suffice for certain results, but the appropriate proof adaptations seem to introduce technicalities that detract from our main purpose.
Thus such considerations are beyond the scope of this work, but they will naturally be of greater importance for related directions, in particular for \emph{proof interpretations}.
\end{remark}

\subsection{Structures and the standard model}
\label{sect:structures-standard-model}

We consider usual Henkin structures for simply typed languages, called \emph{type structures}. Note that we do not, a priori, require $=_\sigma$ to be interpreted as true equality on the interpretation of $\sigma$.

One particular structure, the `standard' or `full set-theoretic' model $\nmod$, is given by the following interpretation of types:
\begin{itemize}
	\item $\nat^\nmod$ is $\Nat$.
	\item $(\sigma \to \tau )^\nmod$ is the set of functions $\sigma^\nmod \to \tau^\nmod$.
	\item $\comp^\nmod$ is just function application, i.e.\ given $f \in \sigma^\nmod$ and $g \in (\sigma \to \tau)^\nmod$, $g \comp^\nmod f \in \tau^\nmod$ is defined as $g(f)$.
	\item For each type $\sigma$, we have an \emph{extensional} equality relation $=_\sigma^\nmod$: 
	\begin{itemize}
		\item $=_\nat^\nmod$ is just equality of natural numbers;
		\item for $f,g \in (\sigma \to \tau )^\nmod$, we have $f=_{\sigma \to \tau}^\nmod g$ just if $ \forall x \in \sigma^\nmod .  f(x) =_\tau^\nmod g(x)$.
	\end{itemize}
\end{itemize}

As a notational convention, for a function $f : X\to {Z^Y}$ and $x \in X$, $y\in Y$, we may write $f(x,y)$ for $f(x)(y)$, and so on, in reference to the usual `Currying' isomorphism $(Z^Y)^X \cong Z^{X \times Y}$.

This structure also has standard extensions to the STTs considered later in this work, which will be presented at the appropriate moments.

\subsection{Example: Combinatory Algebra}
\label{sect:comb-alg}
The language of \emph{combinatory algebra} consists of the following constants:
\begin{itemize}
		\item $\Kcomb \sigma \tau$, for each pair $(\sigma, \tau)$, of type $\sigma \to \tau \to \sigma$.
	\item $\Scomb \rho \sigma \tau$, for each triple $(\rho, \sigma, \tau)$, of type $(\rho \to \sigma \to \tau) \to (\rho \to \sigma) \to \rho \to \tau$.
\end{itemize}
We will typically omit the type subscripts of the combinators in what follows when it is unambiguous.
\textbf{Combinatory Algebra} is a theory over this language that includes all axioms of the form:
\begin{equation}
	\label{eqn:ks-comb-axioms}
	\begin{array}{rcl}
		\Kcomb{}{} \, x\, y & = &   x \\
		\Scomb{}{}{} \, x\, y\, z & = &  x\, z\, (y\, z)
	\end{array}
\end{equation}
The standard model $\nmod$ may be extended to a model of Combinatory Algebra by taking the equations above as definitions (oriented from left to right).
It is well-known that Combinatory Algebra is \emph{complete}:
\begin{fact}
	\label{fact:comb-comp-ks}
	For each term $t$ of type $\tau$ and variable $x$ of type $\sigma$, there is a term $\lambda x t$ of type $\sigma \to \tau$ s.t.\ $\eqref{eqn:ks-comb-axioms} \proves\ (\lambda x t) \, y\ = \ t[y/x]$.
\end{fact}

\subsection{Sequent style type system}
\label{sect:sect:typing-rules}

\emph{Sequent calculi} give us a way to write typed terms that are more succinct with respect to type level, and also enjoy nice proof theoretic properties, e.g.\ cut-elimination. 
From the point of view of the \emph{Curry-Howard correspondence}, they associate sequent proofs of minimal logic to simply typed terms.
Importantly, the sequent presentation and its induced relations between type occurrences makes it easier to define our notion of progressing non-wellfounded derivation later.

\begin{definition}
	[Sequent calculus]
    Sequents are expressions $\vec \sigma \seqar \tau$, where $\vec \sigma$ is a list of types and $\tau $ is a type.
The rules for minimal logic are given in Figure~\ref{fig:seq-calc-min}.
\end{definition}
	\begin{figure}
		[h]
		\[
		\vlinf{\exch}{}{\red{\vec \rho}, \blue{\rho}, \orange{\sigma}, \purple{\vec \sigma }\seqar \tau}{\red{\vec \rho},  \orange{\sigma}, \blue{\rho}, \purple{\vec \sigma }\seqar \tau}
		\qquad
		\vlinf{\wk}{}{\purple{\vec \sigma}, \blue{\sigma} \seqar \tau}{\purple{\vec \sigma }\seqar \tau}
		\qquad
		\vlinf{\cntr}{}{\purple{\vec \sigma}, \blue{\sigma} \seqar \tau}{\purple{\vec \sigma}, \blue{\sigma}, \blue{\sigma }\seqar \tau}
		\qquad 
		\vliinf{\cut}{}{\purple{\vec \sigma }\seqar \tau}{\purple{\vec \sigma} \seqar  \sigma }{\purple{\vec \sigma}, \sigma \seqar \tau}
		\]
		\[
		\vlinf{\id}{}{\purple{\sigma} \seqar \sigma}{}
		\qquad
		\vliinf{\leftimp}{}{\purple{\vec \sigma }, \blue{\rho \to \sigma} \seqar \tau}{\purple{\vec \sigma} \seqar \rho}{\purple{\vec \sigma}, {\sigma }\seqar \tau}
		\qquad 
		\vlinf{\rightimp}{}{\purple{\vec \sigma }\seqar \sigma \to \tau}{\purple{\vec \sigma}, \sigma \seqar \tau}
		\]
		\caption{Sequent style typing rules from minimal logic.}
		\label{fig:seq-calc-min}
	\end{figure}

Here, and throughout this section, colours of each type occurrence in typing rules may be ignored for now and will become relevant later in Section~\ref{sect:ct}. That said, it may be illustrative for the reader to imagine that, once we give interpretations of these rules, type occurrences of the same colour will correspond to identical inputs for the corresponding functionals.
In this way, the colour assigned to a type on the LHS of a sequent is a sort of variable annotation for that occurrence.

Recall that, for a list of types $\vec \sigma = (\sigma_1, \dots, \sigma_n)$, we sometimes write $\vec \sigma \to \sigma$ for the type $\sigma_1 \to \cdots \to \sigma_n \to \sigma$.
Each rule instance (or \textbf{step}) determines a constant of the appropriate type:
\begin{itemize}
	\item A step $\vlinf{}{}{\vec \tau \seqar \tau}{}$ is a constant of type $\vec \tau \to \tau$.
	\item A step $\vlinf{}{}{\vec \tau \seqar \tau}{\vec \sigma \seqar \sigma}$ is a constant of type
	\(
	(\vec \sigma \to \sigma) \to \vec \tau \to \tau
	\).
	\item A step $\vliinf{}{}{\vec \tau \seqar \tau}{\vec \rho\seqar \rho}{\vec \sigma \seqar \sigma}$ is a constant of type
	$(\vec \rho \to \rho) \to (\vec \sigma \to \sigma) \to \vec \tau \to \tau$.
\end{itemize}

We will usually refer to steps only by their labels, e.g.\ $\exch$, $\cntr$, $\rightimp$ etc., rather than explicitly indicating their types in the premisses and conclusions; unless otherwise clear from context, the associated typing should be assumed to be as given in the original specification of the rule, e.g.\ in Figure~\ref{fig:seq-calc-min}.

\begin{definition}
[Derivations and terms]
A sequent calculus \textbf{derivation} of $\vec \sigma \seqar \tau$ determines a term of type $\vec \sigma \to \tau$ in the expected way, by applying all the inference steps according to its structure.
Formally, derivations are construed as terms by inductively setting:
\begin{itemize}
	\item $\vlderivation{
		\vlin{\mathsf r}{}{\vec \tau \seqar \tau}{
		\vltr{s}{\vec \sigma \seqar \sigma}{\vlhy{ \ }}{\vlhy{}}{\vlhy{\ }}
		}
	}$ is the term $\mathsf r\, s$; and,
	\item $\vlderivation{
		\vliin{\mathsf r}{}{\vec \tau \seqar \tau}{
			\vltr{r}{\vec \rho \seqar \rho}{\vlhy{\ }}{\vlhy{ }}{\vlhy{\ }}
		}{
			\vltr{s}{\vec \sigma \seqar \sigma}{\vlhy{\ }}{\vlhy{ }}{\vlhy{\ }}
		}
	}$ is the term $\mathsf r\ r\, s$.
\end{itemize}
We write $t: \vec \sigma \seqar \tau$ if $t$ is a derivation of the sequent $\vec \sigma \seqar \tau$ (and so also a term of type $\vec \sigma \to \tau$).
\end{definition}

Note that, strictly speaking, derivations form a strict subset of all closed terms, since they are not formally closed under application. 
I.e., for $s:\ \seqar \sigma$ and $t:\sigma \seqar \tau$, we have that $t\, s $ is a term of type $\tau$, but $t\, s $ is not, in general, a derivation.\footnote{E.g., if $t$ concludes with a binary step, then the only way to write $t\, s $ in the form $\mathsf r\, r_1 \cdots r_n$, with $\mathsf r$ an inference step, requires $n>2$, which is not possible since all rules are at most binary.}
In what follows, we will see that we can \emph{interpret} the term $t\, s$ as the derivation $\cut\, t\, s$ (with appropriate types), but it will nonetheless be convenient to distinguish these two terms.
In particular this interpretation does not admit the same notion of `thread' for non-wellfounded derivations we consider later.

\begin{remark}
	[Why rules as combinators?]
While it may seem strange to adopt a sequent style type system but construe inference steps as combinators rather than meta-level operations on, say,  $\lambda$-terms, we adopt this approach to facilitate our later notions of non-wellfounded derivations and coterms.
This combinatory approach ensures that the `term associated to a derivation' is actually a \emph{continuous} construction, so that when we later consider non-wellfounded derivations, the corresponding notion of `coterm' is well-defined.

On the other hand, why do we use the sequent calculus at all? This is due to the particular termination criterion we will adopt for non-wellfounded derivations later, exploiting  well-known notions of formula \emph{ancestry} available in the sequent calculus.
This is why the sequent calculus is the standard formalism in circular proof theory.
\end{remark}

\begin{definition}
	[Axiomatisation]
	We define the axiomatisation in Figure~\ref{fig:eq-ax-seq-calc-min}, where the types corresponding to each rule label are as indicated in the corresponding rule instance in Figure~\ref{fig:seq-calc-min}.

	\begin{figure}
		[h]
		\[
		\begin{array}{rcl}
		\id \ x & = & x \\
		\exch \ t \ \vec x\  x \ y\ \vec y & = & t \ \vec x \ y \ x \ \vec y \\
		\wk \ t  \ \vec x \ x & = & t \ \vec x \\
		\cntr \ t\ \vec x \ x & = & t \ \vec x \ x \ x \\
		\cut \ s\ t\ \vec x & = & t \ \vec x \ (s \ \vec x ) \\
		\leftimp \ s\ t\ \vec x\ y & = & t \ \vec x\ (y\ (r\ \vec x) ) \\
		\rightimp \ t \ \vec x \ x & = & t\ \vec x\ x
		\end{array}
		\]
		\caption{Equational axiomatisation of sequent calculus rules.}
		\label{fig:eq-ax-seq-calc-min}
	\end{figure}
\end{definition}
Note that, here and elsewhere, there is no formal reason why we distinguish the arguments corresponding to subderivations by term meta-variables, $s,t$ etc., and other arguments by variables $x,y$ etc. It is purely in order to facilitate the identification of the corresponding arguments. All axioms are closed under substitution of terms for variables.

\begin{remark}
    [Combinatory completeness]
    The sequent calculus in Figure~\ref{fig:seq-calc-min}, under the axiomatisation in Figure~\ref{fig:eq-ax-seq-calc-min}, is equivalent to Combinatory Algebra (from Section~\ref{sect:comb-alg}).
    In particular there are derivations for $\Kcomb \sigma \tau$ and $\Scomb \rho \sigma \tau$ that satisfy the corresponding equational axioms.
    As a consequence, our sequent calculus is also combinatory complete.
\end{remark}

\begin{remark}
[Standard model]
	The structure $\nmod$ from Section~\ref{sect:structures-standard-model} may be extended to one for Figures~\ref{fig:seq-calc-min} and \ref{fig:eq-ax-seq-calc-min} by taking the equations of Figure~\ref{fig:eq-ax-seq-calc-min} as definitions, oriented left-to-right.
\end{remark}

\begin{convention}
	[Rules modulo exchange]
	\label{conv:rules-mod-exch}
In the rest of this work, it will often be convenient to omit instances of exchange, $\exch$, in typing derivations and their corresponding notation as terms.
For example, we may freely write a `rule instance',
\begin{equation}
\label{eqn:cut-with-exchange}
\vliinf{\cut}{}{\purple{\vec \rho}, \blue{\vec \sigma} \seqar \tau}{\purple{\vec \rho}, \blue{\vec \sigma} \seqar \sigma}{\purple{\vec \rho}, \sigma , \blue{\vec \sigma }\seqar \tau}
\end{equation}
instead of the corresponding derivation with exchanges, and an `axiom',
\[
\cut\, s\, t\, \vec x\, \vec y \ = \ t\, \vec x\, (s\, \vec x\, \vec y)\, \vec y
\]
instead of the corresponding equation derived using the $\exch$ and $\cut $ axioms. 
Again, we may omit typing of rule labels when it is unambiguous.
\end{convention}

\section{Preliminaries on G\"odel's system $\T$}
\label{sect:t}

So far we have not imposed any restrictions on our base type $\nat$, despite the fact that the standard model $\nmod$ interprets $\nat$ as $\Nat$. 
System $\T$ is a simple type theory that extends Combinatory Algebra by including new constants, axioms and rules that constrain the interpretation of $\nat$ to this effect.
Its definition is borne out over the following subsections.

\subsection{Constants for natural numbers and recursion}
The language of $\T$ extends the sequent system from Figure~\ref{fig:seq-calc-min} by the typing rules in Figure~\ref{fig:rules-0-s-rec}.
Again, we may omit the subscript (and other typing information) of an instance of $\rec_\tau$ when it is unambiguous.

\begin{figure}[h]
\[
\vlinf{0}{}{\seqar \nat}{}
\qquad
\vlinf{\succ}{}{\purple{\nat} \seqar \nat}{}
\qquad
\vliinf{\rec_\tau}{}{ \purple{\vec \sigma}, \blue{\nat }\seqar \tau}{\purple{\vec \sigma }\seqar \tau}{ \purple{\vec \sigma}, \blue{\nat}, \sigma \seqar \tau}
\]
\caption{Typing rules for $0$, $\succ$ and $\rec_\tau$ combinators.}
\label{fig:rules-0-s-rec}
\end{figure}

When writing terms, we assume $\succ$ binds stronger than application. This is usually visually signified since the symbol $\succ$ will appear in closer proximity to the term it is bound to. E.g. We write $s\ \succ t$ for $s \comp (\succ \comp t)$.

A \textbf{numeral} is a term of the form $\underbrace{\succ \cdots \succ}_n 0$, which we more succinctly write as $\numeral n$.

\subsection{Recursion axioms}

$\T$ includes as axioms the equations from Figure~\ref{fig:rec-axioms}. 
\begin{figure}[h]
    \[
\begin{array}{rcl}
\rec\ s\ t\ \vec x\  0 &=& s\, \vec x \\
\rec\ s\ t\ \vec x\ \succ y &=& t\ \vec x\ y \ (\rec\ s\ t\ \vec x \ y)
\end{array}
\]
    \caption{Equational axioms for recursion combinators, where $y$ is a variable of type $\nat$.}
    \label{fig:rec-axioms}
\end{figure}
Before concluding our definition of $\T$, let us note that the equational axioms thus far presented are enough to expose well-behaved computational content:\footnote{Note that this result also applies to alternative combinatorial bases such as ours, e.g.\ as noted in \cite{curry-hindley-seldin:72}, Chapter B.}

\begin{fact}
	[\cite{Tait:67:normalisation-of-t-+-bar-rec-typ01}]
	\label{fact:normalisation+confluence-of-t}
	Orienting the axioms of Figures~\ref{fig:eq-ax-seq-calc-min} and \ref{fig:rec-axioms} left-to-right yields a terminating and confluent rewriting system on closed terms of $\T$.
\end{fact} 

We will not elaborate now on the rewriting theoretic aspects of $\T$ since we will revisit it in more detail later in Section~\ref{sect:coterm-models}.
However, let us note that Tait's result above induces a well-behaved \emph{term model} of $\T$, with equality simply comparing (unique) normal forms.
The point of Section~\ref{sect:coterm-models} is to establish similar models for the non-wellfounded type system we will introduce in Section~\ref{sect:ct}, formalised in the setting of second-order arithmetic.

\subsection{Number-theoretic axioms}

Finally, $\T$ includes the axioms from Figure~\ref{fig:ax-t-non-eq}, indicating that $(0,\succ)$ generates a free and inductive structure.

\begin{figure}[h]
	\begin{enumerate}
		\item\label{item:ax-0-min} $\cnot\,  \succ x = 0 $
		\item\label{item:ax-succ-inj} $\succ x = \succ y\,  \cimp \, x=y$
		\item[($\induction$)]\label{item:induction-schema-t} 
		If $\proves \phi(0)$ and $\proves \phi(x) \cimp \phi(\succ x)$ then $\proves \phi(t)$, for $\phi$ quantifier-free.
	\end{enumerate}
	\caption{Number-theoretic axioms for $\T$, where $x$ and $y$ are variables of type $\nat$, and $t$ is a term of type $\nat$.}
	\label{fig:ax-t-non-eq}
\end{figure}

\noindent
This concludes the definition of $\T$, i.e.:
\begin{definition}
[System $\T$]
$\T$ is the simple type theory over the language given by Figures~\ref{fig:seq-calc-min} and Figures~\ref{fig:rules-0-s-rec},
axiomatised by the formulas and rules from Figures~\ref{fig:equality-axioms}, \ref{fig:extensionality-axiom}, \ref{fig:eq-ax-seq-calc-min}, \ref{fig:rec-axioms} and \ref{fig:ax-t-non-eq}.
\end{definition}

\smallskip

Going back to Fact~\ref{fact:normalisation+confluence-of-t} and the succeeding discussion, one crucial property of $\T$ is that the only (closed) normal forms of type $\nat$ are numerals.
For the aforementioned term model induced by unique normal forms of closed terms,
this property allows the verification of the induction schema \eqref{item:induction-schema-t} in Figure~\ref{fig:ax-t-non-eq} to be reduced to induction at the meta-level.
Interestingly,
this property \emph{fails} for the non-wellfounded calculus we will present in Section~\ref{sect:ct} (see Remark~\ref{rmk:non-numeral-normal}), necessitating a somewhat specialised construction of corresponding models in Section~\ref{sect:coterm-models}.

\subsection{The standard model and primitive recursive functionals}
The standard model $\nmod$ from Section~\ref{sect:structures-standard-model} may be extended to one of $\T$ by setting,
\begin{itemize}
	\item $0^\nmod:= 0 \in \Nat$.
	\item $\succ^\nmod (n) := n+1$.
\end{itemize}
and taking the axioms for $\rec$ from Figure~\ref{fig:rec-axioms} as definitions, oriented left-to-right.
Note that the interpretation of $\rec$ is indeed well-defined by these axioms, provable by induction on $\Nat$.

The interpretations of terms in this model, i.e.\ the functionals $t^\nmod$, form a higher-order function algebra known as the (G\"odel) \emph{primitive recursive functionals} of finite type, written $\PRF{}$. 
It is well-known by a result of Kreisel that its type 1 functions coincide with those definable by effective transfinite recursion up to $\epsilon_0$ \cite{Kreisel51,Kreisel52}, and moreover that ordinal complexity (height of an $\omega$-tower) can be effectively traded off for abstraction complexity (type level) and vice-versa (cf., e.g., \cite{tait68:constructive-reasoning,schwichtenberg75:highertypes-vs-type-level}).

\subsection{Restricting the level of recursors}
\label{sect:fragments-of-t}
The main subject of study in this work will be fragments of $\T$ induced by restricting the type level of recursors.

\begin{definition}
[Fragments of $\T$]
	 $\nT n$ is the restriction of $\T$ to the language containing only recursors $\rec_\sigma$ where $\level (\sigma) \leq n$.
\end{definition}

The significance of these fragments was investigated in the seminal work of Parsons \cite{Par72:n-quant-ind}.
In particular we have:
\begin{proposition}
	[\cite{Par72:n-quant-ind}]
	\label{prop:parsons-arith-to-t}
	If $\ISn{n+1} \proves \forall \vec x \exists  y A(\vec x,y)$, where $A$ is $\Delta_0$, then there is a $\nT n$ term $t$ with $\nT n \proves A(\vec x, t\, \vec x)$.\footnote{We assume here some standard encoding of $\Delta_0$ formulas into quantifier-free formulas of $\nT 0$. Alternatively we could admit bounded quantifiers into the language of $\T$, on which induction is allowed, without affecting expressivity. We shall gloss over this technicality here.}
\end{proposition}
\noindent
In fact, this result is a direct consequence of G\"odel's famous `Dialectica' functional interpretation \cite{dialectica}, composed with a suitable negative translation. 
The converse, that $\ISn{n+1}$ proves the totality of all type 1 terms of $\nT n$, is obtained by formalising models of hereditarily computable functionals similar to those in Section~\ref{sect:coterm-models}.
Both directions may be alternatively obtained via the aforementioned transfinite recursion theoretical characterisations of $\T$, using purely structural proof theoretic methods, cf.~ \cite{Buss1995witness}.

Both results naturally extend to the conservative extension $\RCA + \CIND{\Sigma^0_{n+1}}$.
\begin{corollary}
\label{cor:rca+isn+1-to-tn}
If $\RCA+ \ISn{n+1} \proves \forall \vec x \exists  y A(\vec x,y)$, where $A$ is $\Delta^0_0$, then there is a $\nT n$ term $t$ with $\nT n \proves A(\vec x, t\, \vec x)$.
\end{corollary}

For the results of Section~\ref{sect:c-sim-t}, it will be useful to have the following normal form of typing derivations:
\begin{proposition}
	[Partial normalisation]
	\label{prop:free-cut-elim}
	Let $t: \vec \sigma \seqar \tau$ be a $\nT n$ derivation, where $\tau$ and each $\sigma_i$ have level $\leq n$. Then there is a $\nT n$-derivation $t':\vec \sigma \seqar \tau$ such that $t^\nmod=t'^\nmod$.
	Moreover, $\nT n \proves t = t' $.
\end{proposition}
Since we could not easily find an explicit statement of this in the literature, a self-contained proof is given in Appendix~\ref{sect:free-cut-elim}.

\subsection{Example: typing the Ackermann-P\'eter function}
\label{sect:ack-type-1-rec}
Let us take a moment to see an example of typing and reasoning within $\T$.
The Ackermann-P\'eter function $A: \Nat \times \Nat \to \Nat$ is defined by the following equations:
\begin{equation}
	\label{eqn:ackermann-equations}
	\begin{array}{rcl}
		A(0,y) & := & y+1\\
		A(x+1,0) &:=& A(x,1)\\
		A(x+1,y+1) &:=& A(x,A(x+1,y))
	\end{array}
\end{equation}
Formally, we may see $A$ as being defined by induction on a lexicographical product order on $\Nat\times \Nat$.
This function may duly be computed by a term of $\T$ by appealing to primitive recursion at type level $1$.
We first define a functional $\mathrm I$ by primitive recursion (at type $\nat$) satisfying:
\begin{equation}
	\label{eqn:iterator-eqns}
	\begin{array}{rcl}
		\mathrm I\, f \, 0 & = & f\, \numeral 1 \\
		\mathrm I\, f \, \succ y & = & f\, (\mathrm I\, f\, y)
	\end{array}
\end{equation}
Formally, $\mathrm I$ may be typed by the following derivation,
\[
\vlderivation{
    \vliin{\rec_\nat}{}{\nat \to \nat , \underline{{\nat}} \seqar \nat}{
        \vliin{\leftimp}{}{\underline{\nat \to \nat} \seqar \nat}{
            \vltr{1}{\seqar \nat}{\vlhy{\ }}{\vlhy{}}{\vlhy{\ }}
        }{
            \vlin{\id}{}{\nat \seqar \nat}{\vlhy{}}
        }
    }{
        \vliin{\leftimp}{}{\underline{\nat \to \nat}, \red{\nat} \seqar \nat}{
            \vlin{\id}{}{\red{\nat} \seqar \nat}{\vlhy{}}
        }{
            \vlin{\wk}{}{\nat, \underline{\red{\nat}} \seqar \nat}{
            \vlin{\id}{}{\nat \seqar \nat}{\vlhy{}}
            }
        }
    }
}
\]
where principal types are underlined and red occurrences of $\red \nat$ correspond to the same input (morally $y$ in \eqref{eqn:iterator-eqns}).
$\T$ proves the defining equations from \eqref{eqn:iterator-eqns} for $\mathrm I$:
\[
\begin{array}{rcll}
     \mathrm I\, f\, 0 & = & \leftimp \, \numeral 1\, \id\, f & \text{by $\rec$ axioms}\\
     & = & \id\, (f\, \numeral 1) & \text{by $\leftimp$ axiom} \\
     &= & f\, \numeral 1 & \text{by $\id$ axiom}
\\
\noalign{\medskip}
     \mathrm I\, f\, \succ y & = & \leftimp\, \id\, (\wk\, \id) \, f \, (\mathrm I\, f\, y) & \text{by $\rec$ axioms}\\
     & = & \wk\, \id\, (f\, (\id\, (\mathrm I\, f\, y)))\, (\mathrm I \, f\, y) & \text{by $\leftimp$ axiom}\\
     &=& \id\, (f\, (\id\, (\mathrm I\, f\, y))) & \text{by $\wk$ axiom}\\
     & = & f\, (\id\, (\mathrm I\, f\, y)) & \text{by $\id$ axiom}\\
     & = & f\, (\mathrm I\, f\, y) &\text{by $\id$ axiom}
\end{array}
\]

From here $\mathrm A $ is obtained by primitive recursion at type $\nat \to \nat$, satisfying:
\begin{equation}
\label{eqn:ackermann-abstracted-eqns}
\begin{array}{rcl}
\mathrm A\, 0 & = & \succ \\
\mathrm A\, \succ x & =& \mathrm I\, (\mathrm A\, x)
\end{array}
\end{equation}
Formally, such $\mathrm A$ may be typed by the following derivation,
\[
\vlderivation{
	\vliin{\rec_{\nat \to \nat}}{}{\underline{\nat} \seqar \nat \to \nat}{
		\vlin{\rightimp}{}{\seqar \underline{\nat \to \nat}}{
		\vlin{\succ}{}{\nat \seqar \nat}{\vlhy{}}	
		}
	}{
		\vlin{\wk}{}{\underline{\nat}, \nat \to \nat \seqar \nat \to \nat}{
		\vlin{\rightimp}{}{\nat \to \nat \seqar \underline{\nat \to \nat}}{
		\vltr{\mathrm I}{\nat \to \nat, \nat \seqar \nat}{\vlhy{\ }}{\vlhy{}}{\vlhy{\ }}
		}
		}
	}
}
\]
where principal types are underlined.

\begin{proposition}
\label{prop:t1-proves-ack-eqns}
$\nT 1$ proves the following equations:
\[
\begin{array}{rcl}
\mathrm A\, 0\, y & = & \succ y \\
\mathrm A\, \succ x \, 0 & = & \mathrm A\, x\, \numeral 1 \\
\mathrm A\, \succ x\, \succ y & = &  \mathrm A\, x\, (\mathrm A\, \succ x \, y)
\end{array}
\]
\end{proposition}
\begin{proof}
We have,
\[
\begin{array}{rcll}
    \mathrm A\, 0 \, y & = & \rightimp \, \succ \, y & \text{by $\rec $ axioms}\\
    & = & \succ y \, & \text{by $\rightimp$ axiom}
\\
\noalign{\medskip}
    \mathrm A\, \succ x \, y & = & \wk\, (\rightimp \, \mathrm I)\, x\, (\mathrm A\, x)\, y & \text{by $\rec$ axioms} \\
    & = & \rightimp \, \mathrm I\, (\mathrm A\, x)\, y & \text{by $\wk$ axiom} \\
    & =& \mathrm I\, (\mathrm A\, x)\, y & \text{by $\rightimp$ axiom} \quad (*)
\end{array}
\]
whence we have immediately $\mathrm A\, \succ x\, 0 \ = \ \mathrm A\, x\, \numeral 1$ by \eqref{eqn:iterator-eqns}.
For $y$ non-zero, we have:
\[
\begin{array}{rcll}
	\mathrm A\, \succ x\, \succ y & = & \mathrm I\, (\mathrm A\, x)\, \succ y & \text{by $(*)$ above} \\
	& = & \mathrm A\, x\, (\mathrm I\, (\mathrm A\, x)\, y) & \text{by \eqref{eqn:iterator-eqns}}\\
		& = & \mathrm A\, x\, (\mathrm A\, \succ x\, y) & \text{by $(*)$ above} \qedhere
\end{array}
\]
\end{proof}

The recursors used to type $\mathrm A$ have level $1$.
This is not a coincidence, since primitive recursion at level $0$ (i.e.\ on only natural numbers) computes just the primitive recursive functions:
\begin{fact}
\label{fact:type-0-rec-prim-rec}
If $t:\nat^k \to \nat$ is a term of $\nT 0$ then $t^\nmod : \Nat^k \to \Nat$ is primitive recursive.
\end{fact}
Note that, together with Proposition~\ref{prop:parsons-arith-to-t}, this constitutes a proof that $\ISn 1 $ (or $\RCA$) well-defines just the primitive recursive functions.

\section{A circular version of $\T$}
\label{sect:ct}

We will now move on to the main subject of study in this work: typing `derivations' that are non-wellfounded and their corresponding notion of term.
Let us henceforth write $\T^-$ for the restriction of $\T$ to the language without recursion combinators $\rec_\tau$.

\begin{definition}
	[Conditional combinator]
We introduce a new typing rule $\cond$ for derivations, as well as corresponding axioms, in Figure~\ref{fig:cond-rule+axioms}.
As before, the colouring of type occurrences above will become apparent soon.

\end{definition}
\begin{figure}[h]
\[
\vliinf{\cond}{}{ \purple{\vec \sigma  }, \blue{\nat }\seqar \tau}{\purple{\vec \sigma }\seqar \tau}{ \purple{\vec \sigma  }, \blue{\nat }\seqar \tau}
\qquad\qquad
	\begin{array}{rcl}
		\cond\ s\ t\ \vec x \  0 & = &  s \ \vec x \\
		\cond \ s\ t\ \vec x \ \succ y  & = & t \ \vec x \  y
	\end{array}
	\]
	\caption{Typing rule and axioms for $\cond$ combinators.}
		\label{fig:cond-rule+axioms}
\end{figure}

Again, the interpretation of $\cond$ in the standard model $\nmod$ is uniquely determined by the defining axioms of Figure~\ref{fig:cond-rule+axioms}.
Throughout this section, we will work in the language of $\T^- + \cond$, unless otherwise specified.

\subsection{Non-wellfounded `terms' and `derivations'}

\textbf{Coterms} are generated \emph{coinductively} from constants and variables under typed application.
Formally, we may construe a coterm as a possibly infinite binary tree (of height $\leq \omega$) where each leaf (if any) is labelled by a typed variable or constant and each interior node is labelled by a typed application operation, having type consistent with the types of its children. I.e., an interior node with children of types $\sigma$ and $\sigma \to \tau$, respectively, must have type $\tau$.
 
At the risk of confusion, we expand the range of the metavariables $s,t,$ etc.\ to include coterms as well as terms, clarifying further only in the case of ambiguity.
We adopt the same writing and bracketing conventions as for terms, where it is meaningful, e.g.\ writing $r\, s\, t$ for $(r \comp s)\comp t$.

We will not dwell much on arbitrary coterms, since we will only deal with those induced by our sequent style type system.

\begin{definition}
	[Coderivations]
		A \textbf{coderivation} is some possibly non-wellfounded `derivation' built from the typing rules of $\T^- + \cond$, in a locally correct manner. 
		Formally, a coderivation is a possibly infinite labelled binary tree (of height $\leq \omega$) whose nodes are labelled by rule instances s.t.\ the premisses of a node (if any) match the conclusions of the node's respective children (if any).
		
		We construe coderivations as coterms in the same way as we construed derivations as terms.
		Namely, we coinductively set:
		\begin{itemize}
			\item $\vlderivation{
				\vlin{\mathsf r}{}{\vec \tau \seqar \tau}{
					\vltr{s}{\vec \sigma \seqar \sigma}{\vlhy{ \ }}{\vlhy{}}{\vlhy{\ }}
				}
			}$ is the coterm $\mathsf r\, s$; and,
			\item $\vlderivation{
				\vliin{\mathsf r}{}{\vec \tau \seqar \tau}{
					\vltr{r}{\vec \rho \seqar \rho}{\vlhy{\ }}{\vlhy{ }}{\vlhy{\ }}
				}{
					\vltr{s}{\vec \sigma \seqar \sigma}{\vlhy{\ }}{\vlhy{ }}{\vlhy{\ }}
				}
			}$ is the coterm $\mathsf r\ r\, s$.
		\end{itemize}
		Note that this association is continuous at the level of the underlying trees, so it is indeed well-defined.
		
		Again overloading notation, we will write $t:\vec \sigma \seqar \tau$ if $t$ is a coderivation of the sequent $\vec \sigma \seqar \tau$, and $t:\sigma$ if $t$ is a coterm of type $\tau$. 
		If we need to distinguish $t$ as a term then we will say so explicitly.
\end{definition}

Note that the equational theory induced by Figures~\ref{fig:eq-ax-seq-calc-min}, \ref{fig:rec-axioms} and \ref{fig:cond-rule+axioms} form a Kleene-Herbrand-G\"odel
 style equational specification for coterms (cf., e.g., \cite{kleene:intro-to-metamath}), now understanding the metavariables $s,t$ etc.\ there to range over coterms. 
We may thus view coterms as \emph{partial recursive functionals} in the standard model $\nmod$ of the appropriate type.  More formally:

\begin{definition}
	[Interpretation of coterms as partial functionals]
	\label{dfn:int-of-coterms-as-partial-fns}
	We define a type structure $\nmod_\bot$ by interpretations $\cdot^\nmod_\bot$ and corresponding `totally undefined functionals' $\bot_\sigma$ as follows:
	\begin{itemize}
		\item $\bot_\nat$ is just some fresh element $\bot$.
		\item $\nat^\nmod_\bot$ is $\Nat \cup \{\bot \}$.
		\item $\bot_{\sigma \to \tau} : \sigma^\nmod_\bot \to \tau^\nmod_\bot$ by $a\mapsto \bot_\tau$, for any $a \in \sigma^\nmod_\bot$.
		\item $(\sigma \to \tau)^\nmod_\bot$ is the set of functions $f:\sigma^\nmod_\bot \to \tau^\nmod_\bot$ s.t.\ $f(\bot_\sigma) = \bot_\tau$. 
		\item $=^\nmod_\bot$ is just extensional equality (for each type).
	\end{itemize}
	A \emph{partial functional of type $\sigma$} is just a function in $\sigma^\nmod_\bot$.
	A \emph{(total) functional of type $\sigma$} is just a partial functional $f$ of type $\sigma$ with $f(x) = \bot $ if and only if $x = \bot$.

	We now define the interpretation of coterms in $\nmod_\bot$ as follows:
	\begin{itemize}
	\item If $t: \nat$ then $t^\nmod_\bot = n \in \Nat$ just if $n$ is the unique interpretation of $t$ (under the equations of Figures~\ref{fig:eq-ax-seq-calc-min} and \ref{fig:cond-rule+axioms}) in $\nmod_\bot$. Otherwise $t^\nmod_\bot$ is $\bot$.
	\item If $t : \sigma \to \tau$ and $a \in \sigma^\nmod_\bot$ then $t^\nmod_\bot (a) := (t\, a)^\nmod_\bot$.\footnote{Here we are implicitly using parameters from the model.}
	\end{itemize}

		Note that total functionals of type $\sigma$ are just elements of $\sigma^\nmod$, when restricted to non-$\bot$ arguments.
		Moreover, for any (finite) term $t$ we have immediately that $t^\nmod_\bot = t^\nmod$, when restricted to non-$\bot$ arguments.
		In light of this, we shall henceforth unambiguously write $t^\nmod$ rather than $t^\nmod_\bot$, and simply write $\nmod$ instead of $\nmod_\bot$.
\end{definition}

	We shall omit here the finer details of this interpretation of coterms as partial functionals, since we will give a more formal (and, indeed, formal\emph{ised}) treatment of `regular' coterms and their computational interpretations in Section~\ref{sect:coterm-models}.
	We point the reader to the excellent book \cite{ho-computability} for further details on models of (partial) (recursive) function(al)s.

\smallskip

Let us now consider some relevant examples of coderivations and coterms, at the same time establishing some foundational results. As before, the reader may safely ignore the colouring of type occurrences in what follows.
That will become meaningful later in the section.

\begin{example}
[Extensional completeness at type 1]
\label{ex:ext-comp}
    For \emph{any} $f : \Nat^k \to \Nat$, there is a coderivation $t:\nat^k \seqar \nat$ s.t.\ $t^\nmod = f$.
    To demonstrate this, we proceed by induction on $k$.\footnote{While we may assume $k=1$ WLoG by the availability of sequence (de)coding, the current argument is both more direct and avoids the use of cuts (on non-numerals).}
    If $k=0$ then the numerals clearly suffice.
    Otherwise, suppose $f: \Nat \times \Nat^k \to \Nat$ and write $f_n$ for the projection $\Nat^k \to \Nat$ by $f_n (\vec x) = f(n,\vec x)$.
    We define the coderivation for $f$ as follows:
    \begin{equation}
    \label{eqn:ext-comp}
     \vlderivation{
        \vliin{\cond}{}{\red{\nat}, \vec \nat \seqar \nat}{
            \vltr{f_0}{\vec \nat \seqar \nat}{\vlhy{\quad }}{\vlhy{}}{\vlhy{\quad }}
        }{
            \vliin{\cond}{}{\red{\nat}, \vec \nat \seqar \nat}{
            \vltr{f_1}{\vec \nat \seqar \nat}{\vlhy{\quad}}{\vlhy{}}{\vlhy{\quad}}
        }{
            \vliin{\cond}{}{\red{\nat}, \vec \nat \seqar \nat}{
            \vltr{f_2}{\vec \nat \seqar \nat}{\vlhy{\quad}}{\vlhy{}}{\vlhy{\quad}}
        }{
            \vlin{\cond}{}{\red{\nat}, \vec \nat \seqar \nat}{\vlhy{\vdots}}
        }
        }
        }
        }
    \end{equation}
where the derivations for each $f_n$ are obtained by the inductive hypothesis.
It is not difficult to see that the interpretation of this coderivation in the standard model indeed coincides with $f$.

Notice that, while we have extensional completeness at type 1, we cannot possibly have such a result for higher types by a cardinality argument: there are only continuum many coderivations. 
\end{example}

	\begin{example}
	[Na\"ive simulation of primitive recursion]
	\label{ex:sim-prim-rec}
	Terms of $\T$ may be interpreted as coterms in a straightforward manner.
	The only difficulty is the simulation of the $\rec$ combinators, which may be interpreted by coderivations as follows:
\begin{equation}
\label{eqn:sim-rec-naively-in-ct}
\vlderivation{
		\vliin{\rec}{}{\vec \sigma, \nat \seqar \sigma}{
			\vlhy{\vec \sigma \seqar \sigma}
		}{
			\vlhy{\vec \sigma, \nat, \sigma \seqar \sigma}
		}
	}
	\quad \leadsto \quad
	\vlderivation{
		\vliin{\cond}{\bullet}{\vec \sigma, \blue{\nat} \seqar \sigma}{
			\vlhy{\vec \sigma \seqar \sigma}
		}{
			\vliin{\cut}{}{\vec \sigma , \blue{\nat} \seqar \sigma}{
				\vlin{\cond}{\bullet}{\vec \sigma , \blue{\nat} \seqar \sigma}{\vlhy{\vdots}}
			}{
				\vlhy{\vec \sigma, \nat , \sigma \seqar \sigma}
			}
		}
	}
\end{equation}
	where the occurrences of $\bullet $ indicate roots of identical derivations.
	
	Denoting the RHS of \eqref{eqn:sim-rec-naively-in-ct} above as $\rec'$, we can check that the two sides of \eqref{eqn:sim-rec-naively-in-ct} are indeed equivalent.
	Formally, we show $\rec\, s\, t\, \vec x\, y \ = \ \rec \, s\, t\, \vec x\, y$ by induction on $y$:
	\[
	\begin{array}{rcll}
	\rec' s\, t\, \vec x\, 0 & = & \cond\, s\, (\cut\, (\rec' s\, t) \, t) \,   \vec x\, 0 & \text{by definition of $\rec'$ above} \\
		& = & s\, \vec x & \text{by $\cond$ axioms} \\
		& = & \rec\, s\, t\, \vec x\, 0 & \text{by $\rec$ axioms}
 \\ \noalign{\medskip}
	\rec'  s\, t\, \vec x\, \succ y &=& \cond\, s\, (\cut\, (\rec' s\, t) \, t)\, \vec x\, \succ y & \text{by definition of $\rec'$ above}\\
		& = & \cut\, (\rec' s\, t) \, t\, \vec x\, y  & \text{by $\cond$ axioms} \\
		& = & t\, \vec x\, y\, (\rec'\, s\, t\, \vec x\, y ) & \text{by $\cut$ axiom} \\
		& = & t\, \vec x\, y\, (\rec\, s\, t\, \vec x\, y) & \text{by inductive hypothesis}\\
		& = & \rec\, s\, t\, \vec x\, \succ y & \text{by $\rec$ axioms}
	\end{array}
	\]
	Note that our reasoning here was completely syntactic, indeed only using axioms and rules from Figures~\ref{fig:eq-ax-seq-calc-min}, \ref{fig:ax-t-non-eq} and \ref{fig:cond-rule+axioms}, (understanding metavariables $s,t$ etc.\ in those figures to now range over coterms as well as terms).
	This is no coincidence, and the argument above will actually turn out to be a formal proof in our theory defined later, thus inducing equivalence of $\rec$ and $\rec'$ in all models.
	\end{example}

\subsection{Regularity}

Until now, our coderivations and coterms were potentially non-uniform in structure and, as exemplified in Example~\ref{ex:ext-comp}, comprise a computational model of extreme expressivity.
Naturally, within formal theories, we would prefer to manipulate only finitely presentable objects. To this end we will study a natural fragment in non-wellfounded proof theory:

\begin{definition}
[Regular coderivations and coterms]
A coderivation $t$ is \textbf{regular} (or \textbf{circular}) if it has only finitely many distinct sub-coderivations.
A regular coterm is similarly just one with finitely many distinct sub-coterms.
\end{definition}

Note that a regular coderivation or coterm is indeed finitely presentable, e.g.\ as a finite directed graph, possibly with cycles, or a finite binary tree with `backpointers'.
When dealing with recursion-theoretic matters, we will implicitly assume such a finitary representation.
Further details on such a formalised representation are given in Section~\ref{sect:form-red-seq}.

Once again we have that regular coderivations are regular coterms, and conversely that closed regular coterms may be interpreted as regular coderivations (under $\cut$-as-composition).

One example we have already seen of a regular coderivation is the RHS of \eqref{eqn:sim-rec-naively-in-ct}.
In fact, it turns out that the regular coterms constitute a sufficiently expressive programming language:

\begin{proposition}
[Turing completeness]
\label{prop:turing-completeness}
    The set of regular coderivations of type level 1 is Turing-complete,\footnote{For a model of program execution, we may simply take the aforementioned Kleene-Herbrand-G\"odel model with \emph{equational derivability}, cf.~\cite{kleene:intro-to-metamath}. Note that this coincides with derivability by the axioms thus far presented.}
     i.e.\ $\{t^\nmod\  |\  \text{$t: \nat^k \seqar \nat$ regular}\}$ includes all partial recursive functions on $\Nat$.
\end{proposition}
\begin{proof}
    We have already seen in Example~\ref{ex:sim-prim-rec} that we can encode the primitive recursive functions, so it remains to simulate minimisation, i.e.\ the operation $\mu x (fx=0)$, for a given function $f$. 
    For this, we observe that $\mu x (fx =0)$ is equivalent to $H\, 0$ where 
    \begin{equation}
    \label{eqn:blind-search-equation}
    H\, x \  =\  \cond \ (f\, x)\ x\ (H\, \succ x)
    \end{equation}
    Note that $H$ may be interpreted by the following coderivation:
 \begin{equation}
 \label{eqn:blind-search-derivation}
     \vlderivation{
     \vliin{\cut}{\bullet}{\blue\nat \seqar \nat}{
         \vltr{f}{\blue\nat \seqar \red \nat}{\vlhy{}}{\vlhy{\quad}}{\vlhy{}}
         }{
         \vliin{\cond}{}{\underline{\red\nat}, \blue \nat \seqar \nat}{
             \vlin{\id}{}{\blue \nat\seqar \nat}{\vlhy{}}
             }{
             \vlin{\wk}{}{\underline{\red\nat} , \blue\nat \seqar \nat}{
             \vliin{\cut}{}{\blue \nat \seqar \nat}{
                 \vlin{\succ}{}{\blue \nat \seqar \purple \nat}{\vlhy{}}
                 }{
                 \vlin{\cut}{\bullet}{\purple \nat \seqar \nat}{\vlhy{\vdots}}
                 }
             }
             }
         }
     }
 \end{equation}
    It is intuitive here to think of the blue $\blue \nat$ standing for $x$, the red $\red\nat$ standing for $f(x)$, and the purple $\purple \nat$ standing for $\succ x$.
    
    Working in the standard model $\nmod$,  may show that $H$ indeed satisfies \eqref{eqn:blind-search-equation} as follows.
    We have that:
    \[
    \begin{array}{rcll}
    H\, x & = & \cut \, f\, (\cond\, \id\, (\wk\, (\cut\, \succ \, H)))\, x & \text{by definition of $H$}\\
    & = & \cond\, \id\, (\wk\, (\cut\, \succ \, H))\, (f\, x)\, x & \text{by $\cut$ axiom}\\
    \end{array}
    \]
    Now, we conduct a case analysis on the value of $f(x)$:
    \begin{itemize}
    \item If $f(x) = 0$, then we have $H\, x \ = \ \id\, x\ =\ x$, by the $\cond$ and $\id$ axioms, thus satisfying \eqref{eqn:blind-search-equation}.
    \item If $f(x) = \succ y$, for some $y$, then,
    \[
    \begin{array}{rcll}
    H\, x & = & \wk\, (\cut\, \succ\, H)\, y\, x & \text{by $\cond$ axioms}\\
    & = & \cut\, \succ\, H\, x  &\text{by $\wk$ axiom}\\
    & = & H\, \succ x & \text{by $\cut $ axiom}
    \end{array}
    \] 
    again satisfying \eqref{eqn:blind-search-equation}. \qedhere
    \end{itemize}
\end{proof}

\begin{remark}
[Reasoning over partial functionals]
\label{rmk:reasoning-partial-functionals}
Note, again, that the reasoning above was entirely syntactic, using only axioms thus far presented in Figures~\ref{fig:eq-ax-seq-calc-min}, \ref{fig:rec-axioms},  \ref{fig:ax-t-non-eq} and \ref{fig:cond-rule+axioms}.
    While the coderivation in \eqref{eqn:ack-cyc-der} will not formally be a symbol of our eventual theory $\C$, 
    the reasoning above hints at well-behaved extensions accommodating partially defined coterms.
\end{remark}
\begin{remark}
[Turing completeness at level 0]
\label{rmk:type-0-turing-completeness}
Note that the argument above required coderivations including only occurrences of $\nat$.
This means that the set of type 0 regular coderivations are already a Turing-complete programming language, under derivability via Figures~\ref{fig:eq-ax-seq-calc-min}, \ref{fig:rec-axioms},  \ref{fig:ax-t-non-eq} and \ref{fig:cond-rule+axioms}.
\end{remark}

\subsection{The progressing criterion}
Despite regular coderivations being finitely presentable, they do not necessarily denote totally defined functionals in the standard model $\nmod$, cf.~Proposition~\ref{prop:turing-completeness}, contrary to the norm for terms in formal theories.
In this work we will consider coderivations satisfying a standard `termination criterion' in non-wellfounded proof theory.
First, let us recall some standard structural proof theoretic concepts about (co)derivations.

\begin{definition}
[Immediate ancestry]
\label{dfn:ancestry}
Let $t $ be a (co)derivation. A type occurrence $\sigma^1$ is an \textbf{immediate ancestor}\footnote{This terminology is standard in proof theory, e.g.\ as in \cite{Bus98:handbook-of-pt}.} of a type occurrence $\sigma^2$ in $t$ if $\sigma^1$ and $\sigma^2$ appear in the LHSs of a premiss and conclusion, respectively, of a rule instance and have the same colour in the corresponding rule typeset in Figure~\ref{fig:seq-calc-min}, \ref{fig:rules-0-s-rec} or \ref{fig:cond-rule+axioms}. If $\sigma^1$ and $\sigma^2$ are elements of an indicated list, say $\vec \sigma$, we also require that they are at the same position of the list in the premiss and the conclusion.
Note that, if $\sigma^1$ is an immediate ancestor of $\sigma^2$, they are necessarily occurrences of the same type.
\end{definition}

The notion of immediate ancestor thus defined, being a binary relation, induces a directed graph whose paths will form the basis of our termination criterion.

\begin{definition}
[Threads and progress]
A \textbf{thread} is a maximal path in the graph of immediate ancestry.
A \textbf{$\sigma$-thread} is a thread whose elements are occurrences of the type $\sigma$.
We say that a $\nat$-thread \textbf{progresses} when it is principal for a $\cond$ step (i.e.\ it is the indicated blue $\blue \nat$ in the $\cond$ rule typeset in Figure~\ref{fig:cond-rule+axioms}).
A (infinitely) \textbf{progressing} thread is a $\nat$-thread that progresses infinitely often (i.e.\ it is infinitely often the indicated blue $\blue \nat$ in the $\cond$ rule typeset in Figure~\ref{fig:cond-rule+axioms}.)

A coderivation is \textbf{progressing} if every infinite branch has a progressing thread.
\end{definition}

Note that progressing threads do not necessarily begin at the root of a coderivation, they may begin arbitrarily far into a branch.
In this way, the progressing coderivations are closed under all typing rules.
Note also that arbitrary coderivations may be progressing, not only the regular ones.

\begin{example}
[Extensional completeness at type $1$, revisited]
\label{ex:ext-comp-revisited}
Recalling Example~\ref{ex:ext-comp}, note that the infinite branch marked $\cdots$ in \eqref{eqn:ext-comp} has a progressing thread along the red $\red \nat$s.
Other infinite branches, say through $f_0,f_1, $ etc., will have progressing threads along their infinite branches by an appropriate inductive hypothesis, though these may progress for the first time arbitraryily far from the root of \eqref{eqn:ext-comp}.
\end{example}

As previously mentioned, we shall focus our attention in this work on the regular coderivations.
Let us take a moment to appreciate some previous (non-)examples of regular coderivations with respect to the progressing criterion.

\begin{example}
[Primitive recursion, revisited]
\label{ex:prim-rec-revisited}
Recalling Example~\ref{ex:sim-prim-rec}, notice that the RHS of \eqref{eqn:sim-rec-naively-in-ct} is a progressing coderivation: there is precisely one infinite branch (that loops on $\bullet$) and it has a progressing thread on the blue $\blue \nat$ indicated there.
\end{example}

\begin{example}
[Turing completeness, revisited]
\label{ex:tur-comp-revisited}
Recalling the proof of Proposition~\ref{prop:turing-completeness}, notice that the coderivation given for $H$ in \eqref{eqn:blind-search-derivation} is \emph{not} progressing: the only infinite branch loops on $\bullet$ and immediate ancestry, as indicated by the colouring, admits no thread along the $\bullet$-loop. 
This is no coincidence, as it turns out that the progressing criterion suffices for coderivations to denote \emph{total} functionals in the standard model $\nmod$, as we will show in the next subsection.
\end{example}

One of the most appealing features of the progressing criterion is that, while being rather expressive and admitting many natural programs, e.g.\ as we will see in Section~\ref{sect:ack-cyclic}, it remains effective (for regular coderivations) thanks to well known arguments in automaton theory:

\begin{fact}
[Folklore]
\label{fact:decidability-of-progressivity}
It is decidable whether a regular coderivation is progressing.
\end{fact}
This well-known result (see, e.g., \cite{DaxHofLang06} for an exposition for a similar circular system) follows from the fact that the progressing criterion is equivalent to the universality of a B\"uchi automaton of size determined by the (finite) representation of the input coderivation.
This problem is decidable in polynomial space, though the correctness of this algorithm requires nontrivial infinitary combinatorics, as formally demonstrated in \cite{KMPS16:buchi-reverse}.
Nonetheless, a non-uniform version of this problem is formalisable in the weakest of the big-five theories of reverse mathematics:
\begin{proposition}
[\cite{Das19:log-comp-cyc-arith}]
\label{prop:prog-decidability+provability-in-RCA}
For any regular progressing coderivation $t$, $\RCA$ proves that $t$ is progressing.
\end{proposition}

As noted in that work, the above result cannot be strengthened to a uniform one unless $\RCA$ (and so PRA) is inconsistent, by a reduction to G\"odel-incompletness.

\subsection{Progressing coterms denote total functionals}

As outlined in Definition~\ref{dfn:int-of-coterms-as-partial-fns}, coderivations denote partial functionals in the standard model $\nmod$.
In fact, the partial functionals induced by progressing coderivations are indeed totally defined, by adapting well-known infinite descent arguments in non-wellfounded proof theory:
\begin{proposition}
\label{prop:termination}
	If $t:\vec \sigma \seqar \tau$ is a progressing coderivation, then $t^\nmod$ is a well-defined total functional in $(\vec \sigma\to \tau)^\nmod$.
\end{proposition}

The idea behind this result is to, by contradiction, assume a non-terminating `run' of a progressing coderivation, and thence extract an infintely decreasing sequence of natural numbers from a progressing thread, contradicting the well-ordering property.
We stop short of giving an explicit `operational semantics' here, being beyond the scope of this work. Rather, let us simply note that the totally defined functionals are closed under composition by typing rules (since typing rules are constants interpreted as totally defined functionals themselves). Contrapositively this means that if a coderivation is interpreted by a non-total functional, then so is one of its immediate sub-coderivations.
\begin{proof}
[Proof of Proposition~\ref{prop:termination}]
	Suppose otherwise and let $\vec a \in \vec \sigma^\nmod$ be inputs on which $t^\nmod$ is not well-defined, i.e.\ $t^\nmod(\vec a) = \bot$.
	We may thus inductively construct an infinite branch $(t_i:\vec \sigma_i \seqar \tau_i)_{i \in \omega}$ and associated inputs $(\vec a_i \in \vec \sigma_i^\nmod)$ s.t.\ $t_i^\nmod (\vec a_i) = \bot$ as follows:
	\begin{itemize}
		\item $t_0 = t$ and $\vec a_0 = \vec a$.
		\item If $t_i$ concludes with a $\wk$, $\exch$ or $\cntr$ step then $t_{i+1}$ is the only immediate sub-coderivation and $\vec a_{i+1}$ is just $\vec a_i$ with the appropriate deletion, switch or duplication of arguments.
		\item ($t_i$ cannot conclude with a nullary step $\id$, $0$ or $\succ$, by assumption that $t_i$ is non-total.)
		\item If $t_i$ concludes with a $\cut $ step, as typeset in Figure~\ref{fig:seq-calc-min}, then $t_{i+1}$ is the left sub-coderivation if it is not totally defined on inputs $\vec a_i$; otherwise $t_{i+1}$ is the right sub-coderivation and $\vec a_{i+1} = (\vec a, a)$, for some $a\in \sigma^\nmod$ s.t.\ $t_{i+1}^\nmod (\vec a_i,a) = \bot$. 
		\item If $t_i$ concludes with a $\leftimp $ step, as typeset in Figure~\ref{fig:seq-calc-min}, then $t_{i+1}$ is the left sub-coderivation if it is not totally defined on inputs $\vec a_i$; otherwise $t_{i+1}$ is the right sub-coderivation and $\vec a_{i+1} = (\vec a_i, a)$, for some $a \in \sigma^\nmod$ s.t.\ $t_{i+1}^\nmod (\vec a_i,a) = \bot$.
\item If $t_i$ concludes with a $\rightimp $ step, as typeset in Figure~\ref{fig:seq-calc-min}, then $t_{i+1}$ is the only immediate sub-coderivation and $\vec a_{i+1} = (\vec a_i , a)$, for some $a \in \sigma^\nmod$ s.t.\ $t_{i+1}^\nmod (\vec a_i, a) = \bot$.
\item If $t_i$ concludes with a $\cond$ step and $\vec a_i = (\vec a_i', n)$, then $t_{i+1}$ is the left sub-coderivation if $n=0$ and $\vec a_{i+1} = \vec a_i'$; otherwise, if $n=m+1$, $t_{i+1}$ is the right sub-coderivation and $\vec a_{i+1} = (\vec a_i', m)$.
	\end{itemize}
	
	Now, notice that, since $t$ is progressing, we must have some progressing thread $(\nat^i)_{i\geq k}$ along some tail $(t_i)_{i\geq k}$.
	Writing $n_i$ for the input in $\vec a_i$ corresponding to $N^i$, notice that $(n_i)_{i\geq k}$ is a non-increasing sequence of natural numbers, by construction of $t_i$ and $\vec a_i$.
	Moreover, we have that $n_{i+1}<n_i$ whenever $\nat^i$ is principal for a $\cond$ step, so for infinitely many $i\geq k$ by definition of a progressing thread.
	Thus $(n_i)_{i\geq k}$ has no least element, contradicting the well-ordering property.
\end{proof}

\subsection{The simply typed theory $\C$ and its fragments}
We are finally ready to give the definition of our circular version of System $\T$.

\begin{definition}
[Circular version of $\T$]
The language of
$\C$ extends contains every regular progressing coderivation of $\T^- + \cond$ as a symbol.
We identify `terms' of this language with coterms in the obvious way, and call them (regular) \textbf{progressing coterms}.
$\C$ itself is a STT axiomatised by
the schemata from
 Figures~\ref{fig:eq-ax-seq-calc-min}, \ref{fig:ax-t-non-eq} and \ref{fig:cond-rule+axioms}, now interpreting the metavariables $s,t$ etc.\ there as ranging over coterms.
\end{definition}

The aim of this work is to compare fragments of $\C$ and fragments of $\T$ delineated by type level.
In light of Proposition~\ref{prop:free-cut-elim}, the following definition gives natural circular counterparts of the fragments $\nT n$ of $\T$:
\begin{definition}
[Type level restricted fragments of $\C$]
	$\nC n$ is the fragment of $\C$ restricted to the language containing only coderivations where all types occurring have level $\leq n$.
	$\nC n$ still has constant symbols for each individual constant of $\T^- + \cond$.
\end{definition}

Notice that, despite the fact that coderivations of $\nC n$ may type only level $n+1$ functionals, progressing coterms are closed under application and so, by definition of a STT, $\nC n$ admits `terms' of arbitrary type by composing with the constants and variables of $\T^- + \cond$.
In particular we still have combinatory completeness, as usual.
In what follows, however, it will usually suffice to only consider the coderivations when proving properties of $\nC n$, the generalisation to progressing coterms following by closure under application.

\subsection{Example: Ackermann-P\'eter, revisited}
\label{sect:ack-cyclic}
Let us revisit the example of the Ackermann-P\'eter function from Section~\ref{sect:ack-type-1-rec}.
	Despite the fact that type level 1 recursion is required to type it in $\T$ (cf.~Fact~\ref{fact:type-0-rec-prim-rec}), $A$ may be circularly typed in $\C$ by a coderivation $\mathrm A$ using only the base type $\nat$:\footnote{For convenience we have implemented some branching rules as context-splitting, namely the $\cut$ steps. Formally, there are implicit $\wk$ steps that are not indicated, a convention that we will henceforth adopt for the sake of easing legibility.}
\begin{equation}
\label{eqn:ack-cyc-der}
\vlderivation{
		\vlin{\cntr}{\bullet}{\underline{\purple \nat}, \nat \seqar \nat}{
		\vliin{\cond}{}{\underline{\red \nat} , \blue \nat , \orange \nat \seqar \nat}{
			\vlin{\wk}{}{\underline \nat , \nat \seqar \nat}{
			\vlin{\succ}{}{\nat \seqar \nat}{\vlhy{}}
			}
		}{
			\vliin{\cond}{}{\red \nat, \blue \nat , \underline{ \orange{\nat}} \seqar \nat}{
				\vlin{\wk}{}{\red{\nat}, \underline \nat \seqar \nat}{
				\vliin{\cut}{}{\red \nat \seqar \nat}{
					\vliq{1}{}{\seqar \nat}{\vlhy{}}
				}{
					\vlin{}{\bullet}{\red{\nat}, \nat \seqar \nat}{\vlhy{\vdots (1)}}
				}
				}
			}{
				\vliin{\cut}{}{\red \nat, \blue \nat , \orange \nat \seqar \nat}{
					\vlin{}{\bullet}{\blue \nat,  \orange \nat \seqar \nat}{\vlhy{\vdots (2)}}
				}{
					\vlin{}{\bullet}{\red \nat, \nat \seqar \nat}{\vlhy{\vdots (3)}}
				}
			}
		}
		}	
	}
\end{equation}
	As usual, the occurrences of $\bullet$ above indicate roots of identical coderivations, and we have indicated three distinct threads, coloured red, blue and orange. Note that the purple $\purple \nat$ prefixes both the red and the blue thread.
	Finally note that the red $\red \nat$ thread progresses on every visit to (1) and (3), while the orange $\orange \nat$ thread progresses on every visit to (2).

\begin{proposition}
$\mathrm A$ is progressing and regular, and so is a symbol of $\nC 0$.
\end{proposition}
	\begin{proof}
	To show that $\mathrm A$ is progressing, we conduct a case analysis on an infinite branch $B$, based on which of the simple loops $(1),(2)$ and $(3)$ are traversed infinitely often:
		\begin{itemize}
		\item $B$ hits only $(1)$ infinitely often. Then there is a progressing thread along the red $\red \nat$.
		\item $B$ hits only $(2)$ infinitely often. Then eventually there is a progressing thread along the orange $\orange \nat$.
		\item $B$ hits only $(3)$ infinitely often. Then eventually there is a progressing thread along the red $\red \nat$.
		\item $B$ hits only (1) and (2) infinitely often. Then eventually there is a progressing thread along the red $\red \nat$ on iterations of (1) (on which it progresses) and along the blue $\blue \nat$ on iterations of (2) (on which it is constant).
		\item $B$ hits only (1) and (3) infinitely often. Then eventually there is a progressing thread along the red $\red \nat$, which progresses on any iteration of (1) or (3).
		\item $B$ hits only (2) and (3) infinitely often. Then eventually there is a progressing thread along the blue $\blue \nat$ on iterations of (2) (on which it is constant) and along the red $\red \nat$ on iterations of (3) (on which it progresses).
		\item $B$ hits all of (1), (2) and (3) infinitely often. Then there is a progressing thread along the red $\red \nat$ on iterations of (1) and (3) (on which it progresses) and the blue $\blue \nat $ on iterations of (2) (on which it is constant).
		\end{itemize}
	Clearly $\mathrm A$ is regular and contains only occurrences of $\nat$, so $\mathrm A$ is indeed a symbol of $\nC 0$.
	\end{proof}

In fact, we may also show that $\nC 0 $ proves the defining equations of $A$ from \eqref{eqn:ackermann-equations}:

\begin{proposition}
\label{prop:ct0-proves-ack-eqns}
$\nC 0 $ proves the following equations:
\[
\begin{array}{rcl}
\mathrm A\, 0\, y & = & \succ y \\
\mathrm A\, \succ x \, 0 & = & \mathrm A\, x\, \numeral 1 \\
\mathrm A\, \succ x\, \succ y & = &  \mathrm A\, x\, (\mathrm A\, \succ x \, y)
\end{array}
\]
\end{proposition}
\begin{proof}
Writing $\mathrm A_0$ and $\mathrm A_1$ for the left and right coderivations, respectively, composed by the lowermost $\cond$ step in \eqref{eqn:ack-cyc-der}, we have:
\begin{equation}
\label{eqn:ackxy-to-cond}
\begin{array}{rcll}
\mathrm A \, x\, y & = & \cntr\, (\cond\, \mathrm A_0 \, \mathrm A_1)\, x \, y & \text{by definition of $\mathrm A$}\\
& = & \cond\, \mathrm A_0\, \mathrm A_1\, x\, x\, y & \text{by $\cond$ axioms}
\end{array}
\end{equation}
From here we obtain the first equation of \eqref{eqn:ackermann-equations}, in $\nC 0$:
\[
\begin{array}{rcll}
\mathrm A\, 0 \, y & = & \mathrm A_0 \, x\, y & \text{by \eqref{eqn:ackxy-to-cond} above and $\cond$ axioms}\\
& = & \wk \, \succ\, x\, y & \text{by definition of $\mathrm A_0$}\\
& = & \succ y & \text{by $\wk$ axiom}
\end{array}
\]

Now writing $\mathrm A_{10}$ and $\mathrm A_{11}$ for the left and right coderivations, respectively, composed by the uppermost $\cond$ step typeset in \eqref{eqn:ack-cyc-der}, we have:
\begin{equation}
\label{eqn:acksxy-to-cond}
\begin{array}{rcll}
\mathrm A\, \succ x\, y & = & \mathrm A_1 \, x\, \succ x\, y & \text{by \eqref{eqn:ackxy-to-cond} above and $\cond$ axioms}\\
& = & \cond\, \mathrm A_{10} \, \mathrm A_{11}\, x\, \succ x\, y & \text{by definition of $\mathrm A_1$}
\end{array}
\end{equation}
From here we obtain the second and third equations of \eqref{eqn:ackermann-equations}, in $\nC 0$:
\[
\begin{array}{rcll}
\mathrm A\, \succ x\, 0 & = & \mathrm A_{10}\, x\, \succ x & \text{by \eqref{eqn:acksxy-to-cond} above and $\cond$ axioms}\\
& = & \wk\, (\cut\, (1\, \mathrm A))\, x\, \succ x & \text{by definition of $\mathrm A_{10}$}\\
& = & \cut\, (1\, \mathrm A)\, x & \text{by $\wk$ axiom}\\
& = & \mathrm A\, x\, \numeral 1 & \text{by $\cut$ axiom}\\
\noalign{\medskip}
\mathrm A\, \succ x\, \succ y & = & \mathrm A_{11}\, x\, \succ x\, y & \text{by \eqref{eqn:acksxy-to-cond} above and $\cond$ axioms}\\
& = & \cut\, (\mathrm A\, \mathrm A)\, x\, \succ x\, y & \text{by definition of $\mathrm A_{11}$}\\
& = & \mathrm A\, x\, (\mathrm A\, \succ x\, y) & \text{by $\cut $ axiom} \qedhere
\end{array}
\]
\end{proof}

Now, let us revisit Section~\ref{sect:ack-type-1-rec}, where we gave a $\T$ term for the Ackermann-P\'eter function based on type 1 recursion. 
Calling that term $\mathrm A'$, temporarily, we now have that $\mathrm A'$ and $\mathrm A$ satisfy the same defining equations from \eqref{eqn:ackermann-equations}.
Working in a combined theory, we may thus deduce their equivalence via a nested induction.

\begin{proposition}
$\nC 0 , \rec_1 \proves\,  \mathrm A\, x\, y \, = \, \mathrm A'x\, y$. 
\end{proposition}
\begin{proof}
Working inside $\nC 0 + \rec_1$, we show $\mathrm A\, x\, = \, \mathrm A' x$ by induction on $x$. We have,
\[
\begin{array}{rcll}
\mathrm A\, 0 \, y & = & \succ y & \text{by Proposition~\ref{prop:ct0-proves-ack-eqns}}\\
& = & \mathrm A' 0\, y & \text{by Proposition~\ref{prop:t1-proves-ack-eqns}}
\end{array}
\]
so $\mathrm A\, 0 = \mathrm A' 0$ by $\extensionality$.
For the inductive step, we show that $\mathrm A\, \succ x\, y = \mathrm A' \succ x\, y$ by a sub-induction on $y$:
\[
\begin{array}{rcll}
\mathrm A\, \succ x\, 0 & = & \mathrm  A\, x\, \numeral 1 & \text{by Proposition~\ref{prop:ct0-proves-ack-eqns}}\\
 & = & \mathrm A'x\, \numeral 1 & \text{by main inductive hypothesis}\\
 & = & \mathrm A'\succ x\, 0 & \text{by Proposition~\ref{prop:t1-proves-ack-eqns}}\\
 \noalign{\medskip}
 \mathrm A\, \succ x\, \succ y & = & \mathrm A\, x\, (\mathrm A\, \succ x\, y) & \text{by Proposition~\ref{prop:ct0-proves-ack-eqns}}\\
 & = & \mathrm A' x\, (\mathrm A\, \succ x\, y) & \text{by main induction hypothesis}\\
 & = & \mathrm A'x\, (\mathrm A' \succ x\, y) & \text{by sub-induction hypothesis}\\
 & = & \mathrm A' \succ x\, \succ y & \text{by Proposition~\ref{prop:t1-proves-ack-eqns}}
\end{array}
\]
Thus we have $\mathrm A\, \succ x\, = \, \mathrm A' \succ x$ by $\extensionality$, proving the main inductive step as required.
\end{proof}

As we mentioned before, general extensionality is not strictly necessary for many of our results, and this is in particular the case for the result above.
We could have also proceeded under \emph{weak extensionality} by an instance of $\Pi_1$-induction ($\forall y (\mathrm A\, x\, y \, = \, \mathrm A' x\, y)$). Standard witnessing theorems then reduce this to an instance of quantifier-free induction, but such a development is beyond the scope of this work.

\section{$\C$ simulates $\T$, more succinctly}
\label{sect:c-sim-t}

In this section we give a simulation of $\nT{n+1}$ terms of level $\leq n+1$ into $\nC n$, provably satisfying the same equational theory, thus improving on the naive simulation of primitive recursion from Examples~\ref{ex:sim-prim-rec} and \ref{ex:prim-rec-revisited}.
This matches similar results from the setting of arithmetic, \cite{Das19:log-comp-cyc-arith}, and the exposition is entirely proof theoretic.

The main goal is to show the following result:

\begin{theorem}
\label{thm:c-sim-t}
If $t:\sigma$ is a term of $\nT{n+1}$ with $\level(\sigma)\leq n+1$ then there is a $\nC n$ coderivation $t' : \sigma$ s.t.:
\begin{equation}
\label{eqn:prov-eq-t-ct-terms-with-ext}
\nC n, \rec_{n+1} \proves t=t'
\end{equation}
\end{theorem}

The idea is to rely on the partial normalisation result, Proposition~\ref{prop:free-cut-elim}, to work with a normal form of $\T$ derivations, and then translate into $\C$ coderivations under a constructive realisation of the deduction theorem (for typing derivations).

\subsection{Derivations with `oracles'}

We consider (co)derivations with fresh intial sequents of the form $\vlinf{f}{}{\sigma_1, \dots, \sigma_n \seqar \tau}{}$, denoting a functional $f$ of type $\sigma_1 \to \cdots \to \sigma_n \to \tau$, as expected.
We write,
\[
\toks0={0.7}
\vlderivation{
\vltrf{t}{\vec \sigma \seqar \tau}{\vlhy{}}{
	\vlin{f_i}{i}{\vec \sigma_i \seqar \tau_i}{\vlhy{}}
}{\vlhy{}}{\the\toks0}
}
\]
for a (co)derivation $t$ of $\vec \sigma \seqar \tau$ with initial sequents among $\vlinf{f_i}{i}{\vec \sigma_i \seqar \tau_i}{}$, with $i$ varying over some fixed range.
We distinguish $f_i$ from variable symbols, since coderivations until now are closed coterms; instead it is more pertinent to think of them as some fresh constant `oracle' symbols.
We thus use metavariables $f,g,$ etc.\ for these new initial sequents.
The interpretation of such (co)derivations into the standard model $\nmod$ is as expected, with (co)derivations now computing (partial) functionals with respect to oracles for each $f_i$.

\subsection{Constructive realisation of the deduction theorem}
We describe how to `realise' a version of the \emph{deduction theorem} for typing derivations in order to lower type level. 
The key feature of this translation is that recursion is simulated by a cyclic derivation in a succinct way, in terms of abstration complexity.

As notation throughout this section, if $\vec f = f_1, \dots, f_n$ then we may simply write $(\vec f\, \vec x)$ for $(f_1\, \vec x)\, (f_2\, \vec x)\, \dots\, (f_n\, \vec x)$, to lighten the syntax.
Our key intermediate result is the following lemma:

\begin{lemma}
\label{lem:ded-thm-realised}
Let $t:\vec \sigma , \vec \nat \seqar \tau$ be  a $\nT{n+1}$ derivation, i.e.,
\[
\vltreeder{t}{\vec \sigma, \vec \nat \seqar \tau}{\ }{\quad }{\ }
\]
with all $\nat$s indicated, s.t.\ all types occurring in $t$ have level $\leq n+1$.

Write $\tau = \vec \tau \to \nat$ and $\sigma_i = \vec \sigma_i \to \nat$. For each $\vec \rho$ of levels $\leq n$ there is a $\nC n$ coderivation $\rhotrans t {\vec\rho} (\vec f) : \vec \rho, \vec \nat , \vec \tau \seqar \nat$ using initial sequents $\vlinf{f_i}{i}{\vec \rho, \vec \sigma_i \seqar \nat}{}$, i.e.,
\[
\toks0={0.7}
\vlderivation{
	\vltrf{ \rhotrans{ t }{\vec \rho}   (\vec f) }{\vec \rho, \vec \nat , \vec \tau \seqar \nat}{\vlhy{}}{
		\vlin{f_i}{i}{\vec \rho,\vec \sigma_i \seqar \nat}{\vlhy{}}
	}{\vlhy{}}{\the\toks0}
}
\]
such that,
\begin{equation}
\label{eqn:ded-thm-invariant}
\nC n, \rec_{n+1} \proves \rhotrans t {\vec \rho}(\vec f)\, \vec x\, \vec y\, \vec z \  = \ t\, (\vec f\, \vec x)\, \vec y\, \vec z
\end{equation}
Moreover, for any $j$, there is a $\rho_j$-thread from the $\rho_j$ in the conclusion of $\rhotrans t {\vec \rho} (\vec f)$ to the $\rho_j$ in any occurrence of the initial sequent $f_i$. 
\end{lemma}

Before giving the proof let us set up some further notation to lighten the exposition as much as possible.
\begin{itemize}
\item We shall sometimes suppress the initial sequent arguments of a coderivation when it is unambiguous, e.g.\ writing $\rhotrans t {\vec \rho}$ instead of $\rhotrans t{\vec \rho}(\vec f)$ etc. We will only do this when the initial sequents, or approrpriate substituted (co)derivations, are explicitly typeset.
\item We shall write $\mathsf r^*$ for the reflexive transitive closure of a rule $\mathsf r$.
\item As in the statement of the lemma, we shall typically assume that a type, say, $\tau$ has the form $\vec \tau \to \nat$ and so on.
\item Variables $\vec x,x$ will typically correspond to types $\vec \rho, \rho$, variables $\vec y, y$ to $\vec \nat, \nat$, and $\vec z\, z$ to $\vec \tau, \tau$.
\end{itemize}

\begin{proof}
[Proof of Lemma~\ref{lem:ded-thm-realised}]
We proceed by induction on the structure of $t$, and refer to an inductive hypothesis for a smaller derivation $s$ by $\IH(s)$.

If $t$ is the initial sequent $\vlinf{\id}{}{\sigma \seqar \sigma}{}$ then:
\begin{itemize}
\item If $\sigma = \nat$, then $\rhotrans t{\vec \rho} $ is:
\[
\vlderivation{
	\vliq{\wk}{}{\vec \rho, \nat \seqar \nat}{
	\vlin{\id}{}{\nat \seqar \nat}{\vlhy{}}
	}
}
\]
There are no new initial sequents so also no thread condition to check. 
To verify \eqref{eqn:ded-thm-invariant}, we have:
\[
\begin{array}{rcll}
\rhotrans t {\vec \rho}\, \vec x\, y & = & \wk^*\, \id\, \vec x\, y & \text{by definition of $\rhotrans t{\vec \rho}$}\\
& = & \id\, y & \text{by $\wk$ axioms}\\
& = & t\, y & \text{by definition of $t$}
\end{array}
\]
\item Otherwise $\rhotrans t {\vec \rho}$ is just: 
\[
\vlinf{f}{}{\vec \rho , \vec \sigma \seqar \nat }{}
\] 
The required threading property is immediate, and to verify \eqref{eqn:ded-thm-invariant} we have:
\[
\begin{array}{rcll}
\rhotrans t{\vec \rho}(f)\, \vec x\, \vec y & = & f\, \vec x\, \vec y & \text{by definition of $\rhotrans t{\vec \rho}$. }\\
& = & \id\, (f\, \vec x)\, \vec y & \text{by $\id$ axiom} \\
& = & t\, (f\, \vec x)\, \vec y & \text{by definition of $t$}
\end{array}
\]
\end{itemize}

If $t$ is the initial sequent $\vlinf{0}{}{\seqar \nat}{}$, then $\rhotrans t{\vec \rho}$ is:
\[
\vlderivation{
	\vliq{\wk}{}{\vec \rho \seqar \nat}{
	\vlin{0}{}{\seqar \nat}{}
	}
}
\]
There are no new initial sequents, so no threading property to check.
To verify \eqref{eqn:ded-thm-invariant} we have:
\[
\begin{array}{rcll}
\rhotrans t{\vec \rho}\, \vec x & = & \wk^*\, 0\, \vec x &\text{by definition of $\rhotrans t{\vec \rho}$}\\
 & = & 0 & \text{by $\wk$ axioms}\\
 & = & t & \text{by definition of $t$}
\end{array}
\]

If $t$ is the initial sequent $\vlinf{\succ}{}{\nat \seqar \nat}{}$ then $\rhotrans t{\vec \rho}$ is:
\[
\vlderivation{
	\vliq{\wk}{}{\vec \rho, \nat \seqar \nat}{
	\vlin{\succ}{}{\nat \seqar \nat}{}
	}
}
\]
There are no new initial sequents, so no threading property to check.
To verify \eqref{eqn:ded-thm-invariant} we have:
\[
\begin{array}{rcll}
\rhotrans t{\vec \rho}\, \vec x\, y & = & \wk^*\, \succ \, \vec x\, y &\text{by definition of $\rhotrans t{\vec \rho}$}\\
 & = & \succ\, y & \text{by $\wk$ axioms}\\
 & = & t\, y & \text{by definition of $t$}
\end{array}
\]

If $t $ concludes with a weakening step,
\[
\vlderivation{
	\vlin{\wk}{}{\vec \sigma , \pi, \vec \nat\seqar \tau}{
	\vltr{s}{\vec \sigma, \vec \nat \seqar \tau}{\vlhy{\  }}{\vlhy{\ }}{\vlhy{\ }}
	}
}
\]
then:
\begin{itemize}
\item If $\pi = \nat$ then we define $\rhotrans t{\vec \rho} (\vec f)$ by just commuting with the weakening step:
\[
\toks0={.8}
\vlderivation{
	\vlin{\wk}{}{\vec \rho, \vec \sigma, \nat , \vec \nat , \vec \tau \seqar \nat}{
	\vltrf{\rhotrans s{\vec \rho}(\vec f)}{\vec \rho, \vec \sigma, \vec \nat, \vec \tau \seqar \nat}{\vlhy{\quad }}{\vlhy{\quad}}{\vlhy{\quad}}{\the\toks0}
	}
}
\]
The threading property is readily obtained from the inductive hypothesis, and we verify \eqref{eqn:ded-thm-invariant} as follows:
\[
\begin{array}{rcll}
\rhotrans t {\vec \rho} (\vec f)\, \vec x\, y\, \vec y\, \vec z & = & \wk\, \rhotrans s {\vec \rho} (\vec f)\, \vec x\, y\, \vec y\, \vec z & \text{by definition of $\rhotrans t{\vec \rho}$}\\
& = & \rhotrans s{\vec \rho}(\vec f)\, \vec x\, \vec y\, \vec z & \text{by $\wk$ axiom} \\
& = & s\, (\vec f\, \vec x)\, \vec y\, \vec z & \text{by $\IH(s)$}\\
& = & \wk\, s\, (\vec f\, \vec x)\, y\, \vec y\, \vec z & \text{by $\wk$ axiom}\\
& = & t\, (\vec f\, \vec x)\, y \, \vec y \, \vec z & \text{by definition of $t$}
\end{array}
\]
\item Otherwise, $\rhotrans t{\vec \rho} (\vec f,g)$ is simply just $\rhotrans s{\vec \rho} (\vec f)$ (the initial sequents $f:\vec \pi\seqar \nat$ are never used), and the required properties are inherited directly from the inductive hypothesis.
\end{itemize}

If $t$ concludes with a contraction step,
\[
\vlderivation{
	\vlin{\cntr}{}{\vec \sigma, \pi, \vec \nat \seqar \tau}{
	\vltr{s}{\vec \sigma, \pi, \pi, \vec \nat \seqar \tau}{\vlhy{\  }}{\vlhy{\ }}{\vlhy{\ }}
	}
}
\]
then:
\begin{itemize}
\item If $\pi = \nat$ then we define $\rhotrans t{\vec \rho}(\vec f)$ by just commuting with the contraction step:
\[
\toks0={.8}
\vlderivation{
	\vlin{\cntr}{}{\vec \rho, \vec \sigma, \nat , \vec \nat , \vec \tau \seqar \nat}{
	\vltrf{\rhotrans s{\vec \rho}(\vec f)}{\vec \rho, \vec \sigma, \nat, \nat , \vec \nat, \vec \tau \seqar \nat}{\vlhy{\quad }}{\vlhy{\quad}}{\vlhy{\quad}}{\the\toks0}
	}
}
\]
The threading property is readily obtained from the inductive hypothesis, and we verify \eqref{eqn:ded-thm-invariant} as follows:
\[
\begin{array}{rcll}
\rhotrans t {\vec \rho} (\vec f)\, \vec x\, y\, \vec y\, \vec z & = & \cntr\, \rhotrans s {\vec \rho} (\vec f)\, \vec x\, y\, \vec y\, \vec z & \text{by definition of $\rhotrans t{\vec \rho}$}\\
& = & \rhotrans s{\vec \rho}(\vec f)\, \vec x\, y\, y\, \vec y\, \vec z & \text{by $\cntr$ axiom} \\
& = & s\, (\vec f\, \vec x)\, y\, y\, \vec y\, \vec z & \text{by $\IH(s)$}\\
& = & \cntr\, s\, (\vec f\, \vec x)\, y\, \vec y\, \vec z & \text{by $\cntr$ axiom}\\
& = & t\, (\vec f\, \vec x)\, y \, \vec y \, \vec z & \text{by definition of $t$}
\end{array}
\]
\item Otherwise, $\rhotrans t{\vec \rho} (\vec f, g)$ is just $\rhotrans{s}{\vec \rho}(\vec f, g,g)$ and the threading property is immediate from the inductive hypothesis. 
\eqref{eqn:ded-thm-invariant} is also easily verified:
\[
\begin{array}{rcll}
\rhotrans t{\vec \rho} (\vec f, g)\, \vec x\, \vec y\, \vec z & = & \rhotrans s{\vec \rho}(\vec f, g, g)\, \vec x\, \vec y \, \vec z & \text{by definition of $t$}\\
& = & s\, (\vec f\, \vec x)\, (g\, \vec x)\, (g\, \vec x)\, \vec y\, \vec z & \text{by $\IH(s)$}\\
& = & \cntr\, s\, (\vec f\, \vec x)\,  (g\, \vec x)\, \vec y\, \vec z & \text{by $\cntr$ axiom}\\
& = & t\, (\vec f\, \vec x)\,  (g\, \vec x)\, \vec y\, \vec z & \text{by definition of $t$}
\end{array}
\]
\end{itemize}

If $t$ concludes with a cut step,
\[
\vlderivation{
	\vliin{\cut}{}{\vec \sigma, \vec \nat \seqar \tau}{
		\vltr{r}{\vec \sigma, \vec \nat \seqar \pi}{\vlhy{\ }}{\vlhy{\ }}{\vlhy{\ }}
	}{
		\vltr{s}{\vec \sigma, \pi, \vec \nat \seqar \tau}{\vlhy{\ }}{\vlhy{\ }}{\vlhy{\ }}
	}
}
\]
then:
\begin{itemize}
\item If $\pi = \nat$ then we define $\rhotrans t{\vec \rho}(\vec f)$ by just commuting with the cut step:
\[
\toks0={.8}
\toks1={.8}
\vlderivation{
	\vliin{\cut_\nat}{}{\vec \rho, \vec \nat, \vec \tau \seqar \nat}{
		\vltrf{\rhotrans r{\vec \rho}(\vec f)}{\vec \rho, \vec \nat \seqar \nat}{\vlhy{\quad}}{\vlhy{\quad}}{\vlhy{\quad}}{\the\toks0}
	}{
		\vltrf{\rhotrans s{\vec \rho}(\vec f)}{\vec \rho, \nat, \vec \nat , \vec \tau \seqar \nat}{\vlhy{\quad}}{\vlhy{\quad}}{\vlhy{\quad}}{\the\toks1}
	}
}
\]
The threading property is readily obtained from the inductive hypotheses and we verify \eqref{eqn:ded-thm-invariant} as follows:
\[
\begin{array}{rcll}
\rhotrans t{\vec \rho} (\vec f)\, \vec x\, \vec y\, \vec z & = & \rhotrans s{\vec \rho}(\vec f)\, \vec x\, (\rhotrans r{\vec \rho}(\vec f) \, \vec x\, \vec y)\, \vec y\, \vec z & \text{by $\cut$ axiom} \\
& = & s\, (\vec f\, \vec x)\, (r\, (\vec f\, \vec x)\, \vec y)\, \vec y\, \vec z & \text{by $\IH(r)$ and $\IH(s)$} \\
& = & t\, (\vec f\, \vec x)\, \vec y\, \vec z & \text{by $\cut$ axiom}
\end{array}
\]
\item Otherwise, we define $\rhotrans t{\vec \rho}(\vec f)$ as:
\[
\toks0={.4}
\toks1={.8}
\vlderivation{
	\vliq{\cntr}{}{ \vec \rho, \vec \nat , \vec \tau \seqar \nat }{
	\vltrf{\rhotrans s{\vec \rho, \vec \nat}}{\vec \rho, \vec \nat, \vec \nat , \vec \tau \seqar \nat}{
		\vliq{\wk}{}{\vec \rho, \vec \nat, \vec \sigma_i \seqar \nat}{
		\vlin{f_i}{i}{\vec \rho, \vec \sigma_i \seqar \nat}{\vlhy{}}
		}
	}{
		\vlhy{}
	}{
		\vltrf{\rhotrans r{\vec \rho}(\vec f)}{\vec \rho, \vec \nat , \vec \pi \seqar \nat}{\vlhy{\quad}}{\vlhy{\quad}}{\vlhy{\quad}}{\the\toks1}
	}{\the\toks0}
	}
}
\]
The threading property is readily obtained from the inductive hypotheses, and we verify \eqref{eqn:ded-thm-invariant} as follows:
\[
\begin{array}{rcll}
\rhotrans t{\vec \rho}\, \vec x\, \vec y\, \vec z & = & \rhotrans s{\vec \rho, \vec \nat} (\wk^* \vec f , \rhotrans r{\vec \rho} (\vec f))\, \vec x\, \vec y\, \vec y\, \vec z & \text{by $\cntr$ axioms}\\
& = & s\, (\wk^* \vec f\, \vec x\, \vec y)\, (\rhotrans r{\vec \rho}(\vec f)\, \vec x\, \vec y )\, \vec y\, \vec z & \text{by $\IH(s)$}\\
& = & s\, (\wk^* \vec f\, \vec x\, \vec y)\,  (r\, (\vec f\, \vec x)\, \vec y)\, \vec y\, \vec z & \text{by $\IH(r)$ and $\extensionality$}\\
& = & s\, (\vec f\, \vec x)\, (r\, (\vec f\, \vec x)\, \vec y)\, \vec y\, \vec z & \text{by $\wk$ axioms and $\extensionality$}\\
& = & t\, (\vec f\, \vec x)\, \vec y\, \vec z & \text{by definition of $t$}
\end{array}
\]
\end{itemize}

If $t$ concludes with a right-implication step,
\[
\vlderivation{
	\vlin{\rightimp}{}{\vec \sigma, \vec \nat \seqar \pi \to \tau}{
	\vltr{s}{\vec \sigma, \pi , \vec \nat \seqar \tau}{\vlhy{\ }}{\vlhy{\ }}{\vlhy{\ }}
	}
}
\]
then:
\begin{itemize}
\item If $\pi=\nat$ then $\rhotrans t{\vec \rho}(\vec f)$ is just:
\[
\toks0={.8}
\vlderivation{
	\vliq{\exch}{}{\vec \rho, \vec \nat, \nat, \vec \tau \seqar \nat}{
	\vltrf{\rhotrans s{\vec \rho}(\vec f)}{\vec \rho, \nat, \vec \nat , \vec \tau \seqar \nat}{\vlhy{\quad}}{\vlhy{\quad}}{\vlhy{\quad}}{\the\toks0}
	}
}
\]
The threading property is readily obtained from the inductive hypothesis and \eqref{eqn:ded-thm-invariant} is easily verified:
\[
\begin{array}{rcll}
\rhotrans t{\vec \rho}(\vec f)\, \vec x\, \vec y\, y\, \vec z & = & \rhotrans s{\vec \rho}(\vec f)\, \vec x\, y\, \vec y\, \vec z & \text{by $\exch$ axioms}\\
& = & s\, (\vec f\, \vec x)\, y\, \vec y\, \vec z & \text{by $\IH(s)$}\\
& = & t\, (\vec f\, \vec x)\, \vec y\, y\, \vec z & \text{by $\rightimp$ axiom}
\end{array}
\]
\item Otherwise $\level (\pi)\leq n$ and we define $\rhotrans t{\vec \rho} (\vec f)$ as:
\[
\toks0={.4}
\vlderivation{
	\vliq{\exch}{}{\vec \rho, \vec \nat , \pi, \vec \tau \seqar \nat}{
	\vltrf{\rhotrans s{\vec \rho, \pi}}{\vec \rho, \pi, \vec \nat ,\vec \tau \seqar \nat}{
		\vlin{\wk}{}{\vec \rho, \pi, \vec \sigma_i \seqar \nat}{
		\vlin{f_i}{i}{\vec \rho, \vec \sigma_i \seqar \nat}{\vlhy{}}
		}
	}{
		\vlhy{}
	}{
		\vliq{\wk}{}{\vec \rho, \pi, \vec \pi \seqar \nat}{
		\vliq{\apply}{}{\pi, \vec \pi \seqar \nat}{\vlhy{}}
		}
	}{\the\toks0}
	}
}
\]
where $\apply$ is a simple (finite) derivation satisfying:
\begin{equation}
\label{eqn:apply-eqn}
\apply\, z\, \vec w \ = \ z\, \vec w
\end{equation}
The threading property is readily obtained from the inductive hypothesis, and we verify \eqref{eqn:ded-thm-invariant} as follows:
\[
\begin{array}{rcll}
\rhotrans t{\vec \rho}\, \vec x\, \vec y\, z\, \vec z & = & \rhotrans s{\vec \rho, \pi} (\wk\, \vec f , \wk^* \apply) \, \vec x\, z\, \vec y\, \vec z & \text{by $\exch$ axioms} \\
& = & s\, (\wk\, \vec f\, \vec x\, z)\, (\wk^*\apply\, \vec x\, z)\, \vec y\, \vec z & \text{by $\IH(s)$} \\
& = & s\, (\vec f\, \vec x) \, (\apply\, z)\, \vec y\, \vec z & \text{by $\wk$ axioms and $\extensionality$}\\
& = & s\, (\vec f\, \vec x)\, z\, \vec y\, \vec z & \text{by \eqref{eqn:apply-eqn} and $\extensionality$}\\
& = & t\, (\vec f\, \vec x)\, \vec y\, z\, \vec z & \text{by $\rightimp$ axiom}
\end{array}
\]
\end{itemize}

If $t$ concludes with a left-implication step,
\[
\vlderivation{
	\vliin{\leftimp}{}{\vec \sigma, \tau' \to \pi, \vec \nat \seqar \tau}{
		\vltr{r}{\vec \sigma , \vec \nat \seqar \tau'}{\vlhy{\ }}{\vlhy{\ }}{\vlhy{\ }}
	}{
		\vltr{s}{\vec \sigma, \pi, \vec \nat \seqar \tau}{\vlhy{\ }}{\vlhy{\ }}{\vlhy{\ }}
	}
}
\]
then:
\begin{itemize}
\item If $\pi = \nat$ then we define $\rhotrans t{\vec \rho} (\vec f,g)$ as:
\[
\toks0={0.8}
\vlderivation{
	\vliin{\cut}{}{\vec \rho, \vec \nat , \vec \tau \seqar \nat}{
		\vliq{\rightimp}{}{\vec \rho, \vec \nat \seqar \tau'}{
		\vltrf{\rhotrans r{\vec \rho}(\vec f)}{\vec \rho, \vec \nat, \vec \tau'\seqar \nat}{\vlhy{\quad}}{\vlhy{\quad}}{\vlhy{\quad}}{\the\toks0}
		}
	}{
		\vliin{\cut}{}{\vec \rho, \tau', \vec \nat , \vec \tau \seqar \nat}{
			\vlin{g}{}{\vec \rho, \tau' \seqar \nat}{\vlhy{}}
		}{
			\vltrf{\rhotrans s{\vec \rho}(\vec f)}{\vec \rho, \nat, \vec \nat, \vec \tau \seqar \nat}{\vlhy{\quad}}{\vlhy{\quad}}{\vlhy{\quad}}{\the\toks0}
		}
	}
}
\]
The threading property is readily obtained from the inductive hypotheses and we verify \eqref{eqn:ded-thm-invariant} as follows:
\[
\begin{array}{rcll}
\rhotrans t{\vec \rho} (\vec f,g)\, \vec x\, \vec y\, \vec z &  = & \cut\, g\, \rhotrans s{\vec \rho}(\vec f)\, \vec x\, (\rightimp^*\, \rhotrans r{\vec \rho}(\vec f)\, \vec x\, \vec y)\, \vec y\, \vec z & \text{by $\cut$ axiom} \\
& = & \rhotrans s{\vec \rho}(\vec f)\, \vec x\, (g\, \vec x\, (\rightimp^*\, \rhotrans r{\vec \rho}(\vec f)\, \vec x\, \vec y))\, \vec y\, \vec z & \text{by $\cut$ axiom}\\
& =& \rhotrans s{\vec \rho}(\vec f)\, \vec x\, (g\, \vec x\, ( \rhotrans r{\vec \rho}(\vec f)\, \vec x\, \vec y))\, \vec y\, \vec z & \text{by $\rightimp$ axioms and $\extensionality$}\\
& = & s\, (\vec f\, \vec x)\, (g\, \vec x\, ( \rhotrans r{\vec \rho}(\vec f)\, \vec x\, \vec y))\, \vec y\, \vec z & \text{by $\IH(s)$}\\
& = & s\, (\vec f\, \vec x)\, (g\, \vec x\, (r\, (\vec f\, \vec x)\, \vec y))\, \vec y\, \vec z & \text{by $\IH(r)$ and $\extensionality$}\\
& = & t\, (\vec f\, \vec x)\, (\vec g\, \vec x)\, \vec y\, \vec z & \text{by $\leftimp$ axiom}
\end{array}
\]
\item Otherwise we define $\rhotrans t{\vec \rho}(\vec f,g)$ as:
\[
\toks0={.3}
\toks1={.8}
\vlderivation{
	\vliq{\cntr}{}{\vec \rho, \vec \nat, \vec \tau \seqar \nat}{
	\vltrf{\rhotrans s{\vec \rho, \vec \nat}}{\vec \rho, \vec \nat, \vec \nat, \vec \tau \seqar \nat}{
		\vliq{\wk}{}{\vec \rho, \vec \nat, \vec \sigma_i \seqar \nat}{
		\vlin{f_i}{i}{\vec \rho, \vec \sigma_i \seqar \nat}{\vlhy{}}
		}
	}{
		\vlhy{}
	}{
		\vliin{\cut}{}{\vec \rho, \vec \nat, \vec \pi \seqar \nat}{
			\vliq{\rightimp}{}{\vec \rho, \vec \nat \seqar \tau' }{
			\vltrf{\rhotrans r{\vec \rho} (\vec f)}{\vec \rho, \vec \nat, \vec \tau'\seqar \nat}{\vlhy{\quad}}{\vlhy{\quad}}{\vlhy{\quad}}{\the\toks1}
			}
		}{
			\vliq{\wk}{}{\vec \rho , \vec \nat, \tau', \vec \pi \seqar \nat}{
			\vlin{g}{}{\vec \rho, \tau', \vec \pi \seqar \nat}{\vlhy{}}
			}
		}
	}{\the\toks0}
	}
}
\]
The threading property is readily obtained from the inductive hypotheses and we verify \eqref{eqn:ded-thm-invariant} as follows:
\[
\begin{array}{rcll}
\rhotrans t{\vec \rho}(\vec f,g)\, \vec x\, \vec y\, \vec z & = & \rhotrans s{\vec \rho, \vec \nat} ( \wk^*\vec f, \cut (\rightimp^* \rhotrans r{\vec \rho}(\vec f)) (\wk^* g) )\, \vec x\, \vec y\, \vec y\, \vec z & \text{by $\cntr$ axioms}\\
& = & s\, (\wk^*\vec f\, \vec x\, \vec y)\, (\cut (\rightimp^* \rhotrans r{\vec \rho}(\vec f)) (\wk^* g) \vec x\, \vec y)\, \vec y\, \vec z & \text{by $\IH(s)$}\\
& = & s\, (\vec f\, \vec x)\, (\cut (\rightimp^* \rhotrans r{\vec \rho}(\vec f)) (\wk^* g) \vec x\, \vec y)\, \vec y\, \vec z & \text{by $\wk$ axioms and $\extensionality$}\\
& = & s\, (\vec f\, \vec x) (\wk^* g\, \vec x\, \vec y\, (\rightimp^* \rhotrans r{\vec \rho}(\vec f)\, \vec x\, \vec y))\, \vec y \, \vec z & \text{by $\cut$ axiom}\\
& = & s\, (\vec f\, \vec x) (g\, \vec x\,  (\rightimp^* \rhotrans r{\vec \rho}(\vec f)\, \vec x\, \vec y))\, \vec y \, \vec z & \text{by $\wk$ axioms and $\extensionality$} \\
& = & s\, (\vec f\, \vec x) (g\, \vec x\,  ( \rhotrans r{\vec \rho}(\vec f)\, \vec x\, \vec y))\, \vec y \, \vec z & \text{by $\rightimp$ axioms and $\extensionality$}\\
& = & s\, (\vec f\, \vec x)\, (g\, \vec x\, (r\, (\vec f\, \vec x) \, \vec y))\, \vec y \, \vec z & \text{by $\IH(r)$ and $\extensionality$}\\
& = & t\, (\vec f\, \vec x)\, (g\, \vec x)\, \vec y\, \vec z & \text{by $\leftimp $ axiom}
\end{array}
\]
\end{itemize}

Finally, if $t$ concludes with a recursion step,
\[
\vlderivation{
	\vliin{\rec}{}{\vec \sigma , \vec \nat, \nat \seqar \tau}{
		\vltr{r}{\vec \sigma , \vec \nat \seqar \tau}{\vlhy{\ }}{\vlhy{\ }}{\vlhy{\ }}
	}{
		\vltr{s}{\vec \sigma, \tau, \vec \nat, \nat \seqar \tau}{\vlhy{\ }}{\vlhy{\ }}{\vlhy{\ }}
	}
}
\]
then:
\begin{itemize}
\item If $\tau = \nat$ (so $\vec \tau$ is empty) then we define $\rhotrans t{\vec \rho}$ similarly to before in Example~\ref{ex:sim-prim-rec},
\begin{flushleft}
\begin{equation}
\label{eqn:rec-to-cyc-low-comp}
\toks0={.8}
\vlderivation{
	\vliin{\cond}{\bullet}{\vec \rho, \vec \nat, \blue{\nat} \seqar \nat}{
		\vltrf{\rhotrans r{\vec \rho} (\vec f)}{\vec \rho, \vec \nat \seqar \nat}{\vlhy{\quad}}{\vlhy{\quad}}{\vlhy{\quad}}{\the\toks0}
	}{
		\vliin{\cut}{}{\vec \rho, \vec \nat, \blue{\nat} \seqar \nat}{
			\vlin{\cond}{\bullet}{\vec \rho, \vec \nat, \blue{\nat} \seqar \nat}{\vlhy{\vdots}}
		}{
			\vltrf{\rhotrans s{\vec \rho}(\vec f)}{\vec \rho, \nat, \vec \nat, \nat \seqar \nat}{\vlhy{\quad}}{\vlhy{\quad}}{\vlhy{\quad}}{\the\toks0}
		}
	}
}
\end{equation}
\end{flushleft}
where $\bullet$ marks roots of identical sub-coderivations.
The threading property is readily obtained from the inductive hypotheses. 
For progressiveness, any infinite branch either just loops on $\bullet$ indefinitely, in which case there is a progressing thread along the blue $\blue \nat$, or is eventually in just $\rhotrans r{\vec \rho}(\vec f)$ or $\rhotrans s{\vec \rho}(\vec f)$, in which case there is a progressing thread by the inductive hypotheses.
To verify \eqref{eqn:ded-thm-invariant} we shall show,
\begin{equation}
\label{eqn:ind-to-cyc-rec-n-ih}
\rhotrans t{\vec \rho}(\vec f)\, \vec x\, \vec y\, y \ = \ t\, (\vec f\, \vec x)\, \vec y\, y
\end{equation}
by (object-level) induction on $y$:
\[
\begin{array}{rcll}
\rhotrans t{\vec \rho}(\vec f)\, \vec x\, \vec y\, 0 & = & \rhotrans r{\vec f}(\vec f)\, \vec x\, \vec y & \text{by $\cond$ axiom}\\
& = & r\, (\vec f\, \vec x)\, \vec y& \text{by $\IH(r)$} \\
& = & t\, (\vec f\, \vec x)\, \vec y\, 0 & \text{by $\rec$ axioms}\\
\noalign{\medskip}
\rhotrans t{\vec \rho}(\vec f)\, \vec x\, \vec y\, \succ y & = & \cut\, \rhotrans t{\vec \rho}(\vec f)\, \rhotrans s{\vec \rho}(\vec f)\, \vec x\, \vec y\, y & \text{by $\cond$ axioms}\\
& = & \rhotrans s{\vec \rho}(\vec f)\, \vec x\, (\rhotrans t{\vec \rho}(\vec f)\, \vec x\, \vec y\, y)\, \vec y\, y & \text{by $\cut$ axiom}\\
& = & s\, (\vec f\, \vec x)\, (\rhotrans t{\vec \rho}(\vec f)\, \vec x\, \vec y\, y)\, \vec y\, y & \text{by $\IH(s)$}\\
& = & s\, (\vec f\, \vec x)\, (t\, (\vec f\, \vec x)\, \vec y\, y)\, \vec y\, y & \text{by inductive hypothesis \eqref{eqn:ind-to-cyc-rec-n-ih}}\\
& = & t\, (\vec f\, \vec x)\, \vec y\, \succ y & \text{by $\rec$ axioms}
\end{array}
\]
\item Otherwise we define $\rhotrans t{\vec \rho}(\vec f)$ to be,
\begin{equation}
\label{eqn:rec-to-cyc-high-comp}
\toks0={.8}
\toks1={.3}
\vlderivation{
	\vliin{\cond}{\bullet}{\vec \rho , \vec \nat, \blue \nat, \vec \tau \seqar \nat}{
		\vltrf{\rhotrans r{\vec \rho}(\vec f)}{\vec \rho, \vec \nat, \vec \tau \seqar \nat}{\vlhy{\quad}}{\vlhy{\quad}}{\vlhy{\quad}}{\the\toks0}
	}{
		\vliq{\cntr}{}{\vec \rho, \vec \nat, \blue{\nat }, \vec \tau \seqar \nat}{
		\vltrf{\rhotrans s{\vec \rho, \vec \nat , \blue{\nat}}}{\vec \rho, \vec \nat , \blue{\nat }, \vec \nat , \nat, \vec \tau \seqar \nat}{
			\vliq{\wk}{}{\vec \rho, \vec \nat, \blue{\nat }, \vec \sigma_i \seqar \nat}{
			\vlin{f_i}{i}{\vec \rho, \vec \sigma_i \seqar \nat}{\vlhy{}}
			}
		}{
			\vlhy{}
		}{
			\vlin{\cond}{\bullet}{\vec \rho, \vec \nat, \blue{\nat}, \vec \tau \seqar \nat}{\vlhy{\vdots}}
		}{\the\toks1}
		}
	}
}
\end{equation}
where $\bullet$ marks roots of identical sub-coderivations.
The threading property is readily obtained from the inductive hypotheses.
For progressiveness, notice that any infinite branch that hits $\bullet$ infinitely often will have a progressing thread along the blue $\blue \nat$, thanks to the threading property from the inductive hypothesis for $s$.
Any other infinite branch is eventually in $\rhotrans r{\vec \rho}(\vec f)$ or $\rhotrans{s}{\vec \rho, \vec \nat, \nat}$, so progressiveness follows from the inductive hypotheses. 

To verify \eqref{eqn:ded-thm-invariant} we show that,
\begin{equation}
\label{eqn:ind-to-cyc-rec-ih}
\rhotrans t{\vec \rho}(\vec f)\, \vec x\, \vec y\, y\, \vec z \ = \ t\, (\vec f\, \vec x)\, \vec y\, y\, \vec z
\end{equation}
by (object-level) induction on $y$:
\[
\begin{array}{rcll}
\rhotrans t{\vec \rho}(\vec f)\, \vec x\, \vec y\, 0 \, \vec z & = & \rhotrans r{\vec \rho}(\vec f)\, \vec x\, \vec y\, \vec z & \text{by $\cond$ axioms}\\
& = & r\, (\vec f\, \vec x)\, \vec y\, \vec z & \text{by $\IH(r)$}\\
& = & t\, (\vec f\, \vec x)\, \vec y\, 0\, \vec z & \text{by $\rec$ axioms}\\
\noalign{\medskip}
\rhotrans t{\vec \rho}(\vec f)\, \vec x\, \vec y\, \succ y\, \vec z & = & \cntr^*\, \rhotrans s{\vec \rho, \vec \nat , \nat}(\wk^*\vec f, \rhotrans t{\vec \rho}(\vec f))\, \vec x\, \vec y\, y\, \vec z & \text{by $\cond$ axioms}\\
& = & \rhotrans s{\vec \rho, \vec \nat , \nat}(\wk^*\vec f, \rhotrans t{\vec \rho}(\vec f))\, \vec x\, \vec y\, y\, \vec y\, y\, \vec z & \text{by $\cntr$ axioms}\\
& = & s\, (\wk^* \vec f\, \vec x\, \vec y\, y)\, (\rhotrans t{\vec \rho}(\vec f)\, \vec x\, \vec y\, y)\, \vec y\, y\, \vec z & \text{by $\IH(s)$}\\
& = & s\, (\vec f\, \vec x)\, (\rhotrans t{\vec \rho}(\vec f)\, \vec x\, \vec y\, y)\, \vec y\, y\, \vec z & \text{by $\wk$ axioms and $\extensionality$} \\
& = & s\, (\vec f\, \vec x)\, (t\, (\vec f \, \vec x)\, \vec y\, y)\, \vec y\, y\, \vec z & \text{by inductive hypothesis \eqref{eqn:ind-to-cyc-rec-ih}}\\
& = & t\, (\vec f\, \vec x)\, \vec y\, \succ y\, \vec z & \text{by $\rec$ axioms} \qedhere
\end{array}
\]
\end{itemize}
\end{proof}

\subsection{$\nC n$ simulates $\nT{n+1}$}

We are now ready to prove the main result of this section, which essentially boils down to an instance of the main Lemma in the previous subsection.

\begin{proof}
[Proof of Theorem~\ref{thm:c-sim-t}]
Without loss of generality assume $t$ is a derivation of $\seqar \sigma$ and apply Lemma~\ref{lem:ded-thm-realised} with $\vec \rho$ empty, to obtain a coderivation $\rhotrans t \emptyset: \vec \sigma \seqar \nat$ s.t.:
\begin{equation}
\label{eqn:thm-csimt-use-of-lemma}
 \nC n, \rec_{n+1} \proves \ t\, \vec x\, =\, \rhotrans t \emptyset\, \vec x
\end{equation}
Note that $\rhotrans t \emptyset$ uses no new initial sequents since the antecedent of the conclusion of $t$ is empty.
Now we define $t'$ as:
\[
\vlderivation{
	\vliq{\rightimp}{}{\seqar \sigma}{
	\vltr{\rhotrans t \emptyset}{\vec \sigma \seqar \nat}{\vlhy{\ }}{\vlhy{\ }}{\vlhy{\ }}
	}
}
\]
We verify \eqref{eqn:prov-eq-t-ct-terms-with-ext} as follows, working inside the theory:
\[
\begin{array}{rll}
& t\, \vec x \, = \, \rhotrans t \emptyset\, \vec x & \text{by \eqref{eqn:thm-csimt-use-of-lemma}}\\
\therefore & t\, \vec x \, = \, \rightimp^*\rhotrans t \emptyset\, \vec x & \text{by $\rightimp$ axioms} \\
\therefore & t\, \vec x\, = \, t'\vec x & \text{by definition of $t'$}\\
\therefore & t = t' & \text{by $\extensionality$} \qedhere
\end{array}
\]
\end{proof}

\section{Coterm-based models of $\T$ and $\C$}
\label{sect:coterm-models}

Our ultimate goal is to establish a converse to the main result of the previous section, which we shall demonstrate in the next two sections.
Before that we need to introduce some type structures arising from our formulation of $\C$. 
Ultimately, we will reduce the simulation of $\nC n$ in $\nT{n+1}$ to the provability in an appropriate arithmetic theory that certain structures are models of $\C$ .

We cannot formalise the standard model $\nmod$ in arithmetic for cardinality reasons, however there are natural models of partial recursive functionals that can be formalised, namely the \emph{hereditarily recursive} and the \emph{hereditarily effective} operations of finite type (see, e.g., \cite{ho-computability}).

Since regular coterms form a Turing-complete programming language (cf.~Proposition~\ref{prop:turing-completeness}), the domains of our structures will simply be classes of regular coterms.
Viewed as Kleene-Herbrand-G\"odel equational specifications (see, e.g., \cite{kleene:intro-to-metamath}), note that the corresponding notion of \emph{computation} is subsumed by provable equality of coterms by rules and axioms of $\C$.
In fact, we will distill from $\C$ a suitable fragment of provable equality, namely that induced by the respective rewriting system.
This will eventually give rise to an intensional model, whence an extensional one can be obtained by, as usual, taking the extensional collapse.

\subsection{Reduction and conversion of coterms}
The reduction relation $\reduces$ on coterms is defined by orienting all the equations in Figures~\ref{fig:eq-ax-seq-calc-min}, \ref{fig:rec-axioms} and \ref{fig:cond-rule+axioms} left-to-right and taking closure under substitution and contexts.
Formally:

\begin{definition}
[Reduction and conversion]
$\reduces$ is the least relation on (co)terms satisfying the reductions in Figure~\ref{fig:reduction-rules}, where the types of each rule label are as indicated in Figures~\ref{fig:seq-calc-min}, \ref{fig:rules-0-s-rec} and \ref{fig:cond-rule+axioms},\footnote{Again there is no reason, other than for ease of legibility, that we write terms variables $s,t$ etc.\ for the premiss inputs and variables $x,y$ etc.\ for the antecedent inputs. Under substitution, it makes no difference.} and closed under substitution and contexts.

We write $\conv$ for the reflexive, symmetric, transitive closure of $\reduces$, and freely use standard rewriting theoretic terminology and notations for these relations.
We sometimes write $\reduces_\sigma$ for the restriction of $\reduces$ to coterms of type $\sigma$, and $\conv_\sigma$ for the restriction of $\conv$ to coterms of type $\sigma$.
\end{definition}

\begin{figure}[h]
\[
\begin{array}{rcl}
\id \ x & \reduces & x \\
\exch \ t \ \vec x\  x \ y\ \vec y & \reduces & t \ \vec x \ y \ x \ \vec y \\
\wk \ t  \ \vec x \ x & \reduces & t \ \vec x \\
\cntr \ t\ \vec x \ x & \reduces & t \ \vec x \ x \ x \\
\cut \ s\ t\ \vec x & \reduces & t \ \vec x \ (s \ \vec x ) \\
\leftimp \ s\ t\ \vec x\ y & \reduces & t \ \vec x\ (y\ (r\ \vec x) ) \\
\rightimp \ t \ \vec x \ x & \reduces & t\ \vec x\ x \\
\noalign{\smallskip}
\rec\ s\ t\ \vec x\ 0 & \reduces & s\ \vec x\\
\rec\ s\ t\ \vec x\ \succ y & \reduces & t\ \vec x\ (\rec\ s\ t\ \vec x \ y)\\
\noalign{\smallskip}
\cond \ s\ t\ \vec x\ 0 & \reduces & s \ \vec x \\
\cond\ s\ t\ \vec x\ \succ y & \reduces & t\ \vec x\ y
\end{array}
\]
\caption{Reduction rules for (co)terms}
\label{fig:reduction-rules}
\end{figure}

We shall address metamathematical matters w.r.t.\ formalising reduction and reduction sequences shortly, but first let us examine some basic mathematical properties of reduction.

\subsection{On normality and numerality}
\label{sect:on-normality-and-numerality}
The term model for $\T$, due to Tait \cite{Tait:67:normalisation-of-t-+-bar-rec-typ01}, may be obtained by proving normalisation and confluence of the rewrite system in Figure~\ref{fig:reduction-rules} for \emph{terms}, and taking the (unique) normal forms of closed terms as the domain of the type structure.
Crucial to this approach is the property that the only closed normal forms of terms of type $\nat$ are the numerals, which allows induction in $\T$ to be reduced to induction at the meta-level.

Since regular coterms are Turing-complete, cf.~Proposition~\ref{prop:turing-completeness},\footnote{Notice from that proof that Turing-completeness is retained under $\reduces$ as a model of execution.} one cannot hope for such a normalisation result, and so our models will be obtained by restricting to normalising coterms.
Normalisation for progressing coterms at type $\nat$ follows by showing that they are indeed elements of the type structure.
We will address confluence shortly, adapting known proofs to the non-wellfounded setting.
However, we point out that the final point, that closed normal forms of type $\nat$ are just the numerals, fails for coterms:

\begin{remark}
[Non-numeral normal coterms of type $\nat$]
\label{rmk:non-numeral-normal}
There are closed $\reduces$-normal coterms of type $\nat$ that are not numerals, e.g.\ the coterm $H$ with,
\[
H = \cond \, 0 \, 1\, H
\]
that returns $0$ or $1$, depending on whether it itself is $0$ or non-zero. Clearly $H$ is undefined in the standard model, though may be consistently interpreted by $0$ or $1$ in extensions of it. 
Note that $H$ is even a \emph{regular} counterexample to numerality of closed $\reduces$-normal coterms, and so the $\reduces$-normal regular coterms will be too large a domain for the type structures we later define.
For this reason, our type structures will restrict the interpretation of $\nat$ to (classes of) coterms that normalise to a numeral.
\end{remark}

\subsection{Formalising reduction sequences of regular coterms}
\label{sect:form-red-seq}
Before continuing, let us make some comments about our arithmetisation of reduction sequences.
First and foremost, to avoid unnecessary technicalities, all the coterms we consider later will be regular, and so can be coded by natural numbers.
Thus all quantification over them is strictly first-order.
In fact, many of our results go through in a more general setting since finite reduction sequences may be coded by finite data, but such a treatment is beyond the scope of this work.
In what follows we shall be rather brief, outlining only the main ideas and proof methods behind the results we need.

While equality for arbitrary coterms is $\Pi^0_1$, we should justify that it is recursive for regular coterms.
Namely, we should argue that we may decide if two \emph{presentations} of regular coterms (as finite labelled graphs) represent the same coterm (as an infinite labelled tree), i.e.\ are \emph{bisimilar}.

For $v \in \{0,1\}^*$ and a coterm $t$, construed as a labelled binary tree, let us temporarily write $t_v$ for the sub-coterm of $t$ rooted at position $v$, and $t[v]$ for the rule instance at position $v$.
For a regular presentation $G$ of a coterm, a finite rooted labelled graph, write $G(v)$ for the node of $G$ reached by following the word $v$ along its edge relation.
Note that all of these notions are provably computable in $t$, $G$ and $v$ in $\RCA$.

\begin{proposition}
\label{obs:bisim-rca}
Bisimilarity of rooted finite labelled directed graphs is provably recursive in $\RCA$.
\end{proposition}
\begin{proof}
For recursivity, just blindly search for a bisimulation relation between the two vertex sets (in exponential time).
Now we need to show that two finite graphs $G$ and $H$ are bisimilar if and only if there is a bisimulation between them:
\begin{itemize}
\item Suppose $R$ is a bisimulation between $G$ and $H$. We prove by induction on node position $v\in \{0,1\}^*$ that $(G(v),H(v)) \in R$. From here, by definition of a bisimulation, any two nodes related by $R$ have the same label, and so the unfoldings of $G$ and $H$ are equal.
\item Suppose $G$ and $H$ have the same unfolding. We inductively construct a bisimulation $R$ from the root by continually adding $(G(v),H(v))$ to $R$. We terminate when we hit a pair that is already in $R$, which will happen in at most $|G||H|+1$ steps, by the \emph{finite} pigeonhole principle.\footnote{Note that, while the usual \emph{infinite} pigeonhole principle is not provable in $\RCA$, it is easy to see that the finite one is provable in $\RCA$ (and even weaker theories), by a straightforward $\Sigma^0_1$ induction.}\qedhere
\end{itemize}
\end{proof}

\begin{corollary}
\label{cor:eq-reg-coterms-rca}
Equality for regular coterms is provably $\Delta^0_1$ in $\RCA$.
\end{corollary}

Now, we better show that reduction preserves certain properties of coterms, not least regularity.
More generally, since we will work with a certain class of regular coterms, we should make sure that this class is closed under reduction.

\begin{proposition}
[$\RCA$]
\label{prop:red-cont}
If $s\reduces t$ then $t$ is finitely composed of sub-coterms of $s$:
\begin{equation}
\label{eqn:red-cont}
\exists \text{ a finite term } r(x_1,\dots,x_n).\, \exists \langle v_1,\dots, v_n\rangle .\,  t = {r(s_{v_1}, \dots, s_{v_n})}
\end{equation}
\end{proposition}

This result holds for arbitrary coterms and, indeed, the existential quantifiers may be explicitly witnessed by primitive recursive functions in terms of $s$ and $t$.
We can take $s_{v_1},\dots ,s_{v_n}$ to include the coderivations indicated in the contractum of a reduction in Figure~\ref{fig:reduction-rules}, as well as the `comb' of the redex of the reduction in $s$, i.e.\ the siblings of all the nodes in the path leading to the redex.
$r(\vec x)$ is now the finite term induced by the contracta and this comb.\footnote{Note that there is no need for weak K\"onig's lemma here, since combs may be explicitly proved to induce finite trees in $\RCA$. In any case, $\WKL$ is arithmetically conservative over $\RCA$.}

As an immediate consequence of the above proposition we have, for arbitrary coterms:
\begin{corollary}
[$\RCA$]
\label{cor:pres-fin-props-by-reduction}
Suppose $s \reduces t$.
\begin{enumerate}
\item\label{item:red-pres-fin-var} If $s$ has only finitely many variable occurrences then so does $t$.
\item\label{item:red-pres-fin-red} If $s$ has only finitely many redexes then so does $t$.
\item\label{item:red-pres-reg} If $s$ is regular then so is $t$.
\item\label{item:red-pres-prog} If $s$ is progressing then so is $t$.
\end{enumerate}
\end{corollary}

Restricting now to regular coterms, we obtain the analogue of Proposition~\ref{prop:red-cont} for reduction \emph{sequences} by $\Sigma^0_1$-induction:
\begin{proposition}[$\RCA$]
\label{prop:red-seq-cont}
If $s$ is regular and $s\reduces^*t$, then \eqref{eqn:red-cont}. 
\end{proposition}

Again, we could also deduce a similar result for arbitrary coterms, specifying reduction sequences as finite lists of redex positions, but formally developing this is beyond the scope of this work.
As expected, we obtain the same properties of Corollary~\ref{cor:pres-fin-props-by-reduction} for regular reduction sequences:

\begin{corollary}
[$\RCA$]
\label{cor:pres-fin-props-by-rtc-reduction}
Suppose $s$ is regular and $s\reduces^*t$. Then the conditions \eqref{item:red-pres-fin-var}, \eqref{item:red-pres-fin-red}, \eqref{item:red-pres-reg} and \eqref{item:red-pres-prog} from Corollary~\ref{cor:pres-fin-props-by-reduction} hold.
\end{corollary}

\subsection{Confluence of reduction}
In order to obtain basic metamathematical properties of the coterm models we later consider, we need to know that our model of computation is \emph{deterministic}, so that coterms have unique interpretations.
There are various ways to prove this in arithmetic, but we will approach it in terms of \emph{confluence} in rewriting theory.

We will need to formalise our argument within $\RCA$ for later results,
comprising an additional contribution to the literature,
bounding the logical strength of confluence for finitary and certain infinitary rewrite systems operating with regular coterms.
To facilitate this formalisation, all the coterms we consider in this section will be finite applications of regular coderivations to variables and constants (`\farc s' for short) which, by Proposition~\ref{prop:red-seq-cont}, are closed under reduction sequences, provably in $\RCA$.
\farc s are not necessarily progressing but, since coderivations are closed and have no redexes or variables, \farc s have at most finitely many variables or redexes, provably in $\RCA$ (each application operation can add at most one variable and redex).

The main goal of this subsection is to prove the following:

\begin{theorem}
	[Church-Rosser, $\RCA$]
	\label{thm:cr}
	Let $t:\sigma$ be a \farc.
	If $t_0 \unreducess t \reduces^* t_1$ then there is $t':\sigma$ such that $t_0 \reduces^* t'\unreducess t_1$.
\end{theorem}

To some extent, we follow a standard approach to proving this result.
In particular we aim to find a relation $\vartriangleright$ such that:
\begin{enumerate}
    \item $s\reduces t $ $\implies$ $s\vartriangleright t$ $\implies$ $s\reduces^*t$ ; and,
    \item $\vartriangleright$ satisfies the `diamond' property: if $t_0 \vartriangleleft t \vartriangleright t_1$ then there is $t'$ such that $t_0 \vartriangleright t' \vartriangleleft t_1$. 
\end{enumerate} 
The first property ensures that the reflexive transitive closure of $\reduces $ and $\vartriangleright$ coincide, i.e. $\reduces^* = \vartriangleright^*$. The second property then ensures confluence of $\reduces^*$

Since coterms are infinite (and, moreover, non-wellfounded), we must carry out our argument without appeal to induction on \emph{term structure}, ruling out standard arguments due to Tait and Martin-L\"of (cf., e.g., \cite{HindleySeldin:combinators}). 
Approaches of Takahashi in \cite{takahashi95:par-red-cr-via-complete-developments} and others that rely on \emph{complete developments} could potentially be adapted, since we are working with \farc s which have only finitely many redexes.
However, instead, we perform
an argument by induction on reduction \emph{length}, as in, e.g., \cite{pfenning92:proof-of-cr-in-lf}, which also seems to require less machinery from rewriting theory.

\begin{definition}
	[Parallel reduction]
	We define the relation $\parred$ on \farc s as follows:
	\begin{enumerate}
		\item\label{item:parred-refl} $t\parred t$ for any \farc\ $t$. 
		\item\label{item:parred-nest} For a reduction step $\mathsf r \, \vec t \reduces r(\vec t)$, if each $t_i \parred t_i' $ then we have $\mathsf r\, \vec t \parred r(\vec t' )$.
		\item\label{item:parred-nest-succ} For a reduction step $\mathsf r\, \vec t\, \succ s\, \reduces\, r(\vec t,s)$ (i.e.\ a $\rec$ or $\cond$ successor step), if each $t_i \parred t_i' $ and $s \parred s'$ then we have $\mathsf r\, \vec t\, \succ s\, \parred\, r(\vec t',s')$.
		\item\label{item:parred-par} If $s\, \parred s' $ and $t \parred t' $ then $s\,t \parred s'\, t'$.
	\end{enumerate}
\end{definition}

Note that we really do seem to require clause \eqref{item:parred-refl}, $t \parred t$, for arbitrary \farc s $t$, not just variables and constants, since we cannot finitely derive the former from the latter.

\begin{proposition}
[$\RCA$]
\label{prop:red-in-parred-in-rtcred}
$s\reduces t\ \implies \ s\parred t\ \implies \ s\reduces^* t $.
\end{proposition}

The proof of this result is not difficult, but before giving an argument let us point out a particular consequence that we will need, obtained by $\Sigma^0_1$-induction on the length of reduction sequences:
\begin{corollary}
[$\RCA$]
\label{cor:parred-refltrans-red-refltrans}
$s\reduces^* t \ \iff \ s\parred^* t$
\end{corollary}

Even though it is not necessary to prove the proposition above, we shall first prove the following useful lemma since we will use it later:

\begin{lemma}
	[Substitution, $\RCA$]
	\label{lem:substitution-finitary}
Suppose $t\parred t'$. If $s\parred s'$ then $s[t/x] \parred s'[t'/x]$, for a variable $x$ of the same type as $t$ and $t'$.
\end{lemma}

In what follows, we may write $d:s \reduces^* t$ or $d:s\parred t$ or even $d:s\parred^* t$ to indicate that $d$ is a (finite) derivation witnessing the respective relation.

\begin{proof}
We show $d: s\parred s' \, \implies \,  s[t/x] \parred s'[t'/x]$ by $\Sigma^0_1$-induction on the structure of the derivation $d:s \parred s'$.

If $d$ is obtained by just \eqref{item:parred-refl}, i.e.\ $s' = s$, then we show $s[t/x] \parred s[t'/x]$ by a subinduction on the maximum depth of an $x$-occurrence in $s$. 
(Recall that \farc s have only finitely many variable occurrences.)
\begin{itemize}
\item if $s$ is just the variable $x$ then $s[t/x] \parred s'[t'/x]$ by assumption of $t\parred t'$.
\item if $s$ is just a variable $y \neq x$ then $s[t/x] = y \parred y = s'[t'/x]$ by \eqref{item:parred-refl}.
\item if $s$ is just a constant symbol $\mathsf r$ then $s[t/x] = \mathsf r \parred \mathsf r = s'[t'/x]$ by \eqref{item:parred-refl}.
\item if $s = s_0s_1$ then all the (finitely many) occurrences of $x$ in $s_0$ and $s_1$ have lower depth, so we have $s_0[t/x]\parred s_0[t'/x]$ and $s_1[t/x] \parred s_1[t'/x]$ by the inductive hypothesis, whence $s[t/x] \parred s[t'/x]$ by \eqref{item:parred-par}.
\end{itemize}

If $d$ ends by \eqref{item:parred-nest} we can write $s =\, \mathsf r\, \vec s$ and $s'= \, r(\vec s')$ s.t.\ $\mathsf r\, \vec s \reduces r(\vec s)$ is a $\mathsf r$-reduction and each $s_i \parred s_i'$.
By the inductive hypothesis we have that $s_i[t/x]\parred s_i' [t'/x]$. 
We also have that $ \mathsf r\, \vec s[t/x] \reduces r(\vec s[t/x]) $ is an $\mathsf r$-reduction, where $r$ is variable-free (by inspection of the reduction rules). 
Thus we have $s[t/x] = \, \mathsf r\, \vec s[t/x]\,  \parred \, r( \vec s'[t/x])\, = s'[t/x]$ by \eqref{item:parred-nest}.

(The argument when $d$ ends by \eqref{item:parred-nest-succ} is similar to that of \eqref{item:parred-nest}.)

If $d$ ends by \eqref{item:parred-par} we can write $s = s_0s_1$ and $s'=s_0's_1'$ s.t.\ $s_0 \parred s_0'$ and $s_1\parred s_1'$. 
By the inductive hypothesis we have $s_0[t/x] \parred s_0'[t/x]$ and $s_1[t/x] \parred s_1'[t/x]$, whence $s[t/x] \parred s'[t/x]$ by \eqref{item:parred-par}.
\end{proof}

Notice that Proposition~\ref{prop:red-in-parred-in-rtcred} now follows immediately, by simply instantiating the Lemma above with $s = s'$ to deduce context-closure of $\parred$.

We can now turn to proving the required `diamond property' for $\parred$:

\begin{lemma}
	[Diamond property of $\parred$, $\RCA$]
	\label{lem:dia-prop-parred}
	Suppose $t_0 \parder s \parred t_1$. Then there is some $u $ with $t_0 \parred u \parder t_1$.
\end{lemma}

Before giving the proof, it will be useful to have the following intermediate result:

\begin{proposition}
[$\RCA$]
\label{prop:head-par-red}
Suppose $d: \, \mathsf r\, \vec s\, \parred \, t$, and there is no redex in $\mathsf r\, \vec s$ involving $\mathsf r$. 
There are some $\vec t$ s.t.\ $t = \mathsf r\, \vec t$ and, for each $i$, some $d_i: \, s_i \parred t_i$ for some $d_i<d$.
\end{proposition}
\begin{proof}
We proceed by $\Sigma^0_1$-induction on the length of $\vec s$.
If $\vec s$ is empty, then only clause \eqref{item:parred-refl} applies to $\mathsf r\, \vec s$, so we may set $\vec t$ empty too.

Suppose $\mathsf r\, \vec s\, s\, \parred t$, and there is no redex involving $\mathsf r$. Then only clause \eqref{item:parred-par} applies to $\mathsf r\, \vec s\, s$, so we have $t=t_0t_1$ and some $d_0 :\,  \mathsf r\, \vec s\, \parred t_0$ and $d_1: s \parred t_1$ with $d_0,d_1 < d$.
By inductive hypothesis there are $\vec t_0$ s.t.\ $t_0 \, =\,  \mathsf r\, \vec t_0$ and some $d_{0i} : s_i \parred t_{0i}$ with $d_{0i} < d_0$.
Thus we have $t\, = \, \mathsf r\, \vec t_0\, t_1$ and $d_{0i} : s_i \parred t_{0i}$ and $d_1: s \parred t_1$ with $d_{0i},d_1 < d$, as required.
\end{proof}

We are now ready to prove the diamond property for $\parred$:
\begin{proof}
[Proof of Lemma~\ref{lem:dia-prop-parred}]
	We proceed by simultaneous induction on the structure of the reductions $t \parred t_0$ and $t \parred t_1$.
	Formally we show,
	\[
	\exists s' .\, ((d_0 : s\parred t_0\ \text{ and } \ d_1:s\parred t_1) \implies (t_0 \parred s' \text{ and } t_1 \parred s' ))
	\]
	by $\Sigma^0_1$-induction on $\min(|d_0|,  |d_1|)$.
	
	\begin{itemize}
		\item If $t_0 = s$ (i.e.\ $d_0$ is just \eqref{item:parred-refl}) then we have $t_0 \parred t_1 \parder t_1$, by assumption and reflexivity, and we are done.
		\item Similarly if $t_1 = s$ (i.e.\ $d_1$ is just \eqref{item:parred-refl}), so we may henceforth assume that $t_0 \neq t \neq t_1$, and $d_0$ and $d_1$ conclude with one of the clauses \eqref{item:parred-nest},\eqref{item:parred-nest-succ} or \eqref{item:parred-par}.
		\item If $d_0$ and $d_1$ both end by clause \eqref{item:parred-nest}, then we have $s = \mathsf r\, \vec s$ and $t_0 = r(\vec s_0)$ and $t_1 = r(\vec s_1)$, for some reduction step $\mathsf r\, \vec s \, \reduces\, r(\vec s)$ and some $\vec s_0 , \vec s_1$ s.t.\ $s_i \parred s_{0i}$ and $ s_i \parred s_{1i}$.
		By the inductive hypothesis we have $\vec s'$ s.t.\ $s_{0i} \parred s_i'$ and $s_{1i} \parred s_i'$.
		Thus by repeatedly applying Lemma~\ref{lem:substitution-finitary} we have $t_0 = r(\vec s_0) \parred r(\vec s')$ and $t_1 = r(\vec s_1)\parred r(\vec s')$, so we may set $s' = r(\vec s')$.
		\item (similarly if $d_0$ and $d_1$ end with clause \eqref{item:parred-nest-succ}) 
		\item (it is not possible for one to end with clause \eqref{item:parred-nest} and the other to end by \eqref{item:parred-nest-succ}, since in all cases precisely one reduction step applies at the head.)
		\item If $d_0$ ends by clause \eqref{item:parred-nest} and $d_1$ ends by clause \eqref{item:parred-par}, then from $d_0$ we have $\mathsf r, r$ s.t.\ $s\, = \, \mathsf r\, \vec s\, u \, \reduces r(\vec s, u)$ and some $\vec s_0,u_0$ with $s_i \parred s_{0i}$ and $u \parred u_0$ with $t_0 = r(\vec s_0,u_0)$.
		From $d_1$ we further have some $t', u_1$ s.t.\ $t_1 = t'u_1$ and $\mathsf r\, \vec s\, \parred t'$ and $u \parred u_1$.
		By Proposition~\ref{prop:head-par-red} we have some $\vec s_1$ s.t.\ $t'= \mathsf r\, \vec s_1$ and smaller derivations of $s_i \parred s_{1i}$.
		By the inductive hypothesis we have $\vec s', u'$ s.t.\ $s_{0i} \parred s_i'\parder s_{1i}$ and $u_0 \parred u'\parder u_1$.
		Thus, by repeatedly applying Lemma~\ref{lem:substitution-finitary} we have that $t_0 = r(\vec s_0, u_0) \parred r(\vec s',u')$, and by \eqref{item:parred-nest} we have $t_1 = \mathsf r\, \vec s_1\, u_1\, \parred \, r(\vec s', u')$, so it suffices to set $t' = r(\vec s',u')$.
		\item (the case when $d_0$ ends by clause \eqref{item:parred-nest} and $d_1$ ends by clause \eqref{item:parred-par} is symmetric to the one above)
		\item If both $d_0$ and $d_1$ end by clause \eqref{item:parred-par}, then we have $s = s_0s_1$, $t_0 = r_0 r_1$, $t_1 = u_0u_1$ s.t.\ $r_0 \parder s_0 \parred u_0$ and $r_1 \parder s_1 \parred u_1$.
		By inductive hypothesis we have $s_0'$ and $s_1'$ s.t.\ $r_0\parred s_0' \parder u_0$ and $r_1 \parred s_1'\parder u_1$.
		Thus we have that $t_0 = r_0r_1 \parred s_0's_1'\parder u_0u_1$ by \eqref{item:parred-par}, so we may set $s'= s_0's_1'$.

	\end{itemize}
\end{proof}

\begin{proposition}
	[Weighted CR for $\parred$, $\RCA$]
	If $t_0 \parder^m t \parred^n t_1$ then there is some $t'$ with $t_0 \parred^n t' \parder^m t_1$.
\end{proposition}
\begin{proof}
We show,
	\[
	 (d_0 : t \parred^m  t_0  \text{ and } d_1 : t\parred^n t_1) \, \implies \, \exists t'  ( t_0 \parred^n t' \, \text{ and } \, d_1': t_1 \parred^m t' )
	\]
	by $\Sigma^0_1$-induction on $ m = |d_0|$.
	\begin{itemize}
		\item If $m = 0$ then $t_0 =t$ and we may simply set $t' = t_1$, whence we have that $t_0 = t \parred^n t'$ by assumption and $t' = t_1\parder^0 t_1$.
		\item Now, suppose $t_0 \parder t_0' \parder^m t \parred^n t_1$. By the inductive hypothesis we have that 
		$$t_0' \parred^n t'' \parder^m t_1$$ for some $t_1'$. This means in particular that $t_0 \parder t_0' \parred^n t_1'$, so we have again from the inductive hypothesis a $t' $ with 
		$$t_0 \parred^n t' \parder t_1'$$
		Putting these together we indeed have $t_0 \parred^n t' \parder^{m+1} t_1$, as required.\qedhere
	\end{itemize}
\end{proof}

The following corollary is immediate:

\begin{corollary}
	[CR for $\parred$, $\RCA$]
	\label{cor:cr-parred}
		If $t_0 \parder^* t \parred^* t_1$ then there is $t'$ s.t.\ $t_0 \parred^* t' \parder^* t_1$.
\end{corollary}

We may finally conclude CR for $\reduces$:
\begin{proof}
	[Proof of Theorem~\ref{thm:cr}]
	Suppose $t_0 \unreduces^* s \reduces^* t_1$. Then, by Corollary~\ref{cor:parred-refltrans-red-refltrans} we have $t_0 \parder^* s \parred^* t_1$.
	By Corollary~\ref{cor:cr-parred} above, we have some $s' $ with $t_0 \parred^* s' \parder^* t_1$, whence $t_0\reduces^* s' \unreduces t_1$ by Corollary~\ref{cor:parred-refltrans-red-refltrans} again.
\end{proof}

From here we derive other properties:
\begin{corollary}
[$\RCA$]
\label{cor:un-for-red-and-conv}
	We have the following:
\begin{enumerate}
	\item\label{item:peaks-and-valleys} (Peaks and valleys) $s\conv t $ if and only if $  \exists r.\, s\reduces^* r \unreduces^* t$. 
	\item\label{item:un-red} (UN for $\reduces$) If $s_0 $ and $s_1$ are normal with $s_0 \unreduces^* t \reduces^* s_1$, then $s_0 = s_1$.
	\item\label{item:un-conv} (UN for $\conv$) If $s_0 $ and $s_1$ are normal with $s_0 \conv t \conv s_1$, then $s_0= s_1$.
\end{enumerate}	
\end{corollary}
\begin{proof}
[Proof sketch]
For \eqref{item:peaks-and-valleys}, the right-to-left implication is obvious, so we prove
 the left-to-right implication by induction on the length of a derivation $s\conv t$.
The critical case is when we have $s\conv t' \reduces t$ (or equivalently $s \unreduces s' \conv t$).
By the inductive hypothesis we have some $r'$ s.t.\ $s \reduces^* r' \unreduces^* t'$.
This means that we have $r' \unreduces^* t' \reduces t$ so by confluence, Theorem~\ref{thm:cr}, we have some $r $ s.t.\ $r' \reduces^* r \unreduces^* t$.

\eqref{item:un-red} is immediate from Theorem~\ref{thm:cr} and the definition of normal form.

\eqref{item:un-conv} is immediate from \eqref{item:peaks-and-valleys}, the definition of normal form and \eqref{item:un-red}.
\end{proof}

\subsection{Structure of finitely applied coderivations, constants and variables}

Before presenting the main type structures of this section, let us take a moment to note that we have now proven enough to show that the set of \farc s, under conversion, forms a model of $\T$ without induction and extensionality, even provably within $\RCA$.

\begin{definition}
[\Farc\ structure]
We write $\FARC$ for the type structure of \farc s, defined as follows:
\begin{itemize}
\item $\sigma^\FARC := \{ t: \sigma \ |\  \text{$t$ is a farc} \}$.
\item $\mathsf r^\FARC$ is just $\mathsf r$, for each constant $\mathsf r$.
\item $t \comp^\FARC s$ is just $ts$.
\item $=^\FARC_\sigma$ is just $\conv_\sigma$.
\end{itemize}
\end{definition}

All the structures we consider in this section will be substructures of $\FARC$ that are moreover closed under conversion.
Thus they will inherit the (quantifier-free) theory of $\FARC$, to which end the following result is naturally indispensable:

\begin{theorem}
[$\RCA$]
\label{thm:farcs-are-a-model-of-t-wo-ext}
$\FARC$ is a model of $\T \setminus \{ \extensionality,\induction \}$.
I.e.\ each $\T$ term $t$ of type $\sigma $ is in $\sigma^\FARC$, and moreover $\FARC$ satisfies the axioms of Figures~\ref{fig:equality-axioms}, \ref{fig:eq-ax-seq-calc-min}, \ref{fig:rec-axioms}, \ref{fig:cond-rule+axioms} and the axioms \eqref{item:ax-0-min} and \eqref{item:ax-succ-inj} of Figure~\ref{fig:ax-t-non-eq}.
\end{theorem}
\begin{proof}
For each $\T$ term $t$ of type $\sigma$ we have $t \in \sigma^\FARC$ simply by definition of a \farc, so we continue to verify the axioms.

The axioms governing the constants,
i.e.\ those from Figures~\ref{fig:eq-ax-seq-calc-min}, \ref{fig:rec-axioms} and \ref{fig:cond-rule+axioms}, follow immediately from the definitions of $\reduces$ and $\conv$, cf.~Figure~\ref{fig:reduction-rules}.

For the equality axioms from Figure~\ref{fig:equality-axioms}:
\begin{itemize}
\item Reflexivity of $\conv$ follows by definition.
\item 
For the Leibniz property, suppose that $t \conv t'$ and $r(t) \conv s(t)$.
Then we have $r(t)\conv r(t')$ and $s(t) \conv s(t')$ by closure of $\conv$ under contexts, and the fact that \farc s have only finitely many redexes, so there are only finitely many occurrences of $t$ in $r(t)$ and $s(t)$ (cf.~Corollary~\ref{cor:pres-fin-props-by-reduction}). Thus we have $r(t')\conv s(t') $ by transitivity of $\conv$.
By symmetry, we also have that if $r(t')\conv s(t') $ then $ r(t) \conv s(t) $, and so in general $\phi(t) \equiv \phi(t')$ for any formula $\phi$.
\end{itemize}

For the first two number-theoretic axioms from Figure~\ref{fig:ax-t-non-eq}:
\begin{itemize}
\item $\cnot \succ 0 \conv 0$ follows from confluence, in particular unique normal forms cf.~Corollary~\ref{cor:un-for-red-and-conv}: both $\succ 0 $ and $0$ are normal but are not identical.
\item Suppose $\succ s \conv \succ t$, so $\succ s \reduces^* u\unreduces^* \succ t$, by Corollary~\ref{cor:un-for-red-and-conv}.
Notice that no reduction rule has $\succ$ at the head, so by $\Sigma^0_1$ induction we can extract reduction sequences $s \reduces^* u' \unreduces^* t $ for some $u' $ with $u = \succ u'$. Thus indeed $s \conv t$. \qedhere
\end{itemize}
\end{proof}

\subsection{Hereditarily total coterms under conversion}
We are now ready to present the type structures that will allow us to obtain a simulation of $\nC n$ within $\nT{n+1}$.
The structure that we present in this subsection is essentially the \emph{hereditrarily recursive operations} of finite type, but where we adopt \farc s under conversion as the underlying model of computation, cf.~Proposition~\ref{prop:turing-completeness}.

\begin{definition}
	[Hereditarily total \farc s]
	We define the following sets of \farc s:
	\begin{itemize}
		\item $\hr\nat := \{ t:\nat \  |\  \exists n \in \Nat. \ t \conv \numeral n \}$
		\item $\hr{\sigma \to \tau} := \{ t: \sigma \to \tau \ | \ \forall s \in \hr\sigma . \ t\, s\, \in \hr \tau \}$
	\end{itemize}
	    We write $\hr n$ for the union of all $\hr{\sigma}$ with $\level (\sigma) \leq n$.
\end{definition}

Note that it is immediate from the definition that each $\hr\sigma$ contains only closed \farc s of type $\sigma$.

Notice that, by the confluence result of the previous subsection, namely by Corollary~\ref{cor:un-for-red-and-conv}, if $t\conv \numeral n$ then $n\in \Nat $ is unique and in fact $t \reduces^* \numeral n$ (provably in $\RCA$).
In this way we can view every element of $\hr \nat$ as computing a unique natural number by means of reduction.

\begin{fact}
\label{fact:log-comp-hr-preds}
$\hr \nat$ is $\Sigma^0_1$, and
if $\level(\sigma) = n>0$ then $\hr\sigma$ is $\Pi^0_{n+1}$.
\end{fact}
\begin{proof}
$\hr\nat$ is $\Sigma^0_{1}$ just since $\conv$ is $\Sigma^0_1$. Furthermore, expanding out the definition:
$$\hro{\sigma \to \tau}(t) \iff \forall s . (s \in \hro \sigma  \implies  t s \in \hro \tau)$$
We proceed by the induction on the structure of $\sigma$.

In the base case, when $\sigma = \tau = \nat$, we have that $\hro\sigma $ and $\hro\tau$ are $\Sigma^0_1$, and so $\hro{\sigma \to \tau}$ is indeed $\Pi^0_2$, as required.

For the inductive hypothesis, we have that $\hro\sigma$ and $\hro\tau$ are $\Pi_n$ and $\Pi_{n+1}$ respectively, and so again $\hro{\sigma\to \tau}$ is indeed $\Pi_{n+1}$.
\end{proof}

\begin{proposition}
[Closure properties of $\hr{}$]
\label{prop:hr-closure-props}
Suppose $\level(\sigma)<n$ and $\level (\tau)\leq n$. 
Then we have the following, provably in $\RCA$:
\begin{enumerate}
\item\label{item:hr-closed-under-composition} If $s \in \hr \sigma$ and $t \in \hr{\sigma \to \tau}$ then $ts \in \hr \tau$. {($\hr{}$ closed under application)}
\item\label{item:hr-closed-under-conversion} If $t \in \hr \tau$ and $t\conv t'$ then $t'\in \hr\tau$. {($\hr{}$ closed under conversion)}
\end{enumerate}
\end{proposition}

Note that provability within $\RCA$ above is \emph{non-uniform} in $\sigma$ and $\tau$, i.e.\ $ \RCA$ proves the statements for each particular $\sigma$ and $\tau$.

\begin{proof}
\eqref{item:hr-closed-under-composition} is immediate from the definition of the sets $\hr{\sigma}$.
\eqref{item:hr-closed-under-conversion} follows by (meta-level) induction on the structure of $\tau$:
    \begin{itemize}
        \item The base case, when $\tau = \nat$, follows from symmetry and transitivity of $\conv$.\footnote{Note here that it was important to take conversion, $\conv$, which is symmetric, rather than reduction, $\reduces^*$, for the definition of $\hr\nat$.}
        \item Suppose $\tau = \sigma \to \tau'$ and let $s \in \hro \sigma$. Since $t \conv t'$ we also have $t\, s \ \conv\  t' s$, by closure of $\conv$ under contexts, and so $t's \, \in \hro{\tau'}$, by the inductive hypothesis. Thus $t' \in \hro\tau$, as required.\qedhere
    \end{itemize}
\end{proof}
 
\noindent
These properties justify defining the following type structure:
\begin{definition}
    [$\hr{}$ structure]
    We simply write $\hr{}$ for the type structure defined as follows:
    \begin{itemize}
    \item $\sigma^\hr{}$ is $\hr \sigma$.
    \item $\mathsf r^\hr{}$ is just $\mathsf r$ for each constant $\mathsf r$.
    \item $t\comp^\hr{} s$ is just $ts$.
    \item $=_\sigma^{\hr{}}$ is $\conv_\sigma$.
    \end{itemize}
\end{definition}

Ultimately we will show that this structure constitutes a model of $\C \setminus \extensionality$ (i.e.\ without extensionality).
In fact, in Section~\ref{sect:t-sim-c} we will prove this within an appropriate theory of arithmetic, so that we may also extract terms of $\T$ that are equivalent under conversion.

Note that $\hr{}$ is a substructure of $\FARC$, in the model-theoretic sense, and so inherits its quantifier-free theory (over the respective domains).
The key feature of $\hr{}$ over $\FARC$ is that it satisfies induction, provably in $\RCA$ in the following sense:

\begin{lemma}
[Induction for $\hr{}$, $\RCA$]
\label{lem:ind-for-hr-in-rca}
Suppose $r(x)$ and $s(x)$ are \farc s.
If $r(0) \conv s(0)$ and $\forall t \in \hr \nat. (r(t) \conv s(t) \implies r(\succ t) \conv s(\succ t))$, then $\forall t \in \hr \nat . r(t) \conv s(t)$.
\end{lemma}
\begin{proof}
This is essentially `forced' by the definition of $\hr\nat$, reducing the statement to induction in $\RCA$.
Assuming the premisses, in particular we have $r(\numeral n) \conv s(\numeral n) $ implies $r(\succ \numeral{n}) \conv s(\succ \numeral n)$ for any $n\in \Nat$, and so by $\Sigma^0_1$-induction on $n$ we have $\forall n \in \Nat.\,  r(\numeral n) \conv s(\numeral n)$.
Now, suppose $t \in \hr\nat$. Then by definition we have $t\conv \numeral n$ for some $n\in \Nat$, and so indeed $r(t) \conv s(t)$ by the Leibniz property (inherited from $\FARC$, cf.~Theorem~\ref{thm:farcs-are-a-model-of-t-wo-ext})
\end{proof}

Notice that induction for arbitrary quantifier-free formulas may be duly reduced to the case of equational formulas in the usual way, interpreting Boolean connectives as operations on (co)terms.
To conclude that $\hr{}$ actually constitutes a model of $\T$ (and later of $\C$), without extensionality, we will further have to show that it can interpret each term of $\T$ (and later coterm of $\C$).
For $\T$, this follows from well-known standard results:

\begin{proposition}
\label{prop:hr-models-t-wo-ext}
$\hr{}$ is a model of $\T \setminus \extensionality$.
\end{proposition}
\begin{proof}[Proof sketch]
Given Theorem~\ref{thm:farcs-are-a-model-of-t-wo-ext} and Lemma~\ref{lem:ind-for-hr-in-rca} above, it remains to show that for each closed term $t$ of type $\tau$ in $\T$, we indeed have that $t \in \hr \tau$. This part of the argument is in fact what requires significant logical complexity, since it implies the consistency of $\T$ (and so also $\PA$ and $\ACA$), and is what cannot be directly transferred to the coterm setting.
Given a term of type $\tau = \vec \tau \to \nat$ and $\vec t \in \hr{\vec \tau}$, we have that $t\, \vec t$ converts to a unique numeral, thanks to Tait's result Fact~\ref{fact:normalisation+confluence-of-t}, and since the only normal terms of type $\nat$ are numerals.
Thus $t \, \vec t \, \in \hr\nat$, and so $t \in \hr \tau$ by repeatedly unfolding its definition.
\end{proof}

In fact this proof can be formalised non-uniformly in the following sense: for each term $t$ of type $\tau$ with $\level(\sigma) \leq n$, we have $\RCA + \CIND{\Sigma^0_{n+1}} \proves \hr \tau (t)$.
We will see a similar situation for membership of $\nC n$ coterms in $\hr n$ later, but with the quantifier complexity of induction increased by $1$. 

\subsection{Modelling extensionality via collapse of conversion}
One issue with the $\hr{}$ structure is that equality is not extensional, and so it is not closed under $\extensionality$.
In particular two distinct coderivations may compute the same function, but are already in normal form if we have not fed any inputs.
A very simple example already at type 1 are the coderivations $\cond\, 0\, 0\, : \nat \to \nat \to \nat$ vs $\cond'\, 0\, 0\, : \nat \to \nat \to \nat$, where $\cond$ examines the first input and $\cond' $ examines the second input. 
Both return $0$, but the two coderivations are not convertible.
Naturally more complex and nontrivial examples abound, e.g.\ distinct programs for sorting, etc.

However, we may recover an extensional equality relation in a standard way thanks to the notion of extensional collapse from higher-order computability theory, cf.~\cite{ho-computability}.

\begin{definition}
[Extensional equality and $\he{}$-structure]
We define the following relations $\exteq_\sigma$ on coterms.
\begin{itemize}
\item $\exteq_\nat$ is just $\conv_\nat$.
\item $t \exteq_{\sigma \to \tau} t'$ if $\forall s \in \hr{\sigma}.\ t\,s \exteq t's$
\end{itemize} 
We define the type structure $\he{}$ just as $\hr{}$, but with $=_\sigma$ interpreted by $\exteq_\sigma$.
\end{definition}

The price to pay for extensionality, however, is high: the equality relation is no longer semi-recursive, and its logical complexity grows with type level.
\begin{fact}
\label{fact:exteq-log-comp}
$\exteq_\nat $ is $\Sigma^0_1$ and,
if $\level(\tau) = n>0$, then $\exteq_\tau$ is $\Pi^0_{n+1}$.
\end{fact}
\begin{proof}
The case of $\nat$ follows immediately from the definition, and for $\level(\tau)>0$ we proceed by induction on the structure of $\tau = \rho \to \sigma$ with $\level(\rho)<n$ and $\level (\sigma)\leq n$.
\begin{itemize}
\item If $\rho = \sigma = \nat$, so $n=1$, then $t\exteq_\tau t'$ iff $\forall s (s \in \hr \nat \, \cimp \, ts \conv t's )$. We have $s \in \hr\nat $ is $\Sigma^0_1$ by Fact~\ref{fact:log-comp-hr-preds} and $\conv$ is also $\Sigma^0_1$, so indeed $t \exteq_\tau t'$ is $\Pi^0_2$.
\item For the inductive step, $t\exteq_\tau t'$ iff $\forall r (r \in \hr\rho\, \cimp \, tr \exteq_\sigma t'r)$. We have $r \in \hr \rho $ is $\Pi^0_n$ by Fact~\ref{fact:log-comp-hr-preds} and $\exteq_\sigma$ is $\Pi^0_{n}$ by inductive hypothesis, so indeed $t\exteq_\tau t'$ is $\Pi^0_{n+1}$. \qedhere
\end{itemize}
\end{proof}

Note that we have the following closure properties in $\he{}$:
\begin{proposition}
\label{prop:he-closure-conds}
Suppose $t \in \hr\tau$. Then:
\begin{enumerate}
\item  if $t \conv_\tau t' $ then $ t \exteq_\tau t'$. ($\exteq$ coarser than $\conv$)
\item if $t\exteq_\tau t'$ then $t'\in \hr\tau$. ($\hr{}$ closed under $\exteq$)
\end{enumerate}
\end{proposition}
\begin{proof}
    We proceed by induction on the structure of $\tau$.
    \begin{itemize}
        \item If $\tau = \nat$ then the statements are immediate from the equivalence of $\conv_\nat$ and $\exteq_\nat$ and symmetry/transitivity.
        \item Suppose $\tau = \rho \to \sigma$ and let $r \in \hr \rho$. 
        \begin{enumerate}
        \item If $t \conv_\tau t'$ then also $t\, r \ \conv \ t'r$, by closure of $\conv$ under contexts. Therefore $t\, r \exteq_\sigma t'r$, by the inductive hypothesis, and so $t \exteq_\tau t'$.
        \item If $t\exteq_\tau t'$ then $t\,r\ \exteq_\sigma \ t' r$, by definition.
        Therefore $t'r \in \hr\sigma$ by the inductive hypothesis,
        and so
        $t' \in \hr \tau$. \qedhere
        \end{enumerate}
    \end{itemize}
\end{proof}

Consequently, $\he{}$ is in turn a substructure of $\hr{}$ and inherits its quantifier-free theory.
Notice that this does not a priori cover satisfaction of induction, but nonetheless,
as promised, we have the following:

\begin{proposition}
\label{prop:he-models-t}
$\he{}$ is a model of $\T$.
\end{proposition}
\begin{proof}
[Proof sketch]
Since $\he{}$ is a substructure of $\FARC$, it remains to verify $\extensionality$ and $\induction$.

Extensionality is `forced' by the definition of $\exteq$. Suppose $t, t'\in \hr{\sigma \to \tau}$ and for all $s \in \hr \sigma$ we have $ts \exteq t's$. 
Then by definition we have $t \exteq t'$.

We may show that $\he{}$ satisfies induction in a similar way as we did for $\hr{}$, Lemma~\ref{lem:ind-for-hr-in-rca}, though note that this is no longer possible in $\RCA$, due to the increasing complexity of $\exteq_\tau$ with the level of $\tau$.
Supose $r(0) \exteq s(0)$ and for all $t\in \hr \nat$ we have $r(t) \exteq s(t) \implies r(\succ t) \exteq s(\succ t)$.
In particular we have that $r(\numeral n) \exteq s(\numeral n) \implies r(\succ \numeral n) \exteq s(\succ \numeral n)$, for any $n \in \Nat$.
We thus have for all $n\in \Nat$ that $r(\numeral n) \exteq s(\numeral n)$ by $\Pi^0_{k+1}$-induction on $n$ (where $\level(r),\level(s)\leq k$).
Now, suppose $t\in \hr\nat$. Then by definition we have $t \conv \numeral n$, and so $t\exteq \numeral n$, for some $n\ \in \Nat$, so indeed $r(t) \exteq s(t)$ by the Leibniz property (inherited from $\FARC$, Theorem~\ref{thm:farcs-are-a-model-of-t-wo-ext}).
\end{proof}

\section{From $\nC n$ to $\nT{n+1}$, via arithmetisation of models}
\label{sect:t-sim-c}
In this section we will present a converse result to that of Section~\ref{sect:c-sim-t}, i.e.\ that terms of $\nT {n+1}$ may simulate terms of $\nC n$, satisfying the same type 1 quantifier-free theory.
The methodology of Section~\ref{sect:c-sim-t} was entirely proof theoretic, thanks in part to the well-foundedness of the $\T$-terms that were simulated.
A priori, we do not admit a similar methodology for the converse direction, due to the non-wellfoundedness of coterms, and so we employ a model-theoretic approach.
Namely, we will show that the type structures $\hr{}$ and $\he{}$ introduced in the previous section indeed constitute models of $\C \setminus \extensionality$ and $\C$, respectively.
In fact, we will formalise the membership of $\nC n$ coterms in $\hr {n+1}$ within the theory $\RCA + \CIND{\Sigma^0_{n+2}}$ (non-uniformly), whence we obtain explicit equivalent terms of $\nT{n+1}$ by proof mining.

Throughout this section we continue to work only with coterms that are finite applications of coderivations, variables and constants (\farc s). We will work mainly within $\RCA + \CIND{\Sigma^0_{n+1}}$ or $\RCA + \CIND{\Sigma^0_{n+2}}$,
to be clearly indicated at the appropriate points.

\subsection{Canonical branches of non-total coterms}
As for $\T$, the main difficulty in showing that $\hr{}$ or $\he{}$ is a model of $\C$ is in showing that each $\C$ coterm is, indeed, interpreted in the structure. 
For us, this will amount to a formalised proof of the totality of (regular) progressing coterms.
Our approach will be to import a suitable version of the proof of Proposition~\ref{prop:termination} but relativise all the quantifiers, there in the standard model, to their respective domains in $\hr{}$.

First let us note that $\hr{}$ is closed under the typing rules of $\C$:
\begin{observation}
\label{obs:preservation-of-hr}
	Consider a rule instance as follows, with $k\leq 2$:
	\[
	\vliiinf{\mathsf r}{}{\vec \sigma \seqar \tau}{\vec \sigma_1 \seqar \tau_1}{\cdots}{\vec \sigma_k \seqar \tau_k}
	\]
	If $t_i \in \hr{\vec \sigma_i \to \tau_i}$ then $\mathsf r \ t_1 \ \cdots \ t_n  \in \hr{\vec \sigma \to \tau}$.
\end{observation}
This follows by simple inspection of the rules of $\nC{}$, and we shall indeed give a proof of a more refined statement shortly.
As a consequence, by contraposition, any coderivation $\notin \hr{}$ must induce an infinite branch of coderivations $\notin \hr{}$, similarly to the proof of Propostion~\ref{prop:termination}.
The next definition formalises a canonical such branch, as induced by an input on which the coderivation is non-hereditarily-total.
We shall present just the definition of the branch first, and then argue that it is well-defined, for each explicit $\nT n$ coderivation, in $\RCA + \CIND{\Sigma^0_{n+2}}$.

\begin{definition}
	[Branch generated by a non-total input]
	\label{dfn:canonical-non-total-branch}
	Let $t_0:\vec \sigma_0 \seqar \tau_0$ be a coderivation and 
	let $\vec s_0 \in \hro{\vec \sigma_0}$ s.t.\ $t \,\vec s \, \notin \hro\tau$.
	We define the branch $(t_i: \vec \sigma_i \seqar \tau_i)_{i\geq 0}$ and inputs $\vec s_i \in \hro{\vec \sigma_i}$, \emph{generated} by $t_0$ and $\vec s_0$ as follows, always satisfying the invariant $t_i \, \vec s_i \, \notin \hro{\tau_i}$:
	\begin{enumerate}
		\item\label{item:can-br-initial} ($t_i$ cannot be an initial sequent).
		\item\label{item:can-br-wk} Suppose $t_i = \vlderivation{
		\vlin{\wk}{}{\vec \sigma, \sigma \seqar \tau }{\vltr{t}{\vec \sigma \seqar \tau}{\vlhy{\ }}{\vlhy{}}{\vlhy{\ }}}
		}$ and $\vec s_i = (\vec s,s)$.
		Then $t_{i+1} := t$ and $\vec s_{i+1} := \vec s$.
		\item\label{item:can-br-ex} Suppose $t_i = \vlderivation{
			\vlin{\exch}{}{\vec \rho,\rho, \sigma,  \vec \sigma \seqar \tau}{
			\vltr{t}{\vec \rho,  \sigma, \rho,\vec \sigma \seqar \tau }{\vlhy{\ }}{\vlhy{}}{\vlhy{\ }}
			}
		}$ and $\vec s_i = (\vec r, r, s, \vec s)$. Then $t_{i+1} := t$ and $\vec s_{i+1} := (\vec r, s, r, \vec s)$.
		\item\label{item:can-br-cntr} Suppose $t_i = 
		\vlderivation{
			\vlin{\cntr}{}{\vec \sigma, \sigma \seqar \tau}{
			\vltr{t}{\vec \sigma, \sigma, \sigma \seqar \tau}{\vlhy{\ }}{\vlhy{}}{\vlhy{\ }}
			}
		}
		$ and $\vec s_i = (\vec s, s)$. Then $t_{i+1}:= t$ and $\vec s_{i+1} := (\vec s, s, s)$.
		\item\label{item:can-br-cut} Suppose $t_i = 
		\vlderivation{
		\vliin{\cut}{}{\vec \sigma \seqar \tau}{
			\vltr{t}{\vec \sigma \seqar \rho}{\vlhy{\ }}{\vlhy{}}{\vlhy{\ }}
		}{
			\vltr{t'}{\vec \sigma, \rho \seqar \tau}{\vlhy{\ }}{\vlhy{}}{\vlhy{\ }}
		}
		}
		$ and $\vec s_i = \vec s$.
		Then if $t \, \vec s \ \in \hr \rho$ then $t_{i+1}:= t' $ and $\vec s_{i+1} : = (\vec s, t\, \vec s)$. Otherwise, $t_{i+1}:= t$ and $\vec s_{i+1} := \vec s$.
		\item\label{item:can-br-leftimp} Suppose $t_i = 
		\vlderivation{
		\vliin{\leftimp}{}{\vec \sigma, \rho \to \sigma \seqar \tau }{
			\vltr{t}{\vec \sigma \seqar \rho}{\vlhy{\ }}{\vlhy{}}{\vlhy{\ }}	
		}{
			\vltr{t'}{\vec \sigma, \sigma \seqar \tau}{\vlhy{\ }}{\vlhy{ }}{\vlhy{\ }}
		}
		}
		$
		and $\vec s_i = (\vec s, s)$.
		If $t\, \vec s \ \in \hr \rho$ then $t_{i+1}:= t' $ and $\vec s_{i+1}:= (\vec s, s\, (t\, \vec s))$.
		Otherwise $t_{i+1}:= t$ and $\vec s_{i+1}:= \vec s$.
		\item\label{item:can-br-rightimp} Suppose $t_i = 
		\vlderivation{
		\vlin{\rightimp}{}{\vec \sigma \seqar \sigma \to \tau}{
		\vltr{t}{\vec \sigma, \sigma \seqar \tau}{\vlhy{\ }}{\vlhy{}}{\vlhy{\ }}
		}
		}
		 $ and $\vec s_i = \vec s$. Let $s $ be the least\footnote{Recall that, strictly speaking, we assume all our objects are coded by natural numbers in the ambient theory (here fragments of second-order arithmetic). Thus we may always find a `least' object satisfying a property when one exists, by induction on that property.} 
		 element of $\hr \sigma$ such that $t\, \vec s\, s\, \notin \hr\tau$.
		We set $t_{i+1}:= t$ and $\vec s_{i+1} := (\vec s, s)$.
		\item\label{item:can-br-cond} Suppose $t_{i} = 
		\vlderivation{
		\vliin{\cond}{}{\vec \sigma, \nat \seqar \tau}{
			\vltr{t}{\vec \sigma \seqar \tau}{\vlhy{\ }}{\vlhy{}}{\vlhy{\ }}
		}{
			\vltr{t'}{\vec \sigma , \nat \seqar \tau}{\vlhy{\ }}{\vlhy{}}{\vlhy{\ }}
		}
		}
		$ and $\vec s_i = (\vec s ,r)$. 
		If $r\conv 0$ then $t_{i+1} := t$ and $\vec s_{i+1} := \vec s$.
		Otherwise, if $r \conv \succ \numeral n$, then $t_{i+1} := t'$ and $\vec s_{i+1} := (\vec s, \numeral n)$.
	\end{enumerate}
\end{definition}

Note that certain arbitrary choices from the proof of Proposition~\ref{prop:termination} have been made canonical in the definition above.
We also define inputs for $\cut$ and $\leftimp$ directly, to ease the formalisation in arithmetic. 

\begin{proposition}
\label{prop:non-tot-branch-well-defined}
	Let $t_0 : \vec \sigma_0 \seqar \tau_0$ be a fixed coderivation in which all types occurring have level $\leq n$.
	$\RCA + \CIND{\Sigma^0_{n+2}}$ proves the following:
	if $\vec s_0 \in \hr{\vec \sigma_0}$ s.t.\ $t_0\, \vec s_0 \, \notin \hr{\tau_0}$ then the branch $(t_i)_i$ and inputs $(\vec s_i)_i$ generated by $t_0$ and $s_0$ are $\Delta^0_{n+2}$-well-defined.
\end{proposition}
\begin{proof}
Let us write $\branch ( i, (t_0, \vec s_0), (t_i, \vec s_i) )$ for ``$t_i$ and $\vec s_i$ are the $i$\textsuperscript{th} sequent and input tuple generated by $t_0$ and $\vec s_0$''.
Notice that the construction of $t_i$ and $\vec s_i$ itself is recursive in $\hr n$, $t_0$ and $\vec s_0$, and so $\branch$ is certainly recursion-theoretically $\Delta^0_{n+2}(t_0, \vec s_0)$, by appealing to Fact~\ref{fact:log-comp-hr-preds}.
To formally prove that $\branch$ is $\Delta^0_{n+2}$ \emph{inside} our theory, it suffices to show determinism:
\[
\forall i . \forall (t_i, \vec s_i), (t_i', \vec s_i') . 
\left(
\begin{array}{rl}
& \branch ( i, (t_0, \vec s_0), (t_i, \vec s_i) ) \cand \branch ( i, (t_0, \vec s_0), (t_i', \vec s_i') )\\
\implies & t_i = t_i' \cand \vec s_i = \vec s_i'
\end{array}
\right)
\]
Writing $\branch$ syntactically as a $\Sigma^0_{n+2}$ formula, the above may be directly proved by $\Pi^0_{n+2}$-induction on $i$, appealing to the cases of Definition~\ref{dfn:canonical-non-total-branch} above.

It remains to show that the construction is total, i.e.\ that each $(t_i,\vec s_i)$ actually exists.
In fact we will simultaneously prove this and the inductive invariant of the construction, so the formula,
\begin{equation}
\label{eqn:ind-inv-non-constr-br}
 \exists (t_i,\vec s_i) . (\branch(i, (t_0, \vec s_0), (t_i, \vec s_i)) \cand t_i\, \vec s_i\, \notin \hr{\tau_i})
\end{equation}
by induction on $i$.
Note that, since $\level(\tau_i) \leq n$ we have that $\hr{\tau_i}$ is $\Pi^0_{n+1}$ by Fact~\ref{fact:log-comp-hr-preds}, and so $t_i\vec s_i \notin \hr{\tau_i}$ is $\Sigma^0_{n+1}$, whereas $\branch(i, (t_0, \vec s_0), (t_i, \vec s_i))$ is $\Delta^0_{n+2}$ as already mentioned.
Thus the inductive invariant in \eqref{eqn:ind-inv-non-constr-br} is indeed $\Sigma^0_{n+2}$.

First, to justify
\eqref{item:can-br-initial},
let us consider the possible initial sequents:
\begin{itemize}
\item For the $0$ rule: we have $0 \in \hr\nat$ by definition;
\item For the $\succ$ rule: if $t\in \hr\nat$, then $t \conv \numeral n$ for some $n \in \Nat$, by definition of $\hr\nat$, and so also $\succ t \conv \succ \numeral n $, by closure of $\conv$ under contexts. Hence $\succ t \in \hr \nat$.
\item For an $\id_\sigma$ rule: 
 if $s \in \hr\sigma $ then $\id\, s \, \conv \, s$ by $\id$ reduction. Hence $\id\, s\, \in \hr\sigma$.
\end{itemize}

Now, the base case, for $i=0$, follows by the assumption on $t_0$ and $\vec s_0$, so let us assume that $\branch(i, (t_0, \vec s_0), (t_i, \vec s_i))$ and $t_i\, \vec s_i \, \notin \hr{\tau_i}$. 
We will witness the existential of the inductive invariant with the coderivation $t_{i+1}$ and inputs $\vec s_{i+1}$ as given in Definition~\ref{dfn:canonical-non-total-branch} above (justifying their existence when necessary), showing $t_{i+1}\, \vec s_{i+1}\, \notin\hr{\tau_{i+1}}$.
We shall also adopt the same notation for inputs and types as in Definition~\ref{dfn:canonical-non-total-branch}.

For \eqref{item:can-br-wk}, the $\wk$ case, we have:
\[
\begin{array}{rll}
& t_i \, \vec s_i \notin \hr \tau & \text{by inductive hypothesis}\\
\therefore & \wk\, t\, \vec s\, s\, \notin \hr {\tau_i}& \text{by definitions}\\
\therefore & t\, \vec s\, \notin \hr \tau & \text{by $\reduces_\wk$ and closure of $\hr\tau$ under $\conv$}\\
\therefore & t_{i+1}\, \vec s_{i+1} \, \notin \hr{\tau_{i+1}} & \text{by definitions}
\end{array}
\]

For \eqref{item:can-br-ex}, the $\exch$ case, we have:
\[
\begin{array}{rll}
& t_i\, \vec s_i \, \notin \hr {\tau_i} & \text{by inductive hypothesis}\\
\therefore & \exch\, t\, \vec r\, r\, s\, \vec s\, \notin \hr\tau & \text{by definitions}\\
\therefore & t\, \vec r\, s\, r\, \vec s\, \notin \hr\tau & \text{by $\reduces_\exch$ and $\because$ $\hr\tau$ closed under $\conv$} \\
\therefore & t_{i+1}\, \vec s_{i+1} \, \notin \hr{\tau_{i+1}} & \text{by definitions}
\end{array}
\]

For \eqref{item:can-br-cntr}, the $\cntr$ case, we have:
\[
\begin{array}{rll}
& t_i \, \vec s_i \, \notin \hr {\tau_i} & \text{by inductive hypothesis}\\
\therefore & \cntr\, t\, \vec s\, s\, \notin \hr\tau & \text{by definitions}\\
\therefore & t\, \vec s\, s\, s\, \notin \hr\tau & \text{by $\reduces_\cntr$ and $\because$ $\hr\tau$ closed under $\conv$ }\\
\therefore & t_{i+1}\, \vec s_{i+1}\, \notin \hr{\tau_{i+1}} & \text{by definitions}
\end{array}
\]

For \eqref{item:can-br-cut}, the $\cut$ case, assume without loss of generality that $t\, \vec s\, \in \hr\tau$.
We have:
\[
\begin{array}{rll}
& t_i \, \vec s_i \, \notin \hr{\tau_i} & \text{by inductive hypothesis}\\
\therefore & \cut\, t\, t' \vec s\, \notin \hr\tau & \text{by definitions}\\
\therefore & t' \vec s\, (t\, \vec s)\, \notin \hr\tau & \text{by $\reduces_\cut$ and $\because$ $\hr\tau$ closed under $\conv$}\\
\therefore & t_{i+1}\, \vec s_{i+1}\, \notin \hr{\tau_{i+1}} & \text{by definitions}
\end{array}
\]

For \eqref{item:can-br-leftimp}, the $\leftimp$ case, assume without loss of generality that $t\, \vec s\, \in \hr\tau$, and so also $s\, (t\, \vec s)\,  \in \hr\sigma$ by Proposition~\ref{prop:hr-closure-props}.
We have:
\[
\begin{array}{rll}
& t_i \, \vec s_i \, \notin \hr{\tau_i} & \text{by inductive hypothesis}\\
\therefore & \leftimp\, t\, t' \vec s\, s\, \notin \hr\tau & \text{by definitions}\\
\therefore & t' \vec s\, (s\, (t\, \vec s))\, \notin \hr\tau & \text{by $\reduces_\leftimp$ and $\because$ $\hr\tau$ closed under $\conv$}\\
\therefore & t_{i+1}\, \vec s_{i+1}\, \notin \hr{\tau_{i+1}} & \text{by definitions}
\end{array}
\] 

For \eqref{item:can-br-rightimp}, the $\rightimp$ case, we have:
\[
\begin{array}{rll}
& t_i \, \vec s_i \, \notin \hr{\tau_i} & \text{by inductive hypothesis}\\
\therefore & \rightimp\, t\, \vec s\, \notin \hr{\sigma \to \tau} & \text{by definitions}\\
\therefore & \exists s' \in \hr\sigma .\ \rightimp \, t\, \vec s\, s' \notin \hr\tau & \text{by definition of $\hr{\sigma\to\tau}$}\\
\therefore & \exists s' \in \hr\sigma .\ t\, \vec s\, s' \notin \hr\tau & \text{by $\reduces_\rightimp$ and $\because$ $\hr\tau$ closed under $\conv$}\\
\therefore & t\, \vec s\, s\, \notin \hr\tau & \text{$\because$ $s$ is well-defined by $\Sigma^0_{n+1}$-minimisation}\\
\therefore & t_{i+1}\, \vec s_{i+1} \, \notin \hr{\tau_{i+1}} & \text{by definitions}
\end{array}
\] 
In the penultimate step, note that we have from the inductive hypothesis $\exists s (s \in \hr\sigma \, \cand\, t\, \vec s\, s \, \notin \hr \tau )$, where $\level(\sigma)<n$ and $\level (\tau)\leq n$. Thus $(s \in \hr\sigma \, \cand\, t\, \vec s\, s \, \notin \hr \tau )$ is indeed $\Sigma^0_{n+1}$, by Fact~\ref{fact:log-comp-hr-preds}, and so $\Sigma^0_{n+1}$-minimisation applies.

For \eqref{item:can-br-cond}, the $\cond$ case, note by the inductive hypothesis we have $r \in \hr\nat$ so by definition of $\hr\nat$ and confluence, namely Corollary~\ref{cor:un-for-red-and-conv}, we have that $r$ converts to a unique numeral.
Thus the two cases considered by the definition of $t_{i+1}$ and $\vec s_{i+1}$ are exhaustive and exclusive, and we consider each separately.

If $r \conv 0$ then we have:
\[
\begin{array}{rll}
& t_i\, \vec s_i \notin \hr{\tau_i} & \text{by inductive hypothesis}\\
\therefore & \cond \, t\, t' \vec s\, r \, \notin \hr \tau & \text{by definitions}\\
\therefore & \cond\, t\, t' \vec s\, 0 \, \notin \hr\tau & \text{by assumption and $\because$ $\hr\tau$ closed under $\conv$}\\
\therefore & t\, \vec s\, \notin \hr\tau & \text{by $\reduces_\cond$ and $\because$ $\hr\tau$ closed under $\conv$}\\
\therefore & t_{i+1}\, \vec s_{i+1}\, \notin \hr{\tau_{i+1}} & \text{by definitions}
\end{array}
\]

If $r \conv \succ \numeral n$ then we have:
\[
\begin{array}{rll}
& t_i\, \vec s_i \notin \hr{\tau_i} & \text{by inductive hypothesis}\\
\therefore & \cond \, t\, t' \vec s\, r \, \notin \hr \tau & \text{by definitions}\\
\therefore & \cond\, t\, t' \vec s\, \succ \numeral n\, \notin \hr \tau & \text{by assumption and $\because$ $\hr\tau$ closed under $\conv$}\\
\therefore & t' \vec s\, \numeral n \, \notin \hr\tau & \text{by $\reduces_\cond$ and $\because$ $\hr\tau$ closed under $\conv$}\\
\therefore & t_{i+1}\, \vec s_{i+1}\, \notin \hr{\tau_{i+1}} & \text{by definitions}
\end{array}
\]

This concludes the proof.
\end{proof}

\subsection{Progressing coterms are hereditarily total}
We are now ready to show that progressing coterms are hereditarily total, i.e.\ that they belong to $\hr{}$ (and $\he{}$).
Now that we have formalised the infinite `non-total' branches of the proof of Proposition~\ref{prop:termination}, relativised to the type structure $\hr{}$, we continue to formalise the remainder of the argument.

\begin{lemma}
	[$\RCA$]
	\label{lem:mon-dec-thread}
	Let $t_0 : \vec \sigma_0 \seqar \tau_0$ and $\vec s_0 \in \hr{\vec \sigma_0}$ be a coderivation and inputs s.t.\ $t_0\, \vec s_0\, \notin \hr{\tau_0}$.
	Furthermore
	let $(t_i: \vec \sigma_i \seqar \tau_i)_i$ and $\vec s_i \in \hr{\vec \sigma_i}$ be a branch and inputs generated by $t_0$ and $\vec s_0$, satisfying Definition~\ref{dfn:canonical-non-total-branch}.

Suppose some $\nat$-occurrence $\nat^{i+1} \in \vec \sigma_{i+1}$ is an immediate ancestor of some $\nat$-occurrence $\nat^i \in \vec \sigma_i$.
Write $s_i \in \vec s_i$ for the coterm in $\hr \nat$ corresponding to $\nat^i$, and similarly 
$s_{i+1} \in \vec s_{i+1}$ for the coterm $s_{i+1} \in \hr \nat$ corresponding to $\nat^{i+1}$.

If $s_i \conv \numeral n_i$ and $s_{i+1} \conv \numeral n_{i+1}$, for some $n_i,n_{i+1} \in \Nat$, then:
	\begin{enumerate}
	    \item\label{item:mon-dec-thread-non-inc} $n_i \geq n_{i+1}$.
	    \item\label{item:mon-dec-thread-strict-dec} If $\nat^i$ is principal for a $\cond$ step, then $n_i > n_{i+1}$.
	\end{enumerate}
\end{lemma}
\begin{proof}
Follows directly from inspection of Definition~\ref{dfn:canonical-non-total-branch}, and confluence of conversion, namely Corollary~\ref{cor:un-for-red-and-conv}.
\end{proof}

In order to complete our formalisation of the totality argument, we actually have to use an arithmetical approximation of thread progression that nonetheless suffices for our purposes.
The reason for this is that, even though non-total branches are well-defined by Proposition~\ref{prop:non-tot-branch-well-defined}, we do not a priori have access to them as \emph{sets} in extensions of $\RCA$ by fragments of arithmetical induction, and so the lack of progressing threads along them does not directly contradict the fact that a coderivation is progressing.
Notice that this is not an issue in the presence of arithmetical comprehension, i.e.\ in $\ACA$, but in that case logical complexity of defined sets is not a stable notion: all of arithmetical comprehension reduces to $\Pi^0_1$-comprehension.

A similar issue underlies the notion of `arithmetical acceptance' for a non-deterministic automaton in \cite{Das19:log-comp-cyc-arith}.
The fact that our arithmetical approximation suffices is borne out by the `moreover' clause in the following result:

\begin{proposition}
	[$\RCA$]
	\label{prop:no-prog-thread-non-rec-branch}
	Suppose $t_i$ and $\vec s_i$ are as in Lemma~\ref{lem:mon-dec-thread}.
	Any $\nat$-thread along $(t_i)_i$ is not progressing.
	Moreover, $\forall k .  \exists m . $ any $\nat $-thread from $t_k$ progresses $\leq m$ times.
\end{proposition}
\begin{proof}
We shall prove only the `moreover' clause, the former following a fortiori.
First, suppose we have a (finite) $\nat$-thread $(\nat^i)_{i=k}^l$ beginning at $t_k$.
Let $s_i \in \vec s_i$ be the corresponding input of $\nat^i$ for $1\leq i\leq l$, and let each $r_i \conv \numeral n_i$, for unique $n_i\in \Nat$, by definition of $\hr\nat$ and confluence, Corollary~\ref{cor:un-for-red-and-conv}.
Letting $m$ be the number of times that $(\nat^i)_{i=1}^l$ progresses, we may show by induction on $l$ that $n_l \leq n_k - m$, using Lemma~\ref{lem:mon-dec-thread} for the inductive steps.

	Now, to prove the `moreover' statement, fix some $k$ and let $\vec \nat^k \subseteq \vec \sigma_k$ exhaust the $\nat$ occurrences in $\vec \sigma_k$.
Let $\vec r_k \subseteq \vec s_k$ be the corresponding inputs, and write $\vec n_k$ for the unique natural numbers such that each $r_{ki} \conv \numeral{n}_{ki}$, by definition of $\hr\nat$ and confluence, Corollary~\ref{cor:un-for-red-and-conv}.
We may now simply set $m :=  \max \vec n_k$, whence no thread from $t_k$ may progress more than $m$ times by the preceding paragraph.
\end{proof}

Finally, we are ready to show that coterms of $\C$ are indeed interpreted in the type structures we have presented.
It is here that we will have to make use of the fact that, for each regular coderivation, $\RCA$ proves whether it progresses and, moreoever, we shall take advantage of the aforementioned implied arithmetical approximation of progression in order to deduce the necessary contradiction without invoking additional set existence principles.

\begin{theorem}
\label{thm:main-non-unif}	Let $t: \vec \sigma \seqar \tau$ be a progressing coderivation containing only types of level $\leq n$ (i.e.\ a $\nC n$-coderivation). 
	Then $\RCA + \CIND{\Sigma^0_{n+2}} \proves t \in \hr{\vec \sigma \to \tau}$.
\end{theorem}
\begin{proof}
First, by Proposition~\ref{prop:prog-decidability+provability-in-RCA} (from \cite{Das19:log-comp-cyc-arith}), we have that $\RCA$ proves that $t$ is progressing. 
Consequently $\RCA$ proves that, for any branch $(t_i)_i$, there is some $k$ s.t.\ there are arbitrarily often progressing finite threads beginning from $t_k$:\footnote{The argument for this is similar to that of Proposition 6.2 from \cite{Das19:log-comp-cyc-arith}.} 
\begin{equation}
\label{eqn:ar-acc}
\text{$\exists k .  \forall m . $ there is a (finite) $\nat$-thread from $t_k$ progressing $>m$ times}
\end{equation}
Note that this statement is purely arithmetical in $(t_i)_i$ and so, if $(t_i)_i$ is $\Delta^0_{n+2}$-well-defined, then in fact $\RCA + \CIND{\Sigma^0_{n+2}}$ proves \eqref{eqn:ar-acc}, by conservativity over $I\Sigma_{n+2}((t_i)_i)$ and then substitution of the $\Delta_{n+2}$-definition of $(t_i)_i$.

Now, working inside $\RCA+ \CIND{\Sigma^0_{n+2}}$, suppose for contradiction that $\vec s \in \hr{\vec \sigma}$ s.t.\ $t\, \vec s\, \notin \hr{\tau}$.
By Proposition~\ref{prop:non-tot-branch-well-defined}, we can $\Delta^0_{n+2}$-well-define the branch $(t_i)_i$ {generated} by $t$ and $\vec s$.
Thus we indeed have \eqref{eqn:ar-acc}, contradicting Proposition~\ref{prop:no-prog-thread-non-rec-branch}.
\end{proof}

\begin{corollary}
$\hr{}$ is a model of $\C \setminus \extensionality$, and $\he{}$ is a model of $\C$.
\end{corollary}
\begin{proof}
All the axioms of $\C\setminus \extensionality$ and $\C $ are already satisfied in $\hr{}$ and $\he{}$ respectively, inherited from Propositions~\ref{prop:hr-models-t-wo-ext} and \ref{prop:he-models-t} respectively.
Thus the result follows immediately from Theorem~\ref{thm:main-non-unif} above (and soundness).
\end{proof}

\begin{corollary}
If $t:\vec \nat \seqar \nat$ is a progressing coterm of $\nC n$, then there is a $\nT{n+1}$-term $t:\vec \nat \to \nat$ such that $\interp{t'} = \interp t$.
\end{corollary}
\begin{proof}
By Theorem~\ref{thm:main-non-unif} we have, in particular, that:
\begin{equation}
\label{eqn:type1-hr-unwound}
\RCA + \CIND{\Sigma^0_{n+2}} \proves \forall \vec m \exists n \, t\vec {\numeral m} \conv \numeral n
\end{equation}
Since this is a $\Pi^0_2$ theorem, we may apply extraction, Proposition~\ref{prop:parsons-arith-to-t}, to obtain the required term $t'$ of $\nT{n+1}$ so that $t'\vec m$ witnesses the existential quantifier in \eqref{eqn:type1-hr-unwound}.
The result now follows by soundness of $\RCA + \CIND{\Sigma^0_{n+2}}$ and since $(\hr,\conv) $ is a substructure of the standard model $\nmod$.
\end{proof}

\subsection{Interpretation of $\nC n (- \extensionality)$ into $\nT{n+1} (-\extensionality)$}
We may now realise our model-theoretic results as bona fide interpretations of fragments of $\C$ into fragments of $\T$.
As a word of warning, coterms of $\C$ in this section, when operating inside $\T$, should formally be understood by their G\"odel {codes}, i.e.\ in this section $\T$ is `one meta-level higher' than $\C$.
Until now we have been formalising the metatheory of $\C$ within second-order arithmetic, and so arithmetising its syntax as natural numbers.
Since we will here invoke program extraction from these fragments of $\ACA$ to fragments of $\T$ to interpret $\C$, the same coding carries over.
At the risk of confusion, we shall suppress this formality in the statements of results that follow, in line with the exposition so far.

\begin{theorem}
If $\nC n \setminus \extensionality \proves s=t$ then $\nT{n+1} \setminus \extensionality \proves s \conv t$.
\end{theorem}

\begin{proof}
Let us work in $\RCA+ \CIND{\Sigma^0_{n+2}}$.
 By Theorem~\ref{thm:main-non-unif} we have that $s,t \in \hr\sigma$, so suppose that $\nC n \setminus \extensionality \proves s=t$ (which is a $\Sigma^0_1$ relation).
 Now, invoking Theorem~\ref{thm:farcs-are-a-model-of-t-wo-ext} and Lemma~\ref{lem:ind-for-hr-in-rca}, we indeed have that $s \conv t$, by $\Sigma^0_1$-induction on the $\nC n \setminus \extensionality$ proof of $s=t$.
 
 Now, invoking the extraction theorem, Proposition~\ref{prop:parsons-arith-to-t}, for the above paragraph, we can extract a $\nT {n+1}$-term $d(\cdot)$ witnessing the following `reflection' principle:
 \[
 \nT{n+1} \setminus \extensionality \proves \text{``$P$ is a $\nC n \setminus \extensionality$ proof of $s=t$''} \ \cimp \ d(P): s \conv t
 \]
 We may duly substitute a concrete $\nC n \setminus \extensionality$ proof $P$ of $s=t$ into the above principle to conclude that $\nT{n+1} \setminus \extensionality \proves s \conv t$, as required.
\end{proof}

We also have an analogous statement in the presence of the extensionality rule though, as expected, we must restrict to only type $1$ equations for the interpretation to hold:
\begin{theorem}
Let $\level(s)= \level(t)\leq 1$. If $\nC n \proves s=t$ then $\nT{n+1} \proves  s \exteq t$.
\end{theorem}
\begin{proof}
The proof is similar to the one above, only using the proof of Proposition~\ref{prop:he-models-t} instead of Lemma~\ref{lem:ind-for-hr-in-rca} to simulate induction steps of $\nC n$ in $\he{}$ (now requiring $\CIND{\Sigma^0_{n+2}}$).
The restriction on levels is required in order to invoke the extraction theorem, since the logical complexity of $\exteq_\sigma$ grows with the type level of $\sigma$.
In particular, $\exteq_1$ is a $\Pi^0_2$ relation, by Fact~\ref{fact:exteq-log-comp}, and unwinding its definition we have $s \exteq_1 t$ is equivalent to $\forall \vec r\, \forall \vec n\, \forall \vec d: {\vec r} \conv {\vec{ \numeral{n}}}  . \  \exists d :\,  {\vec r} \, \conv {t}\, {\vec r}$. 
Thus we extract the following `reflection' principle for some $\nT{n+1}$-term $d$:
\[
\nT{n+1} \proves (\text{``$P$ is a $\nC n$ proof of $s=t$''} \, \cand\, \vec d: {\vec r}\conv {\numeral{\vec n}})\ \cimp\  d(P,\vec d,\vec r) :\, {s}\, {\vec r}\, \conv \, {t}\, {\vec r}
\]
Again substituting a concrete $\nC n $ proof $P$ of $s=t$ into the above principle indeed gives us $\nT{n+1} \proves  s \exteq t$, as required.
\end{proof}

\section{Perspectives and further results}
\label{sect:perspectives}

In this section we shall give some further discussion and results related to the system $\C$ we have presented.

\subsection{On confluence and consistency}
We should point out that, from the point of view of just the extensional properties of extracted programs, it is not necessary to use the confluence result we presented in Section~\ref{sect:coterm-models}.
We could simply assume consistency of $\C$ as an axiom throughout Section~\ref{sect:t-sim-c} (indeed we shall take this direction in the next subsection for simplicity).
Consistency of $\C$ is a true $\Pi^0_1$ statement by meta-level reasoning, and we would be able to extract the same functionals, since $\Pi^0_1$ statements carry no computational content.
This is particularly pertinent when extracting an infinitely descending sequence of natural numbers from a progressing thread in Lemma~\ref{lem:mon-dec-thread}: we require these natural numbers to be uniquely defined, which follows by either confluence or consistency.
However this approach would compromise our ultimate interpretations of $\nC n$, requiring the same consistency principle to be added to the target theory $\nT{n+1}$ (with and without $\extensionality$).

Incidentally, note that our approach, by a special case of the simulations from Sections~\ref{sect:c-sim-t} and \ref{sect:t-sim-c}, implies the \emph{equiconsistency} of $\nC n$ and $\nT{n+1}$, over a weak base theory.

\subsection{Continuity at type 2}
It is well-known that the type 2 functionals of $\T$ are \emph{continuous}, in the sense that any type 1 function input is only queried a finite number of times.
The classical way to prove this is to augment the rewrite system from Section~\ref{sect:coterm-models} by fresh type 1 constant symbols that play the roles of the inputs, say of a type 2 functional, and adding appropriate reductions.
An account of this is given in \cite{Troelstra73:metamathematical-investigations}, though there are several other known arguments, e.g.\ by showing that the strucutre of \emph{total continuous functionals} forms a model for $\T$ (cf.~\cite{scarpellini71:model-barrec-higher-types}) or, more recently, via an elegant form of `syntactic continuity predicate' (cf.~\cite{xu20:continutity-predicate}).

For the case of $\C$, we may actually formalise a variation of this argument within second-order arithmetic, extending the simulation of $\C$ coterms within $\T$ to type 2 functionals.
We shall here refrain from an analysis of abstraction complexity, for the sake of brevity, and also focus solely on the interpretation of (co)terms in the standard model, disregarding their theories.
Our exposition will be brief, since the finer details are adaptations of earlier results.

Let us fix a $\C$ coderivation $t: \vec \sigma \seqar \nat$ s.t.\ each $\sigma_i = \overbrace{\nat \to \cdots \to \nat}^{k_i} \to \nat$, and let us henceforth work in $\ACA$, distinguishing second-order variables $f_i : \overbrace{\Nat \times \cdots \times \Nat}^{k_i} \to \Nat$, intuitively representing the inputs for $t$.

Within $\C$, introduce new (uninterpreted) constant symbols $\numeral f_i : \overbrace{\nat \to \cdots \to \nat}^{k_i} \to \nat$ for each $\sigma_i$, and new reduction steps:
\begin{equation}
\label{eqn:oracle-reduction}
\numeral f_i\, \numeral n_1\, \dots\, \numeral n_{k_i}\, \reduces\, \numeral{f_i (n_1, \dots, n_{k_i})}
\end{equation}
Notice that reduction is now still semi-recursive in the oracles $\vec f$, i.e.\ $\reduces, \reduces^*, \conv$ are now $\Sigma^0_1(\vec f)$.
To save the effort of reproving our confluence results from Section~\ref{sect:coterm-models} with these new oracle symbols, we shall simply henceforth assume a suitable consistency principle: 
\[
\Con \quad : \quad   \forall m, n . \, (\numeral m \conv \numeral n \ \cimp \ m=n )
\]
Note that, since this is a true $\Pi^0_1 $ statement, it carries no computational content and adding it to $\ACA$ still admits extraction into $\T$.
(The drawback of this, as mentioned in the previous subsection, is that we do not recover any bona fide interpretation of $\C$ into $\T$.)

From here, we define $\hr\sigma^{\vec { f}}$ just as $\hr \sigma$, but allowing coterms to include the symbols $\vec{\numeral f}$.
Since each $\hr\sigma$ is arithmetical in $\reduces$, we have that each $\hr \sigma^{\vec f}$ is arithmetical in our extended reduction relation, so with free second-order variables $\vec f$.
Note in particular that we have that each $\numeral f_i \in \hr {\sigma_i}^{\vec f}$, thanks to \eqref{eqn:oracle-reduction} above.
By adapting our approach from Section~\ref{sect:t-sim-c}, we may show the following:
\begin{theorem}
[$\ACA + \Con$]
 $\forall \vec f .\, t\, \numeral{\vec f}\, \in \hr\nat^{\vec f}  $
\end{theorem}
\begin{proof}
[Proof sketch]
The argument is essentially the same as that for Theorem~\ref{thm:main-non-unif}.
Assuming otherwise, for contradiction, we may generate a non-hereditarily-total branch is just as in
Definition~\ref{dfn:canonical-non-total-branch}, and its well-definedness is shown just as in Proposition~\ref{prop:non-tot-branch-well-defined}.
Note that all induction/minimisation used is in fact arithmetical in $\reduces$ and $\hr{\sigma}^{\vec f}$, so the branch is indeed $\Delta^0_{n+2}(\vec f)$-well-defined (for $n$ the maximal type level in $t$).

Since we no longer concern ourselves with the refinement of type levels, the remainder of the argument is actually simpler than that of Section~\ref{sect:t-sim-c}.
Instead of dealing with the arithmetical approximation of progressiveness, we may immediately access the generated non-total branch \emph{as a set}, thanks to the availability of arithmetical comprehension in $\ACA$.
We also have a suitable version of Lemma~\ref{lem:mon-dec-thread} for $\hr{\nat}^{\vec f}$, this time using $\Con$ above instead of confluence, and so the appropriate contradiction of the well-ordering property of $\Nat$ is readily obtained.
\end{proof}

Expanding out this result we have that $\ACA + \Con \proves \forall \vec f . \exists n.\, t\, \vec f \, \conv \numeral n$.
Note that this yields the required syntactic continuity property:
since any $\conv$-sequence is finite, we may compute $t(\vec f)$ by querying each $f_i$ only finitely many times.
 
From here, by applying a relativised version of program extraction (see, e.g., \cite{kohlenbach08:applied-proof-theory}), we may witness the existential by a term $t'(\vec f)$ of $\T$, so that $\T + \Con \proves\, t\, \vec f\, \conv \, \numeral{t'(\vec f)}$.
Now, in the standard model $\nmod$, we have that $\Con$ is true and that $\conv$ is sound for equality (i.e.\ $s\conv t \implies s^\nmod = t^\nmod$), so we finally have:
\begin{corollary}
If $t$ is a level 2 coterm of $\C$, then there is a $\T$ term $t'$ s.t.\ $\interp{t'} = \interp t$.
\end{corollary}

It would be interesting to see if we could adapt this approach to give models of $\C$ based on continuous functionals at higher types (cf.~\cite{ho-computability}), but such a development is beyond the scope of this work.
We point out that forms of continuity at higher types for $\T$ are less canonical, cf.~\cite{Troelstra73:metamathematical-investigations}.
Nonetheless, we shall show in Section~\ref{sect:crf=prf} that, for every $\C$ coterm there is a $\T$ term computing the same functional (in $\nmod$).

\subsection{A `term model' \`a la Tait and strong normalisation}
It is an immediate consequence of our results that $\C$-coterms are \emph{weakly normalising}.

\begin{proposition}
\label{prop:hr-wn}
If $t\in \hr{}$ then $t$ is weakly normalising.
\end{proposition}
\begin{proof}
[Proof sketch]
We proceed by induction on the type of $t$:
\begin{itemize}
\item For $t$ of type $\nat$, we rely on the confluence result, Theorem~\ref{thm:cr}, and normality of numerals.
\item If $t$ has type $\sigma \to \tau$, then for some/any $s \in \hr\sigma$ we have that $ts \in \hr\tau$, by definition, and so is weakly normalising by the inductive hypothesis. We define a new normalisation sequence for $t$ from one for $ts$ by induction on its length, simply ignoring reductions that are not entirely inside $t$.
\qedhere
\end{itemize}
\end{proof}
Thus, by Theorem~\ref{thm:main-non-unif}, we have:
\begin{corollary}
Each closed $\C$ coterm is weakly normalising.
Moreover, any $\nC n$ coterm is provably weakly normalising inside $\RCA + \CIND{\Sigma^0_{n+2}}$.
\end{corollary}
Given that we also have a confluence for $\C$, we are not far from a strong normalisation result.
It would be interesting if we could define an \emph{increasing} measure for reduction to this end, perhaps induced by the progressing thread criterion.
Instead, we show that Tait's `convertibility' predicates may be suitably adapted for this purpose, yielding a minimal model for $\C$.
We will not formalise our exposition within arithmetic, but expect it to go through in a suitable fragment of second-order arithmetic.

 \smallskip
 
We will define a minimal `coterm model' in a similar way to Tait's term models of sytem $\T$ \cite{Tait:67:normalisation-of-t-+-bar-rec-typ01}.
This is complementary to our development of $\hr{}$ and $\he{}$: while those structures were `over-approximations' of the language of $\C$, the structure we are about to define is an `under-approximation', by virtue of its definition.
Naturally, the point is to show that the approximation is, in fact, tight.

\begin{definition}
[Convertibility]
We define the following sets of closed $\C$-coterms:
\begin{itemize}
\item $\cmod\nat := \{t:\nat \ |\ t \text{ is strongly normalising} \} $.
\item $\cmod{\sigma \to \tau} := \{ t:\sigma \to \tau \ | \ \forall s \in \cmod \sigma. \, ts \in \cmod \tau \}$.
\end{itemize} 
\end{definition}

We can adapt suitable versions of Proposition~\ref{prop:hr-wn} and Proposition~\ref{prop:hr-closure-props} to the setting of $\cmod{}$:
\begin{proposition}
\label{prop:cmod-properties}
We have the following:
\begin{enumerate}
\item\label{item:c-closed-under-application} If $t\in \cmod{\sigma \to \tau}$ and $s \in \cmod \sigma $ then $ts \in \cmod \tau$. ($\cmod{}$ closed under application)
\item\label{item:c-closed-under-reduction} If $t \in \cmod \tau$ and $t \reduces t'$ then $t'\in \cmod \tau$. ($\cmod{}$ closed under reduction)
\item\label{item:c-is-sn} If $t\in \cmod{\tau}$ then $t$ is strongly normalising. ($\cmod{} \subseteq \SN{}$)
\end{enumerate}
\end{proposition}
\begin{proof}
\eqref{item:c-closed-under-application} is immediate from the definition of $\cmod{}$.

\eqref{item:c-closed-under-reduction} is proved by induction on type:
\begin{itemize}
\item Suppose $t\in \cmod\nat$ and $t \reduces t'$. By definition $t$ is strongly normalising, so also $t'$ is strongly normalising, and so $t'\in \cmod \nat$ by definition.
\item Suppose $t\in \cmod{\sigma \to \tau}$ and $t \reduces t'$. Let $s \in \cmod \sigma$. Then $ts \reduces t's$ by closure of $\reduces$ under contexts, and so $t's \in \cmod \tau$ by the inductive hypothesis. Since the choice of $s\in \cmod \sigma$ was arbitrary, we have $t' \in \cmod{\sigma \to \tau}$.
\end{itemize}

Finally, \eqref{item:c-is-sn} is also proved by induction on type:
\begin{itemize}
\item Suppose $t \in \cmod \nat$. Then $t$ is strongly normalising by definition.
\item Suppose $t \in \cmod{\sigma \to \tau}$ and, for contradiction, let $t=t_0 \reduces t_1 \reduces \cdots$ be a diverging reduction sequence. Then for some/any $s\in \cmod \sigma$, we have that $ts = t_0s \reduces t_1s \reduces \cdots $ is also a diverging reduction sequence, contradicting the inductive hypothesis. \qedhere
\end{itemize}
\end{proof}
Note that the strong normalisation condition for $\cmod \nat$ is crucial to justify closure under reduction, \eqref{item:c-closed-under-reduction}, at base type $\nat$. 
In contrast, for $\hr\nat$ we only asked for \emph{conversion} to a numeral, and so the analogous property of closure under conversion was a consequence of symmetry.

Let us call a coterm $t$ \textbf{neutral} if any redex of $ts$ is either entirely in $t$ or entirely in $s$.
We also have the following expected characterisation of convertibility:
\begin{lemma}
[Convertibility lemma]
\label{lem:conversion-lemma}
Let $t$ be neutral.
If $\forall t' \unreduces t . \ t' \in \cmod{\tau}$, then $t \in \cmod \tau$.
\end{lemma}
\begin{proof}
We proceed by induction on type. The base case, for type $\nat$, follows by definition of strong normalisation, so let us assume,
\begin{equation}
\label{eqn:unreduces-to-convertible}
\forall t' \unreduces t . \ t' \in \cmod{\sigma \to \tau}
\end{equation}
and assume that the statement of the Proposition holds for all smaller types (IH), in particular $\sigma$ and $\tau$.

To prove $t \in \cmod{\sigma \to \tau}$, let $s \in \cmod \sigma$ and we show that $ts \in \cmod \tau$.
In fact, we will show,
\begin{equation}
\label{eqn:application-unreduces-to-convertible}
\forall r \unreduces ts .\ r \in \cmod \tau 
\end{equation}
whence $ts \in \cmod \tau$ follows by the inductive hypothesis (IH) for $\tau$.
Since $s \in \cmod \sigma$ and so is strongly normalising by Proposition~\ref{prop:cmod-properties}, we may prove \eqref{eqn:application-unreduces-to-convertible} by a sub-induction on the size of the complete reduction tree of $s$, say $\tree(s)$.\footnote{Since regular progressing coterms are, in particular, finite applications of coderivations, constants and variables (i.e.\ \farc s), there are only finitely many redexes in a $\C$ coterm, and so the reduction tree is finitely branching. By K\"onig's lemma, the complete reduction tree is thus finite.}
\begin{itemize}
\item Suppose $r = t's$, so $t \reduces t'$. Then $t' \in \cmod {\sigma \to \tau}$ by \eqref{eqn:unreduces-to-convertible}, and so $t's = r \in \cmod \tau$.
\item Suppose $t = ts' $, so $s \reduces s'$. Then we have $\tree(s')<\tree(s)$ so by the sub-inductive hypothesis \eqref{eqn:application-unreduces-to-convertible} we have $\forall r' \unreduces ts'  . \ r' \in \cmod \tau$. 
Thus by the main inductive hypothesis (IH) for $\tau$, we have $ts' = r \in \cmod \tau$. 
\end{itemize}
In all cases we have that $r \in \cmod \tau$, yielding \eqref{eqn:application-unreduces-to-convertible} as required.
\end{proof}

Now we can go on to define a non-converting branch, just like we did for the standard model $\nmod$ in Proposition~\ref{prop:termination} (non-total branch), and for $\hr{}$ (also $\he{}$) in Definition~\ref{dfn:canonical-non-total-branch} (non-hereditarily-recursive branch).
As in the latter case, we need to prove well-definedness of such a branch, cf.~Observation~\ref{obs:preservation-of-hr} and Proposition~\ref{prop:non-tot-branch-well-defined}.

\begin{proposition}
[Preservation of convertibility]
\label{prop:rules-preserve-cmod}
Let $\vec r \in \cmod{\vec \rho}$ and $\vec s \in \cmod{\vec \sigma}$. We have the following:\footnote{All rules have type as presented in Figures~\ref{fig:seq-calc-min}, \ref{fig:rules-0-s-rec} and \ref{fig:cond-rule+axioms}.}
\begin{itemize}
\item If $s \in \cmod \sigma $ then $\id\, s \in \cmod \sigma$.
\item If $ r \in \cmod \rho,s \in \cmod \sigma$ and $t\, \vec r\, s\, r\, \vec s\, \in \cmod \tau$ then $\exch\, t\, \vec r\, r\, s\, \vec s \, \in \cmod \tau$.
\item If $ s \in \cmod{\sigma}$ and $t\, \vec s \, \in \cmod \tau$ then $\wk\, t\, \vec s\, s\, \in \cmod \tau$.
\item If $s \in \cmod{\sigma}$ and $t\, \vec s\, s\, s\, \in \cmod \tau$ then $\cntr\, t\, \vec s\, s\, \in \cmod \tau$.
\item If $t_0\, \vec s\, \in \cmod \sigma$ and $\forall s \in \cmod \sigma . \ t_1\, \vec s\, s\, \in \cmod \tau$ then $\cut\, t_0\, t_1\, \vec s\, \in \cmod \tau$.
\item If $r \in \cmod{\rho \to \sigma}$ and $t_0 \, \vec s\, \in \cmod \rho$ and $\forall s \in \cmod \sigma.\ t_1\, \vec s\, s\, \in \cmod \tau$ then $\leftimp\, t_0\, t_1\, \vec s\, r\, \in \cmod \tau$.
\item If $\forall s \in \cmod \sigma .\ t\, \vec s\, s\, \in \cmod \tau$ then $\rightimp\, t\, \vec s\, \in \cmod {\sigma \to \tau}$.
\smallskip
\item $0 \in \cmod \nat$ .
\item If $s\in \cmod \nat$ then $\succ s \in \cmod{\nat}$.
\item If $s \in \cmod \nat$ and $t_0\, \vec s\, \in \cmod \tau$ then $\cond\, t_0\, t_1\, \vec s\, 0\, \in \cmod \tau$.
\item If $s \in \cmod \nat$ and $t_1\, \vec s\, s\, \in \cmod \tau$ then $\cond\, t_0\, t_1\, \vec s\, \succ s\, \in \cmod \tau$.
\end{itemize}
\end{proposition}
\begin{proof}
We proceed by induction on $\tree(\vec s) +\tree(s) + \tree(\vec r) +  \tree( r) $.
In most cases there is a redex at the head, so we shall directly use the conversion lemma, rather showing that any term obtained from reduction is convertible. 
\begin{itemize}
\item $\id\, s \reduces \id\, s' \in \cmod \sigma$ by the inductive hypothesis, and $\id \, s \reduces s \in \cmod\sigma$ by assumption. Thus $\id\, s\, \in \cmod \sigma$ by Lemma~\ref{lem:conversion-lemma}.
\item $\exch\, t\, \vec r\, r\, s\, \vec s\, \reduces\, \exch\, t\, \vec r'\, r'\, s'\, \vec s'\, \in \cmod \tau$ by inductive hypothesis, and $\exch\, t\, \vec r\, r\, s\, \vec s\, \reduces \,  t\, \vec r\, s\, r\, \vec s\, \in \cmod \tau$ by assumption. Thus $\exch\, t\, \vec r\, r\, s\, \vec s\, \in \cmod \tau$ by Lemma~\ref{lem:conversion-lemma}.
\item $\wk\, t\, \vec s\, s\, \reduces \, \wk\, t\, \vec s'\, s'\, \in \cmod \tau$ by inductive hypothesis, and $\wk\, t\, \vec s\, s\, \reduces \, t\, \vec s\, \in \cmod \tau $ by assumption. Thus $\wk\, t\, \vec s\, s\, \in \cmod \tau$ by Lemma~\ref{lem:conversion-lemma}.
\item $\cntr\, t\, \vec s\, s\, \reduces\, \cntr\, t\, \vec s'\, s'\, \in \cmod \tau$ by inductive hypothesis, and $\cntr\, t\, \vec s\, s\, \reduces\, t\, \vec s\, s\, s\, \in \cmod \tau$ by assumption. Thus $\cntr\, t\, \vec s\, s\, \in \cmod \tau$ by Lemma~\ref{lem:conversion-lemma}.
\item $\cut\, t_0\, t_1\, \vec s\, \reduces \cut\, t_0\, t_1\, \vec s'\, \in \cmod \tau$ by inductive hypothesis, and $\cut\, t_0\, t_1\, \vec s\, \reduces t_1\, \vec s\, (t_0\, \vec s)\, \in \cmod \tau$ by assumptions.
Thus $\cut\, t_0\, t_1\, \vec s\, \in \cmod \tau$ by Lemma~\ref{lem:conversion-lemma}.
\item $\leftimp\, t_0\, t_1\, \vec s\, r\, \reduces\, \leftimp\, t_0\, t_1\, \vec s'\, r'\, \in \cmod \tau$ by inductive hypothesis, and $\leftimp\, t_0\, t_1\, \vec s\, r\, \reduces\, t_1\, \vec s\, (r\, (t_0 \, \vec s))\, \in \cmod \tau $ by assumptions and closure of $\cmod{}$ under application, Proposition~\ref{prop:cmod-properties}.\eqref{item:c-closed-under-application}.
Thus $\leftimp\, t_0\, t_1\, \vec s\, r\, \in \cmod \tau$ by Lemma~\ref{lem:conversion-lemma}.
\item To show $\rightimp\, t\, \vec s\, \in \cmod{\sigma \to \tau} $, let $s \in \cmod \sigma$ and we shall show that $\rightimp\, t\, \vec s\, s\, \in \cmod \tau$. We proceed by induction on $\tree (s)$ (as well as $\tree(\vec s)$).
 $\rightimp\, t\, \vec s\, s\, \reduces \rightimp\, t\, \vec s'\, s'\, \in \cmod \tau$ by inductive hypothesis, and $\rightimp\, t\, \vec s\, s\, \reduces\, t\, \vec s\, s\, \in \cmod\tau$ by assumption.
 Thus $\rightimp\, t\, \vec s\, s\, \in \cmod \tau$ by Lemma~\ref{lem:conversion-lemma}, as required.
  \smallskip
 \item $0 \in \cmod \nat$ since it is already normal.
 \item There is no reduction at the head of $\succ s$, so any reduction sequence for $\succ s$ projects to one for $s$, and so terminates by assumption.
 \item $\cond\, t_0\, t_1\, \vec s\, 0\, \reduces \, \cond\, t_0\, t_1\, \vec s'\, 0\, \in \cmod{\tau}$ by inductive hypothesis, and $\cond\, t_0\, t_1\, \vec s\, 0\, \reduces t_0\, \vec s\, \in \cmod \tau$ by assumption.
 Thus $\cond\, t_0\, t_1\, \vec s\, 0\, \in \cmod \tau$ by Lemma~\ref{lem:conversion-lemma}.
 \item $\cond\, t_0\, t_1\, \vec s\, \succ s\, \reduces\, \cond\, t_0\, t_1\, \vec s'\, \succ s'\, \in \cmod{\tau}$ by inductive hypothesis, and $\cond\, t_0\, t_1\, \vec s\, \succ s\, \reduces \, t_1\, \vec s\, s\, \in \cmod \tau$ by assumption.
 Thus $\cond\, t_0\, t_1\, \vec s\, \succ s\, \in \cmod \tau$ by Lemma~\ref{lem:conversion-lemma}. \qedhere
\end{itemize}
\end{proof}

As a consequence of our results in Sections~\ref{sect:coterm-models} and \ref{sect:t-sim-c}, observe the following:
\begin{observation}
\label{obs:cmod-nf-nat-numerals}
If $s \in \cmod \nat$ then $s$ reduces to a unique numeral.
\end{observation}
\begin{proof}
Since $\cmod \nat$ contains \emph{only} $\C$-coterms, we have as a special case of Theorem~\ref{thm:main-non-unif} that $ s \conv \numeral n$ for some $n \in \Nat$.
By confluence, namely Corollary~\ref{cor:un-for-red-and-conv}, we have that $n$ is unique and furthermore $s \reduces^* \numeral n$.
\end{proof}
In fact, both our weak normalisation argument and our confluence result seem to be crucial for establishing the above property of normal elements of $\cmod\nat$, as well as the fact that $\cmod{}$ `under-approximates' the class of $\C$-coterms.
It is otherwise not immediate how we rule out possibilities such as the coterm $\cdots\succ \succ \succ x$, reached at the `limit' of reducing the following (non-progressing) coderivation (applied to $x$):\footnote{Note that no induction on `size' is available for general (non-wellfounded) coterms.}
\begin{equation}
\label{eqn:ssssss}
\vlderivation{
	\vliin{\cut}{\bullet}{\nat \seqar \nat}{
		\vlin{\succ}{}{\nat \seqar \nat}{\vlhy{}}
	}{
		\vlin{\cut}{\bullet}{\nat \seqar \nat}{\vlhy{\vdots}}
	}
}
\end{equation}
This is the same issue that we raised in Section~\ref{sect:on-normality-and-numerality}, and in particular another problematic example was given in Remark~\ref{rmk:non-numeral-normal}, in the form of a coterm $\cond\, 0\, 1\, (\cond\, 0\, 1\, (\cdots ))$.
In contrast to that coterm, where there were two consistent interpretations in the standard model ($0$ or $1$), the coterm $\cdots \succ\succ\succ x$ has no consistent interpretation in $\nmod$.

\begin{theorem}
[Convertibility for $\C$]
Any $\C$-coderivation $t: \vec \sigma \seqar \tau$ is in $\cmod{\vec \sigma \to \tau}$.
\end{theorem}
\begin{proof}
Suppose for contradiction we have $\vec s\in \cmod{\vec \sigma} $ such that $t\, \vec s\, \notin \cmod{\tau}$.
We define a branch $(t_i: \vec \sigma_i \seqar \tau_i)_i$ of $t$ and inputs $\vec s_i \in \cmod{\vec \sigma_i}$ s.t.\ $t_i\, \vec s_i\, \notin \cmod{\tau_i}$ by induction on $i$ just like in Definition~\ref{dfn:canonical-non-total-branch} (or the proof of Proposition~\ref{prop:termination}).
The only difference is that we use Proposition~\ref{prop:rules-preserve-cmod} above for preservation in $\cmod{}$ rather than the analogous closure properties for $\hr{}$ (or $\nmod$).

There is one subtlety, which is the treatment of the $\cond$ case.
Suppose we have a derivation,
\[
\vlderivation{
\vliin{\cond}{}{\vec \sigma, \nat \seqar \tau}{
	\vltr{t}{\vec \sigma \seqar \tau}{\vlhy{\ }}{\vlhy{\ }}{\vlhy{\ }}
}{
	\vltr{t'}{\vec \sigma , \nat \seqar \tau}{\vlhy{\ }}{\vlhy{\ }}{\vlhy{\ }}
}
}
\]
and $\vec s_i = (\vec s,s)$ with $\vec s \in \cmod{\vec \sigma}$, $s\in \cmod \nat$ and $\cond\, t\, t' \vec s\, s\, \notin \cmod \tau$.
Since $s\in \cmod \nat$ we have from Observation~\ref{obs:cmod-nf-nat-numerals} that $s$ reduces to a unique numeral $\numeral n$.
We will show that, 
\begin{itemize}
\item if $n=0$ then $t\, \vec s\, \notin \cmod \tau$; and,
\item if $n = m+1$ then there is some $r \in \cmod \nat$ reducing to $\numeral m$ with $t'\vec s\, r\, \notin \cmod \tau$;
\end{itemize}
by induction on $\tree(\vec s) + \tree(s)$.
By the conversion lemma, Lemma~\ref{lem:conversion-lemma}, there must be a reduction from $\cond\, t\, t'\, \vec s\, s$ not reaching $\cmod \tau$. 
Let us consider the possible cases:
\begin{itemize}
\item If $s = 0 $ and $\cond\, t\, t'\, \vec s\, s \, \reduces\, t\, \vec s\, \notin \cmod\tau$ then we are done.
\item If $s = \succ r$ and $\cond\, t\, t'\, \vec s\, s\, \reduces t' \vec s\, r\, \notin \cmod \tau$ then we are done. (Note that such $r$ must strongly normalise to $\numeral m$, and so in particular $r \in \cmod \nat$).
\item If $\cond\, t\, t'\, \vec s\, s\, \reduces \, \cond t\, t' \vec s' s' \notin \cmod \tau$, then by the inductive hypothesis either,
\begin{itemize}
\item $n=0$ and $t\, \vec s' \notin \cmod \tau$, so $t\, \vec s \notin \cmod \tau$ by Proposition~\ref{prop:cmod-properties}.\eqref{item:c-closed-under-reduction}; or,
\item $n=m+1$ and there is some $r \in \cmod \nat$ reducing to $\numeral m$ s.t.\ $t'\, \vec s'\, r \, \notin \cmod \tau$, so $t' \vec s\, r\, \notin \cmod\tau$ by Proposition~\ref{prop:cmod-properties}.\eqref{item:c-closed-under-reduction}.
\end{itemize}
\end{itemize}

From here, any progressing thread $(\nat^i)_{i\geq k}$ along $(t_i)_i$ yields a sequence of coterms $(r_i \in \cmod\nat)_{i\geq k}$ that, under normalisation, induces an infinitely often descending sequence of natural numbers, yielding the required contradiction. 
\end{proof}

Since $\cmod{}$ is closed under application, Proposition~\ref{prop:cmod-properties}.\eqref{item:c-closed-under-application}, we inherit $\cmod{}$ membership for all $\C$-coterms.
Since elements of $\cmod{}$ are strongly normalising, Proposition~\ref{prop:cmod-properties}.\eqref{item:c-is-sn}, and since reduction is confluent, Theorem~\ref{thm:cr}, we finally have:
\begin{corollary}
[Strong normalisation for $\C$]
Any closed $\C$ coterm strongly normalises to a unique normal form.
\end{corollary}

\subsection{Towards an infinitary $\lambda$-calculus for $\C$ and cut-elimination}
From the point of view of the \emph{Curry-Howard} correspondence, rules of the sequent calculus are usually associated with meta-level term forming operations for the $\lambda$-calculus, rather than comprising constants in their own right.
Let us describe this in some more detail here.

First, the appropriate notion of term is defined as follows:
\begin{itemize}
\item $0$ is a term of type $\nat$.
\item $\pred$ is a term of type $\nat \to \nat$.\footnote{The interpretation of $\pred$ is `predecessor' in the standard model: $\pred^\nmod (0):= 0$ and $\pred^\nmod(n+1) := n$.}
\item If $t$ is a term of type $\sigma \to \tau$ and $s$ is a term of type $\sigma $ then $(ts)$ is a term of type $\tau$.
\item If $t$ is a term of type $\tau$ and $x$ is a variable of type $\sigma$ then $\lambda x t$ is a term of type $\sigma \to \tau$.
\item If $r$ is a term of type $\nat$ and $s,t$ are terms of type $\tau$ then $\ifthelse r s t$ is a term of type $\tau$.
\end{itemize}
From here the typing rules of the sequent calculus can be recast in annotated style, with variable annotations on the LHS and term annotation on the RHS:

\[
\small
		\vlinf{\exch}{}{\vec x:{\vec \rho}, x:{\rho}, y:{\sigma}, \vec y:{\vec \sigma }\seqar t(\vec x,x,y,\vec y): \tau}{\vec x: {\vec \rho},  y: {\sigma}, x: {\rho}, \vec y: {\vec \sigma }\seqar t(\vec x,y,x,\vec y):\tau}
		\qquad
		\vlinf{\cntr}{}{\vec x: {\vec \sigma}, x: {\sigma} \seqar t(\vec x,x,x): \tau}{\vec x: {\vec \sigma}, x_0: {\sigma}, x_1: {\sigma }\seqar t(\vec x, x_0, x_1): \tau}
		\]
		\[
		\small
				\vlinf{\id}{}{x: {\sigma} \seqar x:\sigma}{}
						\qquad
						\vlinf{\wk}{}{\vec x: {\vec \sigma}, x:{\sigma} \seqar t(\vec x): \tau}{\vec x: {\vec \sigma }\seqar t(\vec x): \tau}
				\qquad
						\vlinf{\rightimp}{}{\vec x: {\vec \sigma }\seqar \lambda x . t(\vec x,x) : \sigma \to \tau}{\vec x: {\vec \sigma}, x: \sigma \seqar t(\vec x, x): \tau}
		\]
		\[
		\small
		\vliinf{\cut}{}{\vec x: {\vec \sigma }\seqar t(\vec x, s(\vec x)): \tau}{\vec x: {\vec \sigma} \seqar  s(\vec x): \sigma }{\vec x: {\vec \sigma}, x: \sigma \seqar t(\vec x,x): \tau}
		\qquad
		\vliinf{\leftimp}{}{\vec x:{\vec \sigma }, y:{\rho \to \sigma} \seqar t(\vec x, y\, r(\vec x) ) :  \tau}{\vec x: {\vec \sigma} \seqar r(\vec x):\rho}{\vec x: {\vec \sigma}, x:{\sigma }\seqar t(\vec x,x):\tau}
		\]
		\[
		\small
		\vlinf{0}{}{\seqar 0: \nat}{}
		\qquad
		\vlinf{\pred}{}{x:{\nat} \seqar \pred x: \nat}{}
\]
\[
\small
		\vliinf{\cond}{}{ \vec x: {\vec \sigma  }, y:{\nat }\seqar \ifthelse{y}{s(\vec x)}{t(\vec x,\pred y)} :  \tau}{\vec x: {\vec \sigma }\seqar s(\vec x): \tau}{ x: {\vec \sigma  }, y: {\nat }\seqar t(\vec x, y): \tau}
		\]
		
		\smallskip
		
		As we mentioned from the start, the reason for not pursuing such an approach is that this association of a term to a derivation is not \emph{continuous} and so does not a priori produce a well-defined coterm from arbitrary coderivations.
		Consider, for instance, the operation on terms induced by the $\cut$ rule: we may have to go further along the right branch to eventually print the root of any associated coterm.
		As an extreme example (an adaptation of) the coderivation from \eqref{eqn:ssssss}, 
		\[
		\vlderivation{
			\vliin{\cut}{\bullet}{\nat \seqar \nat}{
				\vlin{\pred}{}{\nat \seqar \nat}{\vlhy{}}
			}{
				\vlin{\cut}{\bullet}{\nat \seqar \nat}{\vlhy{\vdots}}
			}
		}
		\]
		has no well-defined coterm representation induced by the type system above.
		One of the reasons for employing our combinatory approach was simply to avoid such problems; however we also point out that our notion of reduction from Section~\ref{sect:coterm-models} gives us a consistent way to interpret the derivation above as an infinitary coterm: `at the limit' it normalises to $\cdots \pred\pred\pred x$ (after application to $x$).
		It would be interesting to make this idea more formal, in particular using tools from \emph{infinitary term rewriting} (see, e.g., \cite{inf-trs,KetSim11:inf-comb-red-sys,simonsen12:habilitation-thesis}).
		
		We should mention also that the system and (co)terms above, as presented, are \emph{not} strongly normalising, due to the presence of predcessor, $\pred$, instead of successor, $\succ$, and the corresponding $\cond$ rule.
		An appropriate treatment should seek to avoid this issue.
		
		On a similar subject, the Curry-Howard correspondence also motivates the investigation of \emph{cut-elimination} for our type system, cf.~\cite{BaeDouSau16:cut-elim,BaeDouKupSau:bouncing-threads}.
		The relationship to normalisation for our reduction system is not entirely clear, since the way we associate (co)derivations and (co)terms means that reduction is not internal for coderivations, i.e.\ if $t \reduces s$ and $t$ is a coderivation, $s$ is not necessarily a coderivation.
		This is related to the fact that we, strictly speaking, distinguish $\cut$ and application.
		On the other hand it is this distinction that admits the aforementioned continuity property, as well as admitting the possibility of cut-free normal (co)terms.
		The precise relationship between normalisation of coterms and cut-elimination subsumes that for minimal logic, between normalisation for \emph{natural deduction} (or simply typed $\lambda$-calculus) and cut-elimination in the sequent calculus, for which there is a substantial literature. See, e.g., \cite{basic-proof-theory} for a detailed account of such matters.

		We should mention that it would make sense, for the pursuit of cut-elimination, to replace the $\succ$ initial sequent with a corresponding rule, $\vlinf{\succ}{}{\vec \sigma \seqar \nat}{\vec \sigma \seqar \nat}$, with the obvious interpretation, to admit an appropriate cut-reduction against the conditional rule, $\cond$.
		
\subsection{Incorporating fixed point operators}
\label{sect:fixed points}
We could naturally extend our type system by (co)inductive fixed points, similarly to work such as \cite{clairambault2010:muLJ} for intuitionistic logic and \cite{baelde12:mumall} for linear logic. Indeed, as we mentioned in the introduction, the circular proof theory of linear logic type systems with fixed points is increasingly well-developed, e.g.~\cite{BaeDouSau16:cut-elim,DeSaurin:infinets:2019,BaeDouKupSau:bouncing-threads}.

\def\rdots{\rotatebox[origin=l]{29}{$\scriptscriptstyle\ldots\mathstrut$}}

At the level of expressivity, $\T$ already has the capacity to express a range of (co)inductive types. 
In particular, being the type theoretic counterpart to Peano Arithmetic, $\T$ admits recursion on effective well-orders of order type $<\epsilon_0$ (see \cite{Kreisel51,tait68:constructive-reasoning}), and so can encode (co)inductive types of closure ordinal $\omega^{\rdots^\omega}$.\footnote{Note, incidentally, that in a higher typed setting, coinductive types may themselves be encoded as inductive types.}

Similarly to \cite{BaeDouSau16:cut-elim}, extensions of $\C$ with (co)inductive types may be duly defined by demanding progressing threads on the LHS on least fixed points, or on the RHS on greatest fixed points.
In terms of conservativity over $\T$, the pertinent question is whether the encoding of (co)inductive types in $\T$ admits a corresponding coding of circular derivations, in particular preserving ancestry and the progressing thread criterion.
This does not seem to be too technical, but a comprehensive treatment is beyond the scope of this work.

However let us consider one pertinent example arising from the aforementioned work \cite{KupPinPou20:sysT}, where a circular version of $\T$ was presented with a slightly different type language. In particular they include a \emph{Kleene star} operator for list formation, along with the accompanying rules (ancestry indicated by colours):
\[
\vliinf{\lefrul{*}}{}{\purple{\vec \sigma}, \blue{\sigma^*}\seqar \tau}{\purple{\vec \sigma }\seqar \tau}{\purple{\vec \sigma }, \sigma, \blue{\sigma^*}  \seqar \tau}
\qquad
\vlinf{\rigrul{*}}{}{\seqar \tau^*}{}
\qquad
\vliinf{\rigrul*}{}{\purple{\vec \sigma }\seqar \tau^*}{\purple{\vec \sigma }\seqar \tau}{\purple{\vec\sigma} \seqar \tau^*}
\]
The semantics of these constants are intuitive, and explained in \cite{KupPinPou20:sysT}.
Being a least fixed point, the associated progressing thread condition is that each infinite branch has an infinite thread on a $*$-type on the LHS that is infinitely often principal.
Semantically this induces a similar totality argument to ours for $\nat$-threads, cf.~Proposition~\ref{prop:termination}: at a progress point the corresponding \emph{list} decreases in length. 

As suggested above, we may duly encode the fixed point $*$, along with its rules and corresponding notion of progressing coderivation, within our type system for $\C$. 
First, temporarily using product types to ease the exposition, we may embed each type $\sigma^*$ into a type $\nat \times (\nat \to \sigma)$.\footnote{Note that this embedding also induces a well-behaved notion of type level in the presence of Kleene $*$: we may set $\level(\sigma^*) := \level (\sigma)$, as long as $\level(\sigma)>0$.
The type $\nat^*$ may also be embedded at level $0$ using a coding of sequences, but to preserve ancestry, similarly to the current encoding, we should nonetheless encode it as a pair $\nat \times \nat$, the first component still representing the length.} 
Semantically the first component represents the length of the list, and the second component represents a stream from which the list is extracted, where we do not care about the values of elements beyond the length specified by the first component.
Using this embedding we may duly derive the translations of the typing rules above.
The first $\rigrul *$ rule is translated to the following derivation:
\[
\small
\vlderivation{
	\vliin{\rigrul\times}{}{\seqar \nat \times (\nat \to \tau)}{
		\vlin{0}{}{\seqar \nat}{\vlhy{}}
	}{
		\vlin{\rightimp}{}{\seqar \nat \to \tau}{
		\vlin{\wk}{}{\nat \seqar \tau}{
		\vliq{0_\tau}{}{\seqar \tau}{\vlhy{}}
		}
		}
	}
}
\]
Here $0_\tau$ is just the $0$ function of type $\tau$, and so we set the empty $\tau$-list to be just the stream $(0_\tau,0_\tau,\dots)$.
Note that it would not matter if we set it to something else, since the first component tells us to ignore elements beyond the length, in this case the entire stream.
The second $\rigrul *$ rule is translated to the following derivation:
\[
\small
\vlderivation{
	\vliiin{2\cut}{}{\purple{\vec \sigma }\seqar \nat \times (\nat \to \tau)}{
		\vlhy{\purple{\vec \sigma} \seqar \tau}
	}{
		\vlhy{\purple{\vec \sigma }\seqar \nat \times (\nat \to \tau)}
	}{
		\vlin{\lefrul\times}{}{\tau, \nat \times (\nat \to \tau) \seqar \nat \times (\nat \to \tau)}{
		\vliin{\rigrul\times}{}{\tau, \nat, \nat \to \tau \seqar \nat \times (\nat \to \tau)}{
			\vlin{\succ}{}{\nat \seqar \nat}{\vlhy{}}
		}{
			\vlin{\rightimp}{}{\tau, \nat \to \tau \seqar \nat \to \tau}{
			\vliin{\cond}{}{\tau, \nat \to \tau, \nat \seqar \tau}{
				\vlin{\id}{}{\tau \seqar \tau}{\vlhy{}}
			}{
				\vliin{\leftimp}{}{\nat \to \tau, \nat \seqar \tau}{
					\vlin{\id}{}{\nat \seqar \nat}{\vlhy{}}
				}{
					\vlin{\id}{}{\tau \seqar \tau}{\vlhy{}}
				}
			}
			}
		}
		}
	}
}
\]
Finally the $\lefrul *$ rule is translated to the following derivation:
\[
\small
\vlderivation{
	\vliin{\cond}{}{\purple{\vec \sigma}, \blue{\nat}, \nat \to \sigma \seqar \tau}{
		\vlin{\wk}{}{\purple{\vec \sigma} , \nat \to \sigma \seqar \tau}{
		\vlhy{\purple{\vec \sigma }\seqar \tau}
		}
	}{
		\vliin{\cut}{}{\purple{\vec \sigma}, \blue{\nat}, \nat \to \sigma \seqar \tau}{
			\vliin{\leftimp}{}{\nat \to \sigma \seqar \sigma}{
				\vlin{0}{}{\seqar \nat}{\vlhy{}}
			}{
				\vlin{\id}{}{\sigma \seqar \sigma}{\vlhy{}}
			}
		}{
			\vliin{\cut}{}{\vec \sigma, \sigma, \nat , \nat \to \sigma \seqar \tau}{
				\vlin{\rightimp}{}{\nat \to \sigma \seqar \nat \to \sigma}{
				\vliin{\cut}{}{\nat \to \sigma, \nat \seqar \sigma}{
					\vlin{\succ}{}{\nat \seqar \nat}{\vlhy{}}
				}{
					\vliin{\leftimp}{}{\nat \to \sigma, \nat \seqar \sigma}{
						\vlin{\id}{}{\nat \seqar \nat}{\vlhy{}}
					}{
						\vlin{\id}{}{\sigma \seqar \sigma}{\vlhy{}}
					}
				}
				}
			}{
				\vlhy{\purple{\vec \sigma}, \sigma, \blue{\nat }, \nat \to \sigma \seqar \tau}
			}
		}
	}
}
\]
The verification of the semantics of these derivations is left as an exercise to the reader. 
Note that, as indicated by the colouring of type occurrences, ancestry is preserved by this translation.
In particular, for the translation of $\sigma^*$, ancestry \emph{and progressiveness} are preserved on the first component (the blue $\blue \nat$) of a left rule.

It should be straightforward to formalise the ideas above to obtain an embedding of the system from \cite{KupPinPou20:sysT} into ours, thereby inheriting similar results, but a comprehensive development is beyond the scope of this work.
It would be interesting, for future work, to more generally study extensions of $\C$ by suitable fixed points, least and greatest, and in particular show that they may be interpreted back into $\C$ using translations like the ones above.
The key point above was to use the first component (length) to store a recursive parameter on which ancestry is preserved. 
For the case of Kleene $*$ this is straightforward since its closure ordinal is $\omega$, however we must be more careful for fixed points with greater closure ordinals: while we can indeed code recursive ordinals by natural numbers, we must be able to implement `circular recursion' on them in a way that preserves ancestry and progressiveness.
Naturally, some use of higher types should be required, trading off ordinal complexity for abstraction complexity, cf.~\cite{Kreisel51,tait68:constructive-reasoning,Par72:n-quant-ind,Buss1995witness}.

\subsection{`Cyclic recursive functionals' are G\"odel primitive recursive}
\label{sect:crf=prf}

Besides the relationships between $\C$ and $\T$ as theories and in the various models we have discussed, like $\hr{}$ and $ \he{} $, it is natural to ask about their relationships in the standard model $\nmod$. 
In this section we shall show that the interpretations of their terms in $\nmod$ in fact comprise the same algebra of functionals, by reduction of $\C$ computability to some form of higher type recursion on ordinals smaller than $\epsilon_0$.

Recall that the G\"odel primitive recursive functionals, $\PRF{}$, are just the interpretations of $\T$ terms in $\nmod$, i.e.\ $\{ \interp t \ | \ \text{$t$ a $\T$-term} \}$.
Now writing $\CRF{}$ (`cyclic recursive functionals') for the interpretations of $\C$ terms in $\nmod$, i.e.\ $\{\interp t \ | \ \text{$t$ a $\C$-coterm} \}$, the main result of this subsection is:

\begin{theorem}
\label{thm:crf=prf}
$\CRF{} = \PRF{}$. I.e.\ for every $\C$ coterm there is a $\T$ term computing the same functional (in $\nmod$) and vice versa.
\end{theorem}
\noindent
Notice that the right-left inclusion follows readily from the encoding of primitive recursion from Example~\ref{ex:sim-prim-rec}, so we concentrate on the left-right inclusion.
The statement above could indeed be refined in terms of type level, as for the main results in this work, but we shall drop such a specialisation for the sake of brevity.

For the remainder of this section we shall work inside the standard model $\nmod$. We employ (higher-order) recursion theoretic methods, and in particular take advantage of some well-known meta-recursion-theoretic results.

\subsubsection{Recursion schemes on well-founded relations}
Let $\lhd$ be a well-founded strict partial order on $\Nat$. 
We say that
$f:\nat \to \tau$ is obtained from $g: \nat \to (\nat \to \tau) \to \tau$ by \textbf{recursion on $\lhd$} (henceforth $\REC(\lhd)$) if:
\begin{equation}
\label{eqn:rec-on-wfpo}
f( n) \ = \ g(n, \lambda z \lhd n . f(z))
\end{equation}
Here $\lambda$, as expected, refers to abstraction of arguments; the `guarding' of the abstraction by $\lhd$ is formally defined as follows,
\[
(\lambda z \lhd n . f(z)) (m) \ := \ \begin{cases}
f(m) & \text{if $m \lhd n$} \\
0_\tau & \text{otherwise}
\end{cases}
\]
where $0_\tau$ is the `zero' functional of type $\tau$, setting $0_\nat := 0$ and $0_{\sigma \to \tau}: x \mapsto 0_\tau$.

We say that $\{f_i : \nat \to \tau_i \}_{i=1}^k$ are obtained by \textbf{simultaneous recursion on $\lhd$} (henceforth $\SREC(\lhd)$) from $\{ h_i : \nat \to (\nat \to \tau_1) \to \cdots \to (\nat \to \tau_k) \}_{i=1}^k$ if:
\begin{equation}
\label{eqn:simrec-wfpo}
f_i (n) \ = \ h_i (n, \lambda z \lhd n . f_1(z) , \dots, \lambda z \lhd n .  f_k(z)  )
\end{equation}
The following is well-known:
\begin{proposition}
[Closure under simultaneous recursion]
\label{prop:rec-implies-srec}
If $\PRF{}$ is closed under $\REC(\lhd)$ then
 $\PRF{}$ is also closed under $\SREC(\lhd)$.
\end{proposition}
\begin{proof}
[Proof sketch]
For simplicity, we shall make use of product types, thanks to usual primitive recursive (de)pairing operations (see, e.g., \cite{Troelstra73:metamathematical-investigations} for more details).
For $x_i$ of type $\sigma_i$, we will write $\tuple{x_1, \dots, x_k}$ for an element of type $\sigma_1 \times \cdots \times \sigma_k$, and conversely for a list $z$ of type $\sigma_1 \times \cdots \times \sigma_k$ we shall write $\proj i z$ for the $i$\textsuperscript{th} element of the list.

Let $f_1, \dots, f_k$ be as above, satisfying \eqref{eqn:simrec-wfpo}, and define $f: \nat \to (\tau_1 \times \cdots \times \tau_k)$ by:
\[
f(n) := \tuple{f_1(n), \dots, f_k(n)}
\]
Note that we have,
\[
\begin{array}{rcl}
f(n) & = & \tuple{f_i(n)}_{i=1}^k \\
	& = & \tuple{h_i (n, \lambda z \lhd n . f_1 (z), \dots, \lambda z \lhd n . f_k (z) ) }_{i=1}^k \\
	& = & \tuple{h_i (n,   \projraised 1 {\lambda z \lhd n .f(z)}, \dots, \projraised k {\lambda z \lhd n .f(z)} ) ) }_{i=1}^k
 \end{array}
\]
where $\pi_j' : (\nat \to (\tau_1 \times \cdots \times \tau_k)) \to \nat \to \tau_j$ by $\projraised j {g, n} := \proj j {g(n)}$.
This is an instance of $\REC(\lhd)$, so $f \in \PRF{}$. From here we indeed have, for $j=1, \dots, k$, that $f_j = \lambda n . (\proj j {f(n)}) \in \PRF{}$.
\end{proof}

\subsubsection{Closure under recursion on provably well-founded orders}
In this subsection we shall assume a standard primitive recursive representation of the ordinals up to $\epsilon_0$ as natural numbers, written $\alpha,\beta$ etc., and the usual (strict) well-order on them, written $\prec$.
Note that, while we may indeed represent $\prec$ primitive recursively in, say, $\RCA$, we certainly cannot prove that it is a well-order on $\epsilon_0$ even in $\ACA$, being its proof-theoretic ordinal.

Recalling the recursion schemes of the previous subsubsection,
we shall write simply $\REC(\alpha)$ for recursion on $\prec$ restricted to the initial segment $\alpha$ of $\epsilon_0$, and $\REC(\prec \epsilon_0)$ for the union of $\REC(\alpha)$ for $\alpha \prec \epsilon_0$.

The following is a well-known result, originally due to Kreisel \cite{Kreisel59:epsilon0-interpretation} by means of G\"odel's Dialectica functional interpretation (see also \cite{tait65:inf-long-terms,howard80:ord-anal-t}):
\begin{theorem}
[Kreisel]
\label{thm:prf-closed-under-recursion-<-epsilon0}
If $\alpha \prec \epsilon_0$ then $\PRF{}$ is closed under $\REC(\alpha)$.
\end{theorem}

We want to eventually reduce well-founded arguments on coderivations to some sort of recursion on ordinals $\prec \epsilon_0$. 
For this, we shall exploit the notion of `provably well-founded relations', going back to Gentzen:

\begin{definition}
[(Provably) inductive relations]
A relation $\lhd$ is \textbf{inductive} if the following holds: if whenever $ \forall y \lhd x X(y) $ we have $ X(x)$, then $\forall x X(x)$.

 $\lhd$ is \textbf{provably inductive in $\PA$} if,
\[
\PA(X) \proves \forall x ( \forall y \lhd x X(y) \cimp  X(x)) \cimp \forall x X(x)
\]
where $X$ is some fresh unary predicate symbol added to $\PA$.
\end{definition}

Gentzen already showed that any recursive well-order $\lhd$ on $\Nat$ provably inductive in $\PA$ has order type some $\alpha\prec \epsilon_0$ \cite{gentzen43:provable-well-orders}. 
This result was arithmetised by Takeuti (and independently Harrington), who further showed the existence of an order-preserving embedding that is $\prec\epsilon_0$-recursive (see \cite{takeuti-proof-theory,FriShe95:elementary-descent-rec}), i.e.\ in $\PRF{}$ (by Kreisel's result above, Theorem~\ref{thm:prf-closed-under-recursion-<-epsilon0}).
For our results we will need a generalisation of this result to well-founded \emph{partial} orders, which are not necessarily total.
Naturally every well-founded partial order can be order-preserving embedded into the ordinals, by well-founded induction, but, again, we need an arithmetised version of such a result to extract a suitable embedding in $\PRF{}$.
We expect that the particular well-founded partial orders we consider admit suitable arithmetisable linearisations, so as to directly apply Takeuti's and Harrington's result, but thankfully a more than suitable generalisation for well-founded partial orders has already been obtained by Arai:

\begin{theorem}
[\cite{arai98:from-the-attic}]
\label{thm:prov-wf-implies-prf}
Let $\lhd$ be a primitive recursive well-founded strict partial order on $\Nat$ that is provably inductive in $\PA$.
Then there is some $\alpha_\lhd \prec \epsilon_0$ and some $\mu_\lhd: \nat \to \nat $ in $ \PRF{}$ with $x \lhd y \implies \mu_\lhd(x)\prec  \mu_\lhd (y) \prec \alpha_\lhd$.
\end{theorem}

In fact Arai's result is much stronger: if $\lhd$ is elementary recursive, then so is $f_\lhd$, and the result above can actually be demonstrated within elementary recursive arithmetic, in particular $I\Delta_0 + \exp $. As it happens, the relations we shall consider (those induced by $\C$ coderivations) will indeed all be elementary recursive, but we shall not need such a strengthening of Theorem~\ref{thm:prov-wf-implies-prf}.

\begin{corollary}
\label{cor:prov-wf-implies-simrec}
Let $\lhd$ be as in Theorem~\ref{thm:prov-wf-implies-prf} above.
Then $\PRF{}$ is closed under $\REC(\lhd)$ and $\SREC(\lhd)$.
\end{corollary}
\begin{proof}
By Proposition~\ref{prop:rec-implies-srec} it suffices to show closure under $\REC(\lhd)$.
Suppose $g: \nat \to (\nat \to \tau) \to \tau$ is in  $\PRF{}$ and we will show that $f: \nat \to \tau$ is in $\PRF{}$ where:
\[
f(n) \ = \ g(n, \lambda z \lhd n . f(z))
\]
Letting $\mu_\lhd$ and $\alpha_\lhd \prec \epsilon_0$ be as obtained by Theorem~\ref{thm:prov-wf-implies-prf} above, let us write: 
\[
f'(\alpha, n) = \begin{cases}
f(n) & \text{if $\mu (n) \prec \alpha$} \\
0_\tau & \text{otherwise}
\end{cases}
\]
Notice that, as long as $\mu(n) \prec \alpha \prec \alpha_\lhd$, we have:
\[
\begin{array}{rcl}
 f'(\alpha) & = & \lambda n. g(n,  f' (\mu(n)))  \\
	& = & \lambda n .  g(n, (\lambda \beta \prec \alpha . f'(\beta)) (\mu(n))  ) 
\end{array}
\]
This is an instance of $\REC(\alpha_\lhd)$ and so $  f' \in \PRF{}$ by Theorem~\ref{thm:prf-closed-under-recursion-<-epsilon0}.
From here we have $f(n) = f'( \mu_\lhd (n) + 1, n)$, and so indeed $f \in \PRF{}$, as required.
\end{proof}

\subsubsection{A well-founded order on `runs' of progressing coderivations}
For the remainder of this subsection, let us fix a $\C$ coderivation $t $ whose (finitely many) distinct sub-coderivations are $\{t_i:  \vec \nat_i , \vec \sigma_i   \seqar \tau_i\}_{i=1}^n$, with all $\nat$ occurrences indicated.

Let $T\subseteq \{0,1\}^*$ be the underlying (infinite) coderivation tree of $t$, and let $\{T_1, \dots, T_n\}$ partition $T$ into the sets of nodes rooting $t_1, \dots, t_n$ respectively.
Notice that all these sets are provably recursive (given $t$) in even $\RCA$.

We define a binary relation $\runreduces$ on $T \times \Nat^*$ (i.e.\ $\runreduces \ \subseteq (T \times \Nat^*)\times (T \times \Nat^*)$) as follows.
$(u,\vec m) \runreduces (v,\vec n)$ if:
\begin{itemize}
\item $v$ is a child of $u$;
\item $u$ roots some $t_i$ (i.e.\ $u \in T_i$) and $v$ roots some $t_j$ (i.e.\ $v \in T_j$).
\item $|\vec m| = |\vec \nat_i|$ and $|\vec n| = |\vec \nat_j|$.
\item if $\nat_{jl}$ is an immediate ancestor of $\nat_{ik}$ then $n_l \leq m_k$.
\item if $\nat_{ik}$ is principal for a $\cond$ step and $\nat_{jl}$ is its immediate ancestor, then $ n_l < m_k$.
\end{itemize}
We do not impose any other constraints on $\runreduces$. 

\begin{example}
[Examples of `runs']
Revisiting Example~\ref{ex:sim-prim-rec}, suppose $t$ is the coderivation on the RHS of \eqref{eqn:sim-rec-naively-in-ct}. We have:
\begin{itemize}
\item $(\epsilon, 6 ) \runreduces (1,5) \runreduces (10,5)$.
\item $(\epsilon, 6 ) \runreduces (1,3) \runreduces (10,2)$.
\item $(\epsilon, 6 ) \runreduces (1,2) \runreduces (11,8)$.
\end{itemize}

Revisiting Section~\ref{sect:ack-cyclic}, now suppose $t$ is the coderivation from \eqref{eqn:ack-cyc-der}. We have:
\begin{itemize}
\item $(\epsilon, 5,3) \runreduces (0, 5, 5, 3) \runreduces (00, 5,3) \runreduces (000,3)$
\item $(\epsilon,5,3) \runreduces (0, 5,4,3) \runreduces (01, 4,4,3) \runreduces (010,4,4) \runreduces (0100,4) \runreduces (01001, 2)$
\item $(\epsilon, 5,3) \runreduces (0,5,5,3) \runreduces (01, 3,5,2) \runreduces (011, 3,4,0 )\runreduces (0110,4,0)$
\item $(\epsilon, 5,3) \runreduces (0, 2, 4,3 ) \runreduces (01, 1, 4, 3) \runreduces (011, 1, 4, 1) \runreduces (0111 , 1, 9)$
\end{itemize}
\end{example}

Note that $\runreduces$ is clearly a polynomial-time recursive relation, and in particular is provably $\Delta^0_1$ in even $\RCA$.
Since we have fixed $t$ in advance, we may actually establish the well-foundedness of $\runreduces$ within $\RCA$:

\begin{theorem}
[$\RCA$]
\label{thm:runs-terminate}
$\runreduces$ is terminating, i.e.\ $\forall f . \exists n .  f(n) \not\runreduces f(n+1)$.
\end{theorem}
\begin{proof}
[Proof sketch]
Suppose, for contradiction, that $f: \Nat \to \Nat $ with $\forall n . f(n) \runreduces f(n+1)$.
Writing $f(n) = (u_n, \vec m_n)$, we have that $(u_n)_{n\in \omega}$ is (the tail of) a branch of $t$.
By Proposition~\ref{prop:prog-decidability+provability-in-RCA} (itself from \cite{Das19:log-comp-cyc-arith}), $\RCA$ proves that $t$ is progressing, and so we have an infinitely progressing thread along $(u_n)_{n \in \omega}$.
From this thread we can extract from $(\vec m_n)_{n\in \omega}$ a sequence of natural numbers $(m_{ni_n})_{n\in \omega}$ corresponding to the thread. However, by construction, $(m_{ni_n})_{n \in \omega}$ is monotone decreasing (by induction on $n$) and has no least element (since it follows a progressing thread). 
\end{proof}
\noindent
Let us note that the uniform version of the above result, quantifying over all $\C$ coderivations $t$, requires $\CIND{\Sigma^0_2}$, cf.~\cite{KMPS19:buchi-reverse,Das19:log-comp-cyc-arith}.

\begin{corollary}
[$\ACA$]
\label{cor:runreduces-inductive}
$\rununreduces$ is inductive, i.e.\ $  (\forall x (\forall y \rununreduces x X(y) \cimp X(x)) \cimp \forall x X(x))$.
\end{corollary}
\begin{proof}
Suppose, for contradiction, that $\cnot X(n)$ and $\forall x (\forall y \rununreduces x X(y) \cimp X(x) )$, i.e.\ $\forall x (\cnot X(x) \cimp \exists y \rununreduces x \cnot X(y) )$.
Since $\runreduces$ is (provably) recursive, we may define $n_i$ with $n_0 = n$ and $n_{i+1}$ least such that $n_i \runreduces n_{i+1}$ and $\cnot X(n_{i+1})$, recursively in $i$.
By comprehension we may have the graph of the function $f(i) := n_i$ for $i\in \omega$, contradicting termination, Theorem~\ref{thm:runs-terminate} above.
\end{proof}

Let us henceforth write $\runge$ for the transitive closure of $\runreduces$.
\begin{corollary}
$\PRF{}$ is closed under $\REC (\runle)$, and so also $\SREC(\runle)$.
\end{corollary}
\begin{proof}
[Proof sketch]
By Corollary~\ref{cor:runreduces-inductive} above, we immediately have that $\runle$ is inductive, a fortiori, provably in $\ACA$, and so also in $\PA(X)$ by conservativity.
The result now follows by Corollary~\ref{cor:prov-wf-implies-simrec}.
\end{proof}

\subsubsection{Main result}

We continue to work with the fixed regular coderivation $t$ from the previous subsubsection and its distinct sub-coderivations $\{t_i : \vec \nat_i, \vec \sigma_i \seqar \tau_i \}_{i=1}^n$. 

\begin{proof}
[Proof of Theorem~\ref{thm:crf=prf}]
We shall assume some basic primitive recursive coding and decoding $\tuple{\cdot}$ of lists, and suppress the explicit functions associated with it, namely those adding and extracting elements from lists.

We show that the functionals $\{g_i : \nat \to \vec \sigma_i \to \tau_i\}_{i=1}^n$ with $g_i(\tuple{u, \vec m} ) := \interp{t_i}(\vec m_i)$, as long as $u \in T_i$, may be obtained by $\SREC(\runle)$.
\begin{itemize}
\item If $t_i $ has form $ \vlinf{\id_\nat}{}{\nat \seqar  \nat}{}$ then 
$g_i(\tuple{u,m}) = m$.
\item If $t_i $ has form $ \vlinf{\id_\sigma}{}{\sigma \seqar \sigma}{}$ with $\sigma \neq \nat$ then 
$g_i (\tuple{u})  = \interp{\id_\sigma}$.
\end{itemize}
Henceforth, we shall always assume that $\sigma \neq \nat$.
\begin{itemize}
\item 
If $t_i  $ has form $ 
\vlderivation{
\vlin{\wk_\nat}{}{ \vec \nat , \nat , \vec \sigma \seqar \tau}{
\vltr{t_j}{ \vec \nat , \vec \sigma \seqar \tau}{\vlhy{\ }}{\vlhy{\  }}{\vlhy{\ }}
}
}
$
then:
\[
\begin{array}{rcl}
 g_i (\tuple{u, \vec m,m}) & = & g_j (\tuple{u0,\vec m})\\
	& = & ({\lambda z \runle \tuple{u,\vec m, m} . g_j (z)}) (\tuple{u0,\vec m})
\end{array}
\]
\end{itemize}
Henceforth, we shall simply write recursive calls more compactly as $g_j (\tuple{u0, \vec m_0})$ or $g_k (\tuple{u1, \vec m_1})$ when $\tuple{u0, \vec m_0} \runle \tuple{u, \vec m}$ and $\tuple{u1, \vec m_1} \runle \tuple{u, \vec m}$, rather than fully writing $(\lambda z \runle \tuple{u, \vec m}. g_j(z)) (\tuple{u0, \vec m_0})$ or $(\lambda z \runle \tuple{u, \vec m}. g_k (z)) (\tuple{u1, \vec m_1})$.
\begin{itemize}
\item If $t_i $ has form $ 
\vlderivation{
\vlin{\wk}{}{ \vec \nat , \vec \sigma, \sigma \seqar \tau}{
\vltr{t_j}{ \vec \nat , \vec \sigma \seqar \tau}{\vlhy{\ }}{\vlhy{\ }}{\vlhy{\ }}
}
}
$
then
 \(
 g_i (\tuple{u, \vec m}) 
   = \interp{\wk} ( g_j (\tuple{u0, \vec m}) )
 \), 
 where the $\wk$ constant has the appropriate type, i.e.\ $(\vec \sigma \to \tau) \to \vec \sigma \to \sigma \to \tau$.
\end{itemize}
Henceforth, we shall omit the types of the constants we use, being determined by the context in which it appears.
\begin{itemize} 
 \item If $t_i $ has form $ 
 \vlderivation{
 	\vlin{\cntr_\nat}{}{\vec \nat, \nat, \vec \sigma  \seqar \tau}{
 	\vltr{t_j}{\vec \nat , \nat ,\nat, \vec \sigma  \seqar \tau}{\vlhy{\ }}{\vlhy{\ }}{\vlhy{\ }}
 	}
 }
 $
 then
 \(
 g_i (\tuple{u, \vec m, m}) =   g_j (\tuple{u0, \vec m, m, m })
 \).
 \item If $t_i $ has form $
  \vlderivation{
  	\vlin{\cntr}{}{\vec \nat, \vec \sigma ,  \sigma \seqar \tau}{
  	\vltr{t_j}{\vec \nat , \vec \sigma, \sigma, \sigma \seqar \tau}{\vlhy{\ }}{\vlhy{\ }}{\vlhy{\ }}
  	}
  }
 $
 then
 \(
 g_i (\tuple{u,\vec m})  =  \interp{\cntr} ( g_j (\tuple{u0, \vec m}) ) 
 \).
 \item If $t_i $ has form $
 \vlderivation{
 	\vliin{\cut_\nat}{}{\vec \nat , \vec \sigma \seqar \tau}{
 		\vltr{t_j}{\vec \nat , \vec \sigma \seqar \nat}{\vlhy{\ }}{\vlhy{\ }}{\vlhy{\ }}
 	}{
 		\vltr{t_k}{\vec \nat , \nat , \vec \sigma\seqar \tau}{\vlhy{\ }}{\vlhy{\ }}{\vlhy{\ }}
 	}
 }
 $
 then
 \[
 g_i (\tuple{u, \vec m})  =  \lambda \vec x. g_k ( \tuple{u1, \vec m, g_j (\tuple{u0, \vec m} , \vec x)}, \vec x)
 \]
 \item If $t_i $ has form $
 \vlderivation{
 	\vliin{\cut}{}{\vec \nat , \vec \sigma \seqar \tau}{
 		\vltr{t_j}{\vec \nat , \vec \sigma \seqar \sigma}{\vlhy{\ }}{\vlhy{\ }}{\vlhy{\ }}
 	}{
 		\vltr{t_k}{\vec \nat , \vec \sigma, \sigma \seqar \tau}{\vlhy{\ }}{\vlhy{\ }}{\vlhy{\ }}
 	}
 }
 $
 then
 \[
 g_i(\tuple{u, \vec m}) = \interp{\cut} ( g_j (\tuple{u0, \vec m}) , g_k (\tuple{u1, \vec m}) )
 \]
 \item If $t_i $ has form $
 \vlderivation{
 	\vliin{\leftimp_\nat}{}{\vec \nat , \vec \sigma, \rho \to \nat\seqar \tau}{
 		\vltr{t_j}{\vec \nat , \vec \sigma\seqar \rho}{\vlhy{\ }}{\vlhy{\ }}{\vlhy{\ }}
 	}{
 		\vltr{t_k}{\vec \nat, \nat , \vec \sigma\seqar \tau }{\vlhy{\ }}{\vlhy{\ }}{\vlhy{\ }}
 	}
 }
 $
 then
 \[
 g_i(\tuple{u, \vec m})  =  \lambda \vec x. \lambda h. ( g_k ( \tuple{u1, \vec m, h(g_j (\tuple{u0, \vec m}, \vec x) ) }, \vec x ) )
 \]
 \item If $t_i$ has form
 $
  \vlderivation{
  	\vliin{\leftimp}{}{\vec \nat , \vec \sigma, \rho \to \sigma \seqar \tau}{
  		\vltr{t_j}{\vec \nat , \vec \sigma \seqar \rho}{\vlhy{\ }}{\vlhy{\ }}{\vlhy{\ }}
  	}{
  		\vltr{t_k}{\vec \nat, \vec \sigma, \sigma \seqar \tau }{\vlhy{\ }}{\vlhy{\ }}{\vlhy{\ }}
  	}
  }
 $
 then
 \[
 g_i (\tuple{u , \vec m}) = \interp{\leftimp} (g_j (\tuple{u0, \vec m}) , g_k (\tuple{u1, \vec m}))
 \]
\item If $t_i $ has form 
$
\vlderivation{
	\vlin{\rightimp_\nat}{}{ \vec \nat , \vec \sigma \seqar \nat \to \tau}{
	\vltr{t_j}{ \vec \nat, \nat , \vec \sigma \seqar \tau}{\vlhy{\ }}{\vlhy{\ }}{\vlhy{\ }}
	}
}$
then
\(
g_i (\tuple{u, \vec m}) =  \lambda \vec x. \lambda m . g_j (\tuple{u0, \vec m , m }, \vec x)
\)
\item If $t_i $ has form
$
\vlderivation{
	\vlin{\rightimp}{}{\vec \nat , \vec \sigma \seqar \sigma \to \tau}{
	\vltr{t_j}{ \vec \nat , \vec \sigma, \sigma\seqar \tau}{\vlhy{\ }}{\vlhy{\ }}{\vlhy{\ }}
	}
}
$
then
\(
g_i (\tuple{u, \vec m}) = \interp{\rightimp} (g_j (\tuple{u0, \vec m}))
\)
\medskip
\item If $t_i$ has form $\vlinf{0}{}{\seqar \nat}{}$ then $g_i ( \tuple{u} ) = 0$.
\item If $t_i$ has form $\vlinf{\succ}{}{\nat \seqar \nat}{}$ then $g_i (\tuple{u,m}) = \succ m $.
\item If $t_i$ has form 
$
\vlderivation{
	\vliin{\cond}{}{\vec \nat, \nat , \vec \sigma\seqar \tau}{
		\vltr{t_j}{\vec \nat , \vec \sigma\seqar \tau}{\vlhy{\ }}{\vlhy{\ }}{\vlhy{\ }}
	}{
		\vltr{t_k}{\vec \nat, \nat , \vec \sigma\seqar \tau}{\vlhy{\ }}{\vlhy{\ }}{\vlhy{\ }}
	}
}
$
then
\[
\begin{array}{rcl}
g_i ( \tuple{u, \vec m, m}) & = & \begin{cases}
g_j (\tuple {u0, \vec m}) & \text{if $m=0$} \\
g_k ( \tuple{u1, \vec m, m-1}) &\text{otherwise}
\end{cases} 
\end{array}
\]
where the conditional ``if then else'' operation is obtained, as usual, by primitive recursion.
\end{itemize}
From here we have $\interp{t_i} (\vec m, \vec x)  = g_i (\tuple{u, \vec m}, \vec x)$, for some $u \in T_i$ (say the least one), so $\interp{t_i} \in \PRF{}$ for $i= 1, \dots, n$.
Since our initial $\C$ coderivation $t$ is just some $t_i$, we thus indeed have $\interp t \in \PRF{}$, so in particular $\interp t$ is computed by a term of $\T$.
\end{proof}

\section{Conclusions}

In this work we presented a circular version $\C$ of G\"odel's system $\T$ and investigated its expressivity at the level of abstraction complexity (i.e.\ type level).
To this end, we showed that $\nC n$ and $\nT{n+1}$ have the same logical and recursion-theoretic strength, by means of interpretations in each direction, over at least the type 1 quantifier-free theory.

We also gave several further results about the coterms of $\C$, for instance models of hereditarily computable functions, continuity at type 2, strong normalisation and confluence, and a translation to terms $\T$ computing the same funtional, at all types.

As mentioned in the Introduction, our ultimate motivation is to bring one of the hallmarks of 20\textsuperscript{th} century proof theory to the setting of non-wellfounded proofs: a bona fide correspondence between theories in predicate logic and type systems for functional programming languages.
The obtention of a Dialectica-style proof interpretation, cf.~\cite{dialectica}, between circular versions of arithmetic \cite{Sim17:cyclic-arith,BerTat17:lics,Das19:log-comp-cyc-arith} and the theory $\C$ here presented is thus the natural next step.

\bibliography{ct-refs}
\bibliographystyle{alpha}

\appendix

\newpage

\section{Partial cut-elimination for $\T$}
\label{sect:free-cut-elim}

In this section we prove a (presumably folklore) result that derivations of $\T$ can be partially normalised so that all types occurring have level dominated by that of a recursor (cf.~Proposition~\ref{prop:free-cut-elim}).
It is well-known that the complexity of cut-elimination (at the level of provability of sequents) is `superexponential' (see, e.g., \cite{hajek-pudlak:metamathematics}),
so this argument should go through already inside $\RCA$ (and even weaker theories), though we stop short of giving a complexity-theoretic analysis to this effect, so as not to complicate the exposition.

\subsection{Adapting the system for cut-elimination}

For simplicity, we no longer consider the exchange rule, but rather close all the instances of the other rules by composition with exchange, making formal how we were already informally typesetting rules in the main body of this paper. The axiomatisation of such combinators and their interpretations in the standard model $\nmod$ are as expected, and we do not formalise this in detail.

We shall also assume that all identity initial sequents are in atomic form.
Formally, we define $\id_\nat'$ as $\id_\nat$, and inductively define $\id_{\sigma\to\tau}'$ as below, left.
Clearly $\id_\nat' = \id_\nat$, and we give an argument by induction on type that $\id_{\sigma \to \tau}'= \id_{\sigma \to \tau}$ below, right (under $\extensionality$):
\[
\vlderivation{
	\vlin{\rightimp}{}{\sigma \to \tau \seqar \sigma \to \tau}{
	\vliin{\leftimp}{}{\sigma \to \tau, \sigma \seqar \tau}{
		\vltr{\id_\sigma' }{\sigma \seqar \sigma}{\vlhy{\ }}{\vlhy{\ }}{\vlhy{\ }}
	}{
		\vltr{\id_\tau'}{\tau \seqar \tau}{\vlhy{\ }}{\vlhy{\ }}{\vlhy{\ }}
	}
	}
}
\qquad
\begin{array}{rcll}
\id_{\sigma \to \tau}'x\, y & = & \rightimp \, (\leftimp\, \id_\sigma'\id_\tau')\, x\, y & \text{by definition of $\id'$}\\
& = & \leftimp \, \id_\sigma'\id_\tau' x\, y & \text{by $\rightimp$ axiom}\\
& = & \id_\tau' (x\, (\id_\sigma' y )) \\
& = & x\, y & \text{by inductive hypotheses}\\
& = & \id_{\sigma \to \tau}\, x\, y & \text{by $\id$ axiom}
\end{array}
\]

Furthermore, to simplify the termination argument for cut-elimination, we shall admit contraction by absorbing it into the other rules.
For this, we shall need the following variant $\leftimp' $ of the left-arrow rule, and its corresponding axioms:
\begin{equation}
\label{eqn:rule-leftimp-with-contraction}
\vliinf{\leftimp' }{}{\vec \sigma, \rho \to \sigma \seqar \tau}{\vec \sigma, \rho \to \sigma \seqar \rho}{\vec \sigma, \rho \to \sigma, \sigma \seqar \tau}
\qquad
\leftimp' s\, t\, \vec x\, y \ = \ t\, \vec x\, y\, (y\, (s\, \vec x\, y ))
\end{equation}

Note that this rule is easily derivable in $\T$, along with its corresponding equation, by combining the $\leftimp$ and $\cntr$ rules:
\[
\vlderivation{
	\vlin{\cntr}{}{\vec \sigma, \underline{\rho \to \sigma} \seqar \tau}{
	\vliin{\leftimp}{}{\vec \sigma, \rho \to \sigma, \underline{\rho \to \sigma} \seqar \tau}{
		\vlhy{\vec \sigma, \rho \to \sigma \seqar \rho}
	}{
		\vlhy{\vec \sigma, \rho \to \sigma, \sigma \seqar \tau}
	}
	}
}
\qquad
\begin{array}{rl}
&\cntr\, (\leftimp\, s\, t)\, \vec x\, y  \\
 = & \leftimp\, s\, t\, \vec x\, y\, y \\
	 = & t\, \vec x\, y (y\, (s\, \vec x\, y))
\end{array}
\]

We may freely use $\leftimp' $ in derivations, understanding it to be shorthand for the derivation above.
A \emph{contraction-free} derivation is one whose only $\cntr$-steps are already part of $\leftimp' $ steps.

\begin{proposition}
\label{prop:contr-free}
For every $\T$ derivation $t: \vec \sigma \seqar \tau$ there is a contraction-free derivation $t' : \vec \sigma \seqar \tau$, possibly with more $\cut_\nat$-occurrences, s.t.\ $\T\proves t = t'$.
\end{proposition}
\begin{proof}
[Proof sketch]
Follows by a straightforward induction on the structure of $t$, commuting $\cntr$ steps over those above.
The most interesting commutation is when a $\cntr$ step is immediately below a $\rec$ step, for which we must introduce some $\nat$-cuts:
\[
\vlderivation{
\vlin{\cntr}{}{\vec \sigma, \nat \seqar \tau}{
\vliin{\rec}{}{\vec \sigma , \red\nat, \underline{\blue \nat} \seqar \tau}{
	\vltr{s}{\vec \sigma , \red \nat \seqar \tau}{\vlhy{\ }}{\vlhy{\ }}{\vlhy{ \ }}
}{
	\vltr{t}{\vec \sigma , \red \nat , \blue \nat , \tau \seqar \tau}{\vlhy{\ }}{\vlhy{\ }}{\vlhy{\ }}
}
}
}
\quad\mapsto\quad
\vlderivation{
\vliin{\rec}{}{\vec \sigma, \nat \seqar \tau}{
	\vliin{\cut}{}{\vec \sigma \seqar \tau}{
		\vlin{0}{}{\seqar \nat}{\vlhy{}}
	}{
		\vltr{s}{\vec \sigma , \red \nat \seqar \tau}{\vlhy{\ }}{\vlhy{\ }}{\vlhy{ \ }}
	}
}{
	\vliin{\cut}{}{\vec \sigma, \blue{\nat}, \tau \seqar \tau}{
		\vlin{\succ}{}{\blue{\nat} \seqar \nat}{\vlhy{}}
	}{
		\vltr{t}{\vec \sigma , \red \nat , \blue \nat , \tau \seqar \tau}{\vlhy{\ }}{\vlhy{\ }}{\vlhy{\ }}
	}
}
}
\]
We can verify that the derivation after transformation is provably equivalent in $\T$ by an object-level induction on the recursion parameter:
\[
\begin{array}{rcll}
 \rec\, (\cut\, 0\, s)\, (\cut\, \succ\, t)\, \vec x\, 0 & = & \cut\, 0\, s\, \vec x & \text{by $\rec$ axioms}\\
& = & s\, \vec x\, 0 & \text{by $\cut$ axiom}\\
& = & \rec\, s\, t\, \vec x\, 0\, 0 & \text{by $\rec$ axioms}\\
&=& \cntr\, (\rec\, s\, t)\, \vec x\, 0 & \text{by $\cntr $ axiom}\\
\noalign{\medskip}
 \rec\, (\cut\, 0\, s)\, (\cut\, \succ\, t)\, \vec x\, \succ y & = & \cut\, \succ\, t\, \vec x\, y\, (\rec\, (\cut\, 0\, s)\, (\cut\, \succ\, t)\, \vec x\,  y) & \text{by $\rec$ axioms}\\
& = &  \cut\, \succ\, t\, \vec x\, y\, (\cntr\, (\rec\, s\, t)\, \vec x\, y) & \text{by inductive hypothesis}\\
& = & t\, \vec x\, \succ y\, y\, (\cntr\, (\rec\, s\, t)\, \vec x\, y) & \text{by $\cut$ axiom}\\
& = & \rec\, s\, t\, \vec x\, \succ y\, \succ y & \text{by $\rec$ axioms}\\
& = & \cntr\, (\rec\, s\, t)\, \vec x\, \succ y & \text{by $\cntr$ axiom} \qedhere
\end{array}
\]
\end{proof}

We shall also use a generalised version of the cut rule that incorporates both context sharing and context splitting behaviour:
\begin{equation}
\label{eqn:rule-cut-with-weakening}
\vliinf{\cut'}{}{\vec \rho_0, \vec \rho_1, \vec \sigma \seqar \sigma}{\vec \rho_0, \vec \sigma, \sigma \seqar \tau}{\vec \rho_1, \vec \sigma \seqar \tau}
\qquad
\cut' s\, t\,  \vec x_0\, \vec x_1\, \vec y \ = \ t\, \vec x_1\, \vec y\, (s\, \vec x_0\, \vec y)
\end{equation}
The point of this rule is to absorb extraneous weakening steps from cut-reductions into cuts.
Note that we have been implicitly using this until now, e.g.\ as in the derivation in the proof sketch of Proposition~\ref{prop:contr-free}.
Similarly to $\leftimp' $, the rule $\cut' $ and its corresponding equation can be derived from $\cut$ and $\wk$:
\[
\vlderivation{
	\vliin{\cut}{}{\vec \rho_0, \vec \rho_1, \vec \sigma \seqar \sigma }{
		\vliq{\wk}{}{\vec \rho_0, \underline{\vec \rho_1} , \vec \sigma \seqar \sigma}{
		\vlhy{\vec \rho_0, \vec \sigma \seqar \sigma}
		}
	}{
		\vliq{\wk}{}{\underline{\vec \rho_0}, \vec \rho_1, \vec \sigma, \sigma \seqar \tau}{
		\vlhy{\vec \rho_1, \vec \sigma \seqar \tau}
		}
	}
}
\qquad
\begin{array}{rl}
 & \cut \, (\wk^* s)\, (\wk^* t)\, \vec x_0\, \vec x_1\, \vec y\\
  = & \wk^* t\, \vec x_0\, \vec x_1\, \vec y ( \wk^* s\, \vec x_0\, \vec x_1\, \vec y ) \\
	= & t\,  \vec x_1\, \vec y\, ( \wk^* s\, \vec x_0\, \vec x_1\, \vec y ) \\
	 = & t\,  \vec x_1\, \vec y \, (  s\, \vec x_0\, \vec y )
\end{array}
\]

Again, we may freely use $\cut' $ in derivations, understanding it to be shorthand for the derivation above.
However, crucially, we will only count instances of $\leftimp' $ and $\cut' $ as single steps, in order to facilitate the upcoming induction arguments.\footnote{Another approach could be to ignore all $\wk$ and $\cntr$ steps (`weak inferences') when counting the size/depth of derivations, but then we would need a lemma for commuting above those steps in order to reduce an induction measure based on size/depth.}

\subsection{Main free-cut elimination argument}
As usual, our overall argument will be by an induction on the complexity of cut-formulas:
\begin{itemize}
	\item The \emph{level} of a $\cut'$ on $\sigma$ is just $\level(\sigma)$.
	\item The $d$-level of a derivation is the multiset of all its cut-levels $>d$.
\end{itemize}
We assume that multisets of natural numbers are ordered in the usual way.\footnote{It is well-known that for any well-order $<$ of order type $\alpha$, the corresponding multiset order has order type $\omega^\alpha$. Thus the multiset order on even $\omega$ is already not available in $\RCA$, whose proof-theoretic ordinal is $\omega^\omega$.
However, cut-elimination arguments based on the multiset ordering are nonetheless typically formalisable in $\RCA$ thanks to explicit complexity bounds on the multiset branching, i.e.\ the number of smaller number occurrences that replace a larger one when reducing.}
We may write $\cut_d'$, $\cut_{\leq d}'$ and $\cut_{>d}'$ for a $\cut'$ instance of level $d$, $\leq d$ or $>d$ respectively.

\begin{lemma}
\label{lem:lower-multiset-of-cuts}
	Let $t: \vec \sigma \seqar \tau$ be a contraction-free $\nT n$-derivation with at least one $\cut_{>n}'$. Then there is a contraction-free $\nT n$-derivation $t': \vec \sigma \seqar \tau$ of lower $n$-level such that $\nT n \proves t = t' $.
\end{lemma}
\begin{proof}
We proceed by induction on $|t|$, by inspection of a topmost $\cut_{>n}$.
All permutation cases are given in the subsequent subsections, as well as arguments verifying that the derivations before and after transformation are provably equivalent in $\nT n$.
In each case, we may have to apply the induction hypothesis zero, one or two times to smaller subderivations.
Note that the only transformation that introduces more cuts (even after applying inductive hypotheses) is the $\rightimp$-$\leftimp'$ `key' case.
\end{proof}

Proposition~\ref{prop:free-cut-elim} is now an immediate consequence of the following result:
\begin{theorem}
	For any $\nT n$-derivation $t:\vec \sigma \seqar \tau$ there is a $\cut_{>n}'$-free $\nT n$-derivation $t':\vec \sigma \seqar \tau$, such that $\T \proves t'\vec x \ =  \ t\, \vec x$.
\end{theorem}
\begin{proof}
Assume $t$ is contraction-free by Proposition~\ref{prop:contr-free}.
The result follows by induction on the $n$-level of $t$, applying Lemma~\ref{lem:lower-multiset-of-cuts} above for the inductive steps.
\end{proof}

\subsection{$\rightimp$-$\leftimp'$ key case}
\[
\vlderivation{
	\vliin{\cut'}{}{\vec \rho_0, \vec \rho_1 , \vec \sigma \seqar \tau}{
		\vlin{\rightimp}{}{\vec \rho_0, \vec \sigma \seqar \rho \to \sigma}{
		\vltr{s}{\vec \rho_0, \vec \sigma, \rho\seqar \sigma}{\vlhy{\ }}{\vlhy{\ }}{\vlhy{\ }}
		}
	}{
		\vliin{\leftimp' }{}{\vec \rho_1, \vec \sigma, \rho \to \sigma \seqar \tau}{
			\vltr{r}{\vec \rho_1, \vec \sigma, \rho \to \sigma \seqar \rho}{\vlhy{\ }}{\vlhy{\ }}{\vlhy{\ }}
		}{
			\vltr{t}{\vec \rho_1, \vec \sigma , \rho \to \sigma , \sigma \seqar \tau}{\vlhy{\ }}{\vlhy{\ }}{\vlhy{\ }}
		}
	}
}
\]

is transformed to:

\[
\vlderivation{
\vliin{\cut}{}{\vec \rho_0, \vec \rho_1, \vec \sigma \seqar \tau}{
		\vliin{\cut'}{}{\vec \rho_0, \vec \rho_1 , \vec \sigma \seqar \sigma}{
			\vliin{\cut'}{}{\vec \rho_0, \vec \rho_1, \vec \sigma \seqar \rho}{
				{
				\vltr{s'}{
				\vec \rho_0, \vec \sigma \seqar \rho \to \sigma
				}{\vlhy{\ }}{\vlhy{\ }}{\vlhy{\ }}
				}
			}{
				\vltr{r}{\vec \rho_1, \vec \sigma, \rho \to \sigma \seqar \rho}{\vlhy{\ }}{\vlhy{\ }}{\vlhy{\ }}
			}
		}{
			\vltr{s}{\vec \rho_0, \vec \sigma, \rho\seqar \sigma}{\vlhy{\ }}{\vlhy{\ }}{\vlhy{\ }}
		}
	}{
		\vliin{\cut'}{}{\vec \rho_0, \vec \rho_1, \vec \sigma, \sigma \seqar \tau}{
			{
			\vltr{s'}{
			\vec \rho_0, \vec \sigma \seqar \rho \to \sigma
			}{\vlhy{\ }}{\vlhy{\ }}{\vlhy{\ }}
		}
	}{
		\vltr{t}{\vec \rho_1, \vec \sigma , \rho \to \sigma , \sigma \seqar \tau}{\vlhy{\ }}{\vlhy{\ }}{\vlhy{\ }}
	}
}
}
\]
where we write $s'$ for the derivation 
\[
\vlderivation{
	\vlin{\rightimp}{}{\vec \rho_0, \vec \sigma \seqar \rho \to \sigma}{
			\vltr{s}{\vec \rho_0, \vec \sigma, \rho\seqar \sigma}{\vlhy{\ }}{\vlhy{\ }}{\vlhy{\ }}
			}
}
\]

We verify that $\nT n$ proves the equality of the derivations before and after transformation as follows:
\[
\begin{array}{rll}
& \cut\, (\cut' (\cut' s' r)\, s )\, (\cut' s' t )\, \vec x_0\, \vec x_1\, \vec y &  \\ 
= & \cut' s' t \, \vec x_0\, \vec x_1\, \vec y\, (\cut' (\cut' s' r)\, s \, \vec x_0\, \vec x_1\, \vec y) & \text{by $\cut$ axiom}\\
 = & \cut' s' t \, \vec x_0\, \vec x_1\, \vec y\, ( s\, \vec x_0\, \vec y\, ( \cut' s' r\, \vec x_0\, \vec x_1\, \vec y ) ) & \text{by $\cut'$ axiom} \\
 =&  \cut' s' t \, \vec x_0\, \vec x_1\, \vec y\, ( s\, \vec x_0\, \vec y\, ( r\,  \vec x_1\, \vec y\, (s' \vec x_0\, \vec y) ) ) & \text{by $\cut'$ axiom}\\
 =& t\, \vec x_1\, \vec y\, (s' \vec x_0\, \vec y)\, ( s\, \vec x_0\, \vec y\, ( r\,  \vec x_1\, \vec y\, (s' \vec x_0\, \vec y) ) ) & \text{by $\cut'$ axiom} \\
 = & t\, \vec x_1\, \vec y\, (s' \vec x_0\, \vec y)\, ( s'  \vec x_0\, \vec y\, ( r\,  \vec x_1\, \vec y\, (s' \vec x_0\, \vec y) ) ) & \text{by $\rightimp$ axiom} \\
 = & \leftimp \, r\, t\, \vec x_1\, \vec y\, (s' \vec x_0\, \vec y) & \text{by $\leftimp$ axiom}\\
 =& \cut' s' (\leftimp\, r\, t)\, \vec x_0\, \vec x_1\, \vec y & \text{by $\cut' $ axiom}
\end{array}
\]

\subsection{$\cut'$-$\rec$ commutative case}
The most interesting commutative case is commuting a $\cut_{>n}'$ above a $\rec_{\leq n}$ step:
\[
\vlderivation{
	\vliin{\cut'}{}{\vec \rho_0, \vec \rho_1, \vec \sigma , \nat \seqar \tau}{
		\vltr{r}{\vec x_0, \vec \sigma , \nat \seqar \sigma}{\vlhy{\ }}{\vlhy{ \ }}{\vlhy{\ }}
	}{
		\vliin{\rec_\tau}{}{\vec \rho_1, \vec \sigma, \sigma , \nat \seqar \tau}{
			\vltr{s}{\vec \rho_1, \vec \sigma, \sigma \seqar \tau}{\vlhy{\ }}{\vlhy{ \ }}{\vlhy{ \ }}
		}{
			\vltr{t}{\vec \rho_1, \vec \sigma, \sigma , \nat ,\tau \seqar \tau}{\vlhy{\ }}{\vlhy{\ }}{\vlhy{ \ }}
		}
	}
}
\]
is transformed to,
\[
\vlderivation{
\vlin{\cntr}{}{\vec \rho_0, \vec \rho_1, \vec \sigma , \underline{\nat} \seqar \tau}{
\vliin{\rec_\tau}{}{\vec \rho_0, \vec \rho_1, \vec \sigma , \red{\nat}, \underline{{\blue{\nat}}}  \seqar \tau}{
	\vliin{\cut'}{}{\vec \rho_0, \vec \rho_1, \vec \sigma, \red{\nat} \seqar \tau}{
		\vltr{r}{\vec x_0, \vec \sigma , \red{\nat} \seqar \sigma}{\vlhy{\ }}{\vlhy{ \ }}{\vlhy{\ }}
	}{
		\vltr{s}{\vec \rho_1, \vec \sigma, \sigma \seqar \tau}{\vlhy{\ }}{\vlhy{ \ }}{\vlhy{ \ }}
	}
}{
	\vliin{\cut'}{}{\vec \rho_0, \vec \rho_1, \vec \sigma, \red{\nat }, \blue{\nat} , \tau \seqar \tau}{
		\vltr{r}{\vec x_0, \vec \sigma , \red{\nat} \seqar \sigma}{\vlhy{\ }}{\vlhy{ \ }}{\vlhy{\ }}
	}{
		\vltr{t}{\vec \rho_1, \vec \sigma, \sigma , \blue{\nat} ,\tau \seqar \tau}{\vlhy{\ }}{\vlhy{\ }}{\vlhy{ \ }}
	}
}
}
}
\]
where we have underlined principal types and used colours to identify type occurrences according to ancestry.

We verify that $\nT n$ proves the equality of the derivations before and after transformation as follows.
First we show that,
\begin{equation}
\label{eqn:comm-cut-rec-ih}
\rec\, (\cut' r\, s)\, (\cut' r\, t)\, \vec x_0\, \vec x_1\, \vec y\, a\, z \ = \ \rec\, s\, t\, \vec x_1\, \vec y\, (r\, \vec x_0\, \vec y\, a)\, z
\end{equation}
by (object-level) induction on $z$.
For the base case:
\[
\begin{array}{rll}
& \rec\, (\cut' r\, s)\, (\cut' r\, t)\, \vec x_0\, \vec x_1\, \vec y\, a\, 0 & \\
= & \cut' r\, s\, \vec x_0\, \vec x_1\, \vec y\, a & \text{by $\rec$ axioms}\\
=& s\, \vec x_1\, \vec y\, (r\, \vec x_0\, \vec y\, a) & \text{by $\cut'$ axiom}\\
= & \rec\, s\, t\, \vec x_1 \, \vec y\, (r\, \vec x_0\, \vec y\, a)\, 0 & \text{by $\rec$ axioms}
\end{array}
\]
For the inductive step:
\[
\begin{array}{rll}
& \rec\, (\cut' r\, s)\, (\cut' r\, t)\, \vec x_0\, \vec x_1\, \vec y\, a\, \succ z & \\
= & \cut' r\, t\, \vec x_0\, \vec x_1\, \vec y\, a\, z\, (\rec\, (\cut' r\, s) \, (\cut' r\, t)\, \vec x_0\, \vec x_1\, \vec y\, a\, z ) & \text{by $\rec$ axioms}\\
=& \cut' r\, t\, \vec x_0\, \vec x_1\, \vec y\, a\, z\, (\rec \, s\, t\, \vec x_1\, \vec y \, (r\, \vec x_0\, \vec y\, a)\, z) & \text{by inductive hypothesis, \eqref{eqn:comm-cut-rec-ih}}\\
=& t\, \vec x_1\, \vec y \,(r\, \vec x_0\, \vec y\, a)\, z\, (\rec \, s\, t\, \vec x_1\, \vec y \, (r\, \vec x_0\, \vec y\, a)\, z) & \text{by $\cut'$ axiom}\\
=& \rec\, s\, t\, \vec x_1\, \vec y\, (r\, \vec x_0\, \vec y\, a)\, \succ z & \text{by $\rec$ axioms}
\end{array}
\]
From here we conclude the verification in $\nT n$ by the $\cntr$ and $\cut$ axioms.

\subsection{$\cut'$-$\cut$ commutative cases}
Sometimes we have to commute a $\cut_{>n}'$ over a $\cut_{\leq n}$.
It is notationally cumbersome to consider all splitting possibilities for commuting $\cut'$ over $\cut'$, so instead we just permute a $\cut_{>n}'$ step over a $\cut_{\leq n}$ step. 
From here, commutation over a $\cut_{\leq n}'$ step follows by further commutation over $\wk$-steps (see next subsection).

\[
\vlderivation{
	\vliin{\cut_{>n}'}{}{\vec \rho_0, \vec \rho_1, \vec \sigma \seqar \tau}{
		\vliin{\cut_{\leq n}}{}{\vec \rho_0, \vec \sigma \seqar \sigma}{
			\vltr{r}{\vec \rho_0, \vec \sigma \seqar \rho}{\vlhy{\ }}{\vlhy{\ }}{\vlhy{ \ }}
		}{
			\vltr{s}{\vec \rho_0, \vec \sigma, \rho \seqar \sigma}{\vlhy{\ }}{\vlhy{\ }}{\vlhy{\ }}
		}
	}{
		\vltr{t}{\vec \rho_1, \vec \sigma, \sigma \seqar \tau}{\vlhy{\ }}{\vlhy{\ }}{\vlhy{\ }}
	}
}
\]
is transformed to:
\[
\vlderivation{
	\vliin{\cut_{\leq n}'}{}{\vec \rho_0, \vec \rho_1, \vec \sigma \seqar \tau }{
		\vltr{r}{\vec \rho_0, \vec \sigma \seqar \rho}{\vlhy{\ }}{\vlhy{\ }}{\vlhy{ \ }}
	}{
		\vliin{\cut_{>n}'}{}{\vec \rho_0, \vec \rho_1, \vec \sigma, \rho \seqar \tau}{
			\vltr{s}{\vec \rho_0, \vec \sigma, \rho \seqar \sigma}{\vlhy{\ }}{\vlhy{\ }}{\vlhy{\ }}
		}{
			\vltr{t}{\vec \rho_1, \vec \sigma, \sigma \seqar \tau}{\vlhy{\ }}{\vlhy{\ }}{\vlhy{\ }}
		}
	}
}
\]

We verify that $\nT n$ proves the equality of the derivations before and after transformation as follows:
\[
\begin{array}{rll}
& \cut' r\, (\cut' s\, t)\, \vec x_0\, \vec x_1\, \vec y & \\
=& \cut' s\, t\, \vec x_0\, \vec x_1\, \vec y\, (r\, \vec x_0\, \vec y) & \text{by $\cut'$ axiom}\\
=& t\, \vec x_1\, \vec y\, (s\, \vec x_0\, \vec y\, (r\, \vec x_0\, \vec y)) & \text{by $\cut'$ axiom}\\
=& t\, \vec x_1\, \vec y\, (\cut\, r\, s\, \vec x_0, \vec y) & \text{by $\cut $ axiom}\\
=& \cut' (\cut\, r\, s)\, t\, \vec x_0\, \vec x_1\, \vec y & \text{by $\cut' $ axiom}
\end{array}
\]

There is also a similar cut-commutation the other way around, when the right premiss of a $\cut'$ step ends with $\cut$.
There is no commutation of $\cut'$ above the left side of a $\rec$ step, since that would immediately imply that the level of the cut is bounded by that of a recursor.

\subsection{$\cut'$-$\wk$ key and commutative case}

There are two possible interactions between $\cut'$ and $\wk$, depending on whether the cut-formula is weakened or not. Both are relatively simple.

\[
\vlderivation{
	\vliin{\cut'}{}{\vec \rho_0, \vec \rho_1, \vec \sigma \seqar \tau}{
		\vltr{s}{\vec \rho_0, \vec \sigma \seqar \sigma}{\vlhy{\ }}{\vlhy{\ }}{\vlhy{\ }}
	}{
		\vlin{\wk}{}{\vec \rho_1, \vec \sigma, \sigma \seqar \tau}{
		\vltr{t}{\vec \rho_1, \vec \sigma \seqar \tau}{\vlhy{\ }}{\vlhy{\ }}{\vlhy{\ }}
		}
	}
}
\]
is transformed to:
\[
\vlderivation{
	\vliq{\wk}{}{\vec \rho_0, \vec \rho_1, \vec \sigma \seqar \tau}{
	\vltr{t}{\vec \rho_1, \vec \sigma \seqar \tau}{\vlhy{\ }}{\vlhy{\ }}{\vlhy{\ }}
	}
}
\]
We verify that $\nT n$ proves the equality of the derivations before and after transformation as follows:
\[
\begin{array}{rll}
& \wk^*t\, \vec x_0\, \vec x_1\, \vec y & \\
= & t\, \vec x_1\, \vec y & \text{by $\wk$ axioms}\\
= & \wk\, t\, \vec x_1, \vec y\, (s\, \vec x_0\, \vec y) & \text{by $\wk $ axiom}\\
=& \cut' s\, (\wk\, t)\, \vec x_0\, \vec x_1, \vec y & \text{by $\cut'$ axiom}
\end{array}
\]

\[
\vlderivation{
\vliin{\cut'}{}{\vec \rho_0, \rho, \vec \rho_1, \vec \sigma \seqar \tau}{
	\vltr{s}{\vec \rho_0, \purple \rho, \vec \sigma \seqar \sigma }{\vlhy{\ }}{\vlhy{\ }}{\vlhy{\ }}
}{
	\vlin{\wk}{}{\rho, \vec \rho_1, \vec \sigma, \sigma \seqar \tau}{
	\vltr{t}{\vec \rho_1, \vec \sigma, \sigma \seqar \tau}{\vlhy{\ }}{\vlhy{\ }}{\vlhy{\ }}
	}
}
}
\]
is transformed to either the left or right derivations below, depending on whether the purple $\purple \rho$ is present or not, respectively:
\[
\vlderivation{
\vliin{\cut'}{}{\vec \rho_0, \rho, \vec \rho_1, \vec \sigma \seqar \tau}{
	\vltr{s}{\vec \rho_0, \purple \rho, \vec \sigma \seqar \sigma }{\vlhy{\ }}{\vlhy{\ }}{\vlhy{\ }}
}{
	\vltr{t}{\vec \rho_1, \vec \sigma, \sigma \seqar \tau}{\vlhy{\ }}{\vlhy{\ }}{\vlhy{\ }}
}
}
\qquad
\vlderivation{
	\vlin{\wk}{}{\vec \rho_0, \rho, \vec \rho_1, \vec \sigma \seqar \tau}{
	\vliin{\cut'}{}{\vec \rho_0, \vec \rho_1, \vec \sigma \seqar \tau}{
		\vltr{s}{\vec \rho_0, \vec \sigma \seqar \sigma }{\vlhy{\ }}{\vlhy{\ }}{\vlhy{\ }}
	}{
		\vltr{t}{\vec \rho_1, \vec \sigma, \sigma \seqar \tau}{\vlhy{\ }}{\vlhy{\ }}{\vlhy{\ }}
	}
	}
}
\]
We verify that $\nT n$ proves the equality of the derivations before and after transformation, respectively, as follows:
\[
\begin{array}{rll}
& \cut' s\, t\, \vec x_0\, x\, \vec x_1\, \vec y & \\
= & t\, \vec x_1\, \vec y\, (s\, \vec x_0\, x\, \vec y) & \text{by $\cut'$ axiom}\\
= & \wk\, t\, x\, \vec x_1\, \vec y\, (s\, \vec x_0\, x\, \vec y) & \text{by $\wk$ axiom}\\
= & \cut' s\, (\wk\, t)\, \vec x_0\, x\, \vec x_1\, \vec y & \text{by $\cut'$ axiom}
\end{array}
\qquad
\begin{array}{rll}
& \wk\, (\cut'\, s\, t)\, \vec x_0\, x\, \vec x_1\, \vec y & \\
= & \cut'\, s\, t\, \vec x_0\, \vec x_1\, \vec y & \text{by $\wk$ axiom}\\
= & t\, \vec x_1\, \vec y\, (s\, \vec x_0\, \vec y) & \text{by $\cut'$ axiom} \\
= & \wk\, t\, x\, \vec x_1\, \vec y\, (s\, \vec x_0\, \vec y) & \text{by $\wk $ axiom}\\
= & \cut' s\, (\wk\, t)\, \vec x_0\, x\, \vec x_1\, \vec y & \text{by $\cut'$ axiom}
\end{array}
\]

There are similar cases when the left premiss of a $\cut'$ concludes a $\wk$ step.

\subsection{$\cut'$-$\leftimp'$ commutative cases}
We have to treat the cases when a $\leftimp'$ step is on the left or on the right separately.

\[
\vlderivation{
	\vliin{\cut'}{}{\vec \rho_0, \vec \rho_1, \vec \sigma ,\rho \to \sigma \seqar \tau }{
		\vliin{\leftimp'}{}{\vec \rho_0, \vec \sigma , \rho \to \sigma \seqar \pi}{
			\vltr{r}{\vec \rho_0, \vec \sigma , \rho \to \sigma \seqar \rho}{\vlhy{\ }}{\vlhy{\ }}{\vlhy{\ }}
		}{
			\vltr{s}{\vec \rho_0, \sigma, \rho \to \sigma,\sigma\seqar \pi }{\vlhy{\ }}{\vlhy{\ }}{\vlhy{\ }}
		}
	}{
		\vltr{t}{\vec \rho_1, \vec \sigma, \purple{\rho \to \sigma}, \pi, \seqar \tau}{\vlhy{\ }}{\vlhy{\ }}{\vlhy{\ }}
	}
}
\]
is transformed to,
\[
\vlderivation{
	\vliin{\leftimp'}{}{\vec \rho_0, \vec \rho_1, \vec \sigma, \rho \to \sigma \seqar \tau}{
		\vliq{\wk}{}{\vec \rho_0, \vec \rho_1, \vec \sigma , \rho \to \sigma \seqar \rho}{
		\vltr{r}{\vec \rho_0, \vec \sigma , \rho \to \sigma \seqar \rho}{\vlhy{\ }}{\vlhy{\ }}{\vlhy{\ }}
		}
	}{
		\vliin{\cut'}{}{\vec \rho_0, \vec \rho_1 , \vec \sigma, \rho \to \sigma , \sigma \seqar \tau}{
			\vltr{s}{\vec \rho_0, \sigma, \rho \to \sigma,\sigma\seqar \pi }{\vlhy{\ }}{\vlhy{\ }}{\vlhy{\ }}
		}{
			\vltr{t}{\vec \rho_1, \vec \sigma, \purple{\rho \to \sigma}, \pi \seqar \tau}{\vlhy{\ }}{\vlhy{\ }}{\vlhy{\ }}
		}
	}
}
\]
where the purple occurrence of $\purple{\rho\to\sigma}$ after transformation is present only if it is present before transformation.

We verify that $\nT n$ proves the equality of the derivations before and after transformation as follows,
\[
\begin{array}{rll}
&\leftimp' (\wk^*r)\, (\cut' s\, t)\, \vec x_0\, \vec x_1\, \vec y\, z & \\
=& \cut' s\, t\, \vec x_0\, \vec x_1\, \vec y\, z\, (z\, (\wk^* r\, \vec x_0\, \vec x_1\, \vec y\, z)) & \text{by $\leftimp'$ axiom}\\
=& \cut' s\, t\, \vec x_0\, \vec x_1\, \vec y\, z\, (z\, ( r\, \vec x_0\, \vec y\, z)) & \text{by $\wk$ axioms}\\
=& t\, \vec x_1\, \vec y\, \purple{z}\, (s\, \vec x_0\, \vec y\, z\, (z\, ( r\, \vec x_0\, \vec y\, z)) ) & \text{by $\cut'$ axiom}\\
=& t\, \vec x_1\, \vec y\, \purple{z}\, (\leftimp' r\, s\, \vec x_0\, \vec y\, z ) & \text{by $\leftimp'  $ axiom}\\
=& \cut' (\leftimp' r\, s)\, t\, \vec x_0\, \vec x_1\, \vec y\, z
\end{array}
\]
where, again, the purple $\purple z$ is present just if the purple $\purple{\rho\to\sigma}$ are present in the derivations.

\[
\vlderivation{
	\vliin{\cut'}{}{\vec \rho_0, \vec \rho_1, \vec \sigma, \rho \to \sigma \seqar \tau}{
		\vltr{r}{\vec \rho_0, \vec \sigma, \purple{\rho \to \sigma} \seqar \pi}{\vlhy{\ }}{\vlhy{\ }}{\vlhy{\ }}
	}{
		\vliin{\leftimp'}{}{\vec \rho_1, \vec \sigma, \rho\to \sigma, \pi \seqar \tau}{
			\vltr{s}{\vec \rho_1, \vec \sigma, \rho \to \sigma, \pi \seqar \rho}{\vlhy{\ }}{\vlhy{\ }}{\vlhy{\ }}
		}{
			\vltr{t}{\vec \rho_1, \vec \sigma, \rho \to \sigma, \sigma, \pi \seqar \tau}{\vlhy{\ }}{\vlhy{\ }}{\vlhy{ \ }}
		}
	}
}
\]
is transformed to,
\[
\vlderivation{
	\vliin{\leftimp'}{}{\vec \rho_0, \vec \rho_1, \vec \sigma, \rho \to \sigma \seqar \tau}{
		\vliin{\cut'}{}{\vec \rho_0, \vec \rho_1, \vec \sigma, \rho \to \sigma \seqar \rho}{
			\vltr{r}{\vec \rho_0, \vec \sigma, \purple{\rho \to \sigma} \seqar \pi}{\vlhy{\ }}{\vlhy{\ }}{\vlhy{\ }}
		}{
			\vltr{s}{\vec \rho_1, \vec \sigma, \rho \to \sigma, \pi \seqar \rho}{\vlhy{\ }}{\vlhy{\ }}{\vlhy{\ }}
		}
	}{
		\vliin{\cut'}{}{\vec \rho_0, \vec \rho_1, \vec \sigma, \rho \to \sigma, \sigma \seqar \tau}{
			\vltr{r}{\vec \rho_0, \vec \sigma, \purple{\rho \to \sigma} \seqar \pi}{\vlhy{\ }}{\vlhy{\ }}{\vlhy{\ }}
		}{
			\vltr{t}{\vec \rho_1, \vec \sigma, \rho \to \sigma, \sigma, \pi \seqar \tau}{\vlhy{\ }}{\vlhy{\ }}{\vlhy{ \ }}
		}
	}
}
\]
where the purple occurrences of $\purple{\rho \to \sigma}$ after transformation are present just if they are before transformation.

We verify that $\nT n$ proves the equality of the derivations before and after transformation as follows,
\[
\begin{array}{rll}
& \leftimp ' (\cut'r\, s)\, (\cut'r\, t)\, \vec x_0\, \vec x_1\, \vec y\, z & \\
= & \cut'r\, t\, \vec x_0\, \vec x_1\, \vec y\, z\, (z\, (\cut'r\, s\, \vec x_0\, \vec x_1\, \vec y\, z ) )& \text{by $\leftimp'$ axiom }\\
=& \cut'r\, t\, \vec x_0\, \vec x_1\, \vec y\, z\, (z\, (s\, \vec x_1\, \vec y\, z\, (r\, \vec x_0\, \vec y\, \purple z) ) ) & \text{by $\cut'$ axiom}\\
=& t\, \vec x_1\, \vec y\, z\, (z\, (s\, \vec x_1\, \vec y\, z\, (r\, \vec x_0\, \vec y\, \purple z) ) )\, (r\, \vec x_0\, \vec y\, \purple z) & \text{by $\cut'$ axiom}\\
=& \leftimp' s\, t\, \vec x_1\, \vec y\, z\, (r\, \vec x_0\, \vec y\, \purple z) & \text{by $\leftimp'$ axiom}\\
=& \cut'r\, (\leftimp's\, t)\, \vec x_0\, \vec x_1\, \vec y\, z & \text{by $\cut'$ axiom}
\end{array}
\]
where the purple $\purple z$ occurrences are present only if the purple $\purple{\rho\to\sigma}$ occurrences are present in the derivations.

\subsection{$\cut'$-$\rightimp$ commutative case}

\[
\vlderivation{
	\vliin{\cut'}{}{\vec \rho_0, \vec \rho_1, \vec \sigma \seqar \rho \to \tau}{
		\vltr{s}{\vec \rho_0, \vec \sigma \seqar \sigma}{\vlhy{\ }}{\vlhy{\ }}{\vlhy{\ }}
	}{
		\vlin{\rightimp}{}{\vec \rho_1, \vec \sigma, \sigma \seqar \rho \to \tau}{
		\vltr{t}{\vec \rho_1, \vec \sigma, \sigma, \rho \seqar \tau}{\vlhy{\ }}{\vlhy{\ }}{\vlhy{\ }}
		}
	}
}
\]
is transformed to:
\[
\vlderivation{
	\vlin{\rightimp}{}{\vec \rho_0, \vec \rho_1, \vec \sigma \seqar \rho \to \tau}{
	\vliin{\cut' }{}{\vec \rho_0, \vec \rho_1, \vec \sigma, \rho \seqar \tau}{
		\vltr{s}{\vec \rho_0, \vec \sigma \seqar \sigma}{\vlhy{\ }}{\vlhy{\ }}{\vlhy{\ }}
	}{
		\vltr{t}{\vec \rho_1, \vec \sigma, \sigma, \rho \seqar \tau}{\vlhy{\ }}{\vlhy{\ }}{\vlhy{\ }}
	}
	}
}
\]

We verify that $\nT n$ proves the equality of the derivations before and after transformation as follows:
\[
\begin{array}{rll}
& \rightimp\, (\cut's\, t) \vec x_0\, \vec x_1\, \vec y\, z & \\
=& \cut' s\, t\, \vec x_0\, \vec x_1\, \vec y\, z & \text{by $\rightimp$ axiom}\\
=& t\, \vec x_1\, \vec y\, (s\, \vec x_0\, \vec y)\, z & \text{by $\cut'$ axiom}\\
=& \rightimp\, t\, \vec x_1\, \vec y\, (s\, \vec x_0\, \vec y)\, z & \text{by $\rightimp$ axiom}\\
=& \cut'\, s\, (\rightimp\, t)\, \vec x_0\, \vec x_1\, \vec y\, z & \text{by $\cut'$ axiom}
\end{array}
\]

\end{document}